\documentclass[a4paper,UKenglish,cleveref,autoref,thm-restate]{article}

\usepackage{microtype}
\usepackage{hyperref}
\usepackage{authblk}
\usepackage{orcidlink}
\usepackage{amsmath}
\usepackage{amsfonts}
\usepackage{amssymb}
\usepackage{amsthm}
\usepackage{mathtools}
\usepackage{scalerel}
\usepackage{moreverb}
\usepackage{cite}
\usepackage{braket}
\usepackage{stmaryrd}
\usepackage{cmll}
\usepackage{array}
\usepackage{algorithm2e}
\usepackage{bussproofs}
\usepackage{tikz-cd}
\usepackage{xspace}
\usepackage{tipa}
\usepackage[draft]{tikzit}
\input{blocks.tikzdefs}

\tikzstyle{box}=[fill=white, draw=black, shape=rectangle, minimum width=5mm, minimum height=5mm]
\tikzstyle{dot}=[fill=black, draw=black, shape=circle]
\tikzstyle{minidot}=[dot, inner sep=1pt]
\tikzstyle{discard}=[circuit ee IEC, thick, ground, rotate=90, scale=1.5, inner sep=-2mm, tikzit fill={rgb,255: red,176; green,176; blue,176}]
\tikzstyle{maxmix}=[circuit ee IEC, thick, ground, rotate=-90, scale=1.5, inner sep=-2mm, tikzit fill={rgb,255: red,102; green,102; blue,102}]

\tikzstyle{arr}=[->]
\tikzstyle{shade}=[-, fill={rgb,255: red,191; green,191; blue,191}, opacity=.5]
\tikzstyle{empty}=[-, fill=white]
\tikzstyle{dashes}=[-, dashed]
\tikzstyle{blue}=[-, draw=blue]
\tikzstyle{red}=[-, draw=red]
\tikzstyle{redarr}=[->, draw=red]

\allowdisplaybreaks

\newcommand{\MatRp}{\ensuremath{\textrm{Mat}[\mathbb R^+]}\xspace}
\newcommand{\MatRaff}{\ensuremath{\textrm{Mat}[\mathbb R]}\xspace}

\newcommand{\CPs}{\ensuremath{\textrm{CP}^*}\xspace}
\newcommand{\aff}{\textsf{aff}}
\newcommand{\affp}{\textsf{aff}^+}

\newcommand{\id}{\mathrm{id}}
\newcommand{\mll}{\mathrm{MLL}}
\newcommand{\mllmix}{\mathrm{MLL+Mix}}
\newcommand{\bv}{\mathrm{BV}}
\newcommand{\caus}[1]{\mathrm{Caus}\left[#1\right]}
\newcommand{\sub}[1]{\mathrm{Sub}\left(#1\right)}
\newcommand{\catc}{\mathcal{C}}
\newcommand{\ob}[1]{\mathrm{Ob}\left(#1\right)}
\newcommand{\focat}[1]{\mathrm{FO}\left(#1\right)}
\newcommand{\Var}{\mathrm{Var}}
\newcommand{\fsub}[1]{\left[#1\right]}
\newcommand{\tbool}{\texttt{t}}
\newcommand{\fbool}{\texttt{f}}

\makeatletter
\def\namedlabel#1#2{\begingroup
    #2%
    \def\@currentlabel{#2}%
    \phantomsection\label{#1}\endgroup
}
\makeatother

\newcommand\discard{
\mathbin{\text{\begin{tikzpicture}[circuit ee IEC,yscale=0.6,xscale=0.5]
\draw (0,-2ex) to (0,0) node[ground,rotate=90,xshift=.65ex] {};
\end{tikzpicture}}}
}
\newcommand\maxmix{
\mathbin{\text{\begin{tikzpicture}[circuit ee IEC,yscale=0.6,xscale=0.5]
\draw (0,2ex) to (0,0) node[ground,rotate=-90,xshift=.65ex] {};
\end{tikzpicture}}}
}
\newcommand\ndiscard{
\mathbin{\text{\begin{tikzpicture}[yscale=0.6,xscale=0.5]
\draw (0,-2ex) to (0,0) node[minidot] {};
\end{tikzpicture}}}
}
\newcommand\nmaxmix{
\mathbin{\text{\begin{tikzpicture}[yscale=0.6,xscale=0.5]
\draw (0,2ex) to (0,0) node[minidot] {};
\end{tikzpicture}}}
}

\DeclareMathOperator*{\bigseq}{\scalerel*{<}{\textstyle\sum}}

\newenvironment{scprooftree}[1]%
  {\gdef\scalefactor{#1}\begin{center}\proofSkipAmount \leavevmode}%
  {\scalebox{\scalefactor}{\DisplayProof}\proofSkipAmount \end{center} }

\theoremstyle{definition}
\newtheorem{theorem}{Theorem}[section]
\newtheorem{lemma}[theorem]{Lemma}
\newtheorem{corollary}[theorem]{Corollary}
\newtheorem{definition}[theorem]{Definition}
\newtheorem{remark}[theorem]{Remark}
\newtheorem{proposition}[theorem]{Proposition}

\newtheorem{example}[theorem]{Example}

\title{A complete logic for causal consistency}

\author[1,2]{Will Simmons \orcidlink{0000-0001-7996-6577} \thanks{will.simmons@quantinuum.com}}
\author[1]{Aleks Kissinger \orcidlink{0000-0002-6090-9684}}

\affil[1]{Department of Computer Science, University of Oxford, Oxford, UK}
\affil[2]{Quantinuum, Terrington House, 13-15 Hills Road, Cambridge, UK}

\date{}

\begin{document}

\maketitle

\begin{abstract}
The $\caus{-}$ construction takes a base category of ``raw materials'' and builds a category of higher order causal processes, that is a category whose types encode causal (a.k.a. signalling) constraints between collections of systems. Notable examples are categories of higher-order stochastic maps and higher-order quantum channels. Well-typedness in $\caus{-}$ corresponds to a composition of processes being causally consistent, in the sense that any choice of local processes of the prescribed types yields an overall process respecting causality constraints. It follows that closed processes always occur with probability 1, ruling out e.g. causal paradoxes arising from time loops. It has previously been shown that $\caus{\mathcal C}$ gives a model of MLL+MIX and BV logic, hence these logics give sufficient conditions for causal consistency, but they fail to provide a complete characterisation. In this follow-on work, we introduce graph types as a tool to examine causal structures over graphs in this model. We explore their properties, standard forms, and equivalent definitions; in particular, a process obeys all signalling constraints of the graph iff it is expressible as an affine combination of factorisations into local causal processes connected according to the edges of the graph. The properties of graph types are then used to prove completeness for causal consistency of a new causal logic that conservatively extends pomset logic. The crucial extra ingredient is a notion of distinguished atoms that correspond to first-order states, which only admit a flow of information in one direction. Using the fact that causal logic conservatively extends pomset logic, we finish by giving a physically-meaningful interpretation to a separating statement between pomset and $\bv$.
\end{abstract}

\section{Introduction}\label{sec:intro}

\textit{Black box causal reasoning} aims to reason about causal relationships - whether choices made at one point can influence behaviour at another - between the inputs and outputs of a (possibly unknown) process purely by studying its compositional properties. For example, one can formulate a family of so-called \textit{non-signalling} processes by positing the existence of a preferred \textit{discarding map} $\discard$ and requiring that a process $\Phi$ with two inputs and two outputs decomposes as follows when one or the other system is discarded:
\begin{equation}\label{eq:non-signalling}
\input{./non-signalling.tikz}  
\end{equation}
for some reduced processes $\phi_1, \phi_2$. Intuitively, this captures the notion that an agent who only has access to the left input/output pair of $\Phi$ cannot use the $\Phi$ process to send a message to another agent that only has access to the right input/output pair, and vice-versa.

A related concept to non-signalling is \textit{localisability}. A localisable process is one that decomposes as:
\begin{equation}\label{eq:localisable}
\input{./localise.tikz}
\end{equation}
where each process $\phi_i$ satisfies the so-called \textit{causality} equation:
\begin{equation}\label{eq:causality}
  \input{./causal.tikz}
\end{equation}
Such processes are sometimes referred to as \textit{Bell scenarios} in the foundations of quantum theory, and play a central role in the study of quantum non-locality. Using just the the equations above, it is straightforward to show that any localisable process is non-signalling. More generally, one can enumerate the non-signalling-type equations for a black-box process with many input/output pairs that factorises according to a generic directed acyclic graph~\cite{Kissinger2017a}.

We can formalise such results in a theory-independent manner by treating them as equations between morphisms in a symmetric monoidal category $\mathcal C$ equipped with a family of discarding maps $\discard_A : A \to I$ from any object into the monoidal unit. Notable examples are the symmetric monoidal categories of stochastic matrices, where discarding corresponds to marginalisation, and categories of quantum channels, where discarding is given by the trace functional (the quantum analogue to marginalisation).

One can make some progress in characterising the causal structure of first-order processes by decomposing them into sequential and parallel compositions of simpler processes using the symmetric monoidal structure of $\mathcal C$. However, things become more interesting when we consider \textit{higher order causal processes}. Rather than first-order boxes, which correspond to single, atomic events, higher-order processes may decompose into components that themselves have non-trivial causal structures such as interleaving of inputs and outputs over time. For example:
\begin{equation}\label{eq:boxes_with_holes}
 \input{./ho_box.tikz} = \input{./ho_scenario.tikz} = \input{./fo_scenario.tikz} 
\end{equation}
Decompositions into higher-order maps such as the middle diagram above can be formalised in terms of compact closed structure, which enables us to turn inputs into outputs and vice-versa by ``bending wires'':
\begin{equation}\label{eq:boxes_with_holes_compact}
\input{./ho_scenario.tikz} := \input{./ho_scenario_compact1.tikz}
\end{equation}
Higher-order decompositions can be used, for example, to model standard techniques from statistical causal inference such as \textit{$c$-factorisation}~\cite{jacobs2019causal} and perhaps more interestingly, allow us to study processes with unknown or even (quantum) indefinite causal structures~\cite{Oreshkov2012,Chiribella2013}.

Compact closed structure enables us to abstract over notions like ``input'' and ``output'', instead composing components over more general higher-order interfaces. However, not all such compositions yield processes that can be recreated in a laboratory. For instance, nothing in the categorical structure is preventing us from introducing time loops, which could in general introduce logical paradoxes, as witnessed e.g. by probabilistic processes that are no longer normalised. Examples include connecting together the input and output of a $\mathrm{NOT}$ gate, describing a simple instance of the Grandfather Paradox~\cite{BaumelerParadox} that ``occurs'' with probability $0$, or doing the same for an identity for the Bootstrap Paradox with probability generally exceeding $1$.
\begin{align}
\input{./paradox_grandfather.tikz} = 0 \neq \id_I \neq \id_I + \id_I = \input{./paradox_bootstrap.tikz} \label{eq:paradoxes}
\end{align}
Whilst it is possible to determine conditions for causally consistent compositions of specific kinds of higher-order processes on an ad-hoc basis \cite{Jia2018,Wilson2022}, it generally reduces down to checking concretely for the preservation of normalisation \cite{Chiribella2013} or non-signalling conditions \cite{Apadula2022}.

The $\caus{-}$ construction \cite{Kissinger2019a}, as an instance of double-glueing \cite{Hyland2003}, restricts morphisms to those that preserve normalisation, so every closed diagram describing a complete event sequence evaluates to probability $1$, and every composition preserves this sense of  determinism. It can be applied abstractly to suitable base theories like $\MatRp$ (the category of finite matrices of positive reals) or $\CPs$ (the category of finite-dimensional C*-algebras and completely positive maps) to give us deterministic theories of higher-order probabilistic or mixed quantum processes respectively. Objects in the category correspond to causal structures between a collection of local systems, such as states of $\mathbf{A < B}$ being compatible with ``$\mathbf{A}$ before $\mathbf{B}$'', and $\mathbf{A \otimes B}$ with ``$\mathbf{A}$ and $\mathbf{B}$ unrelated''. In both of the bipartite examples, there are formulations based on signalling constraints or affine combinations of factorisations, which have been shown to be equivalent \cite{Gutoski2008,Chiribella2013,Simmons2022,Cavalcanti2022}.

The first major contribution of this paper is the introduction of graph types to categories of the form $\caus{\catc}$, generalising the signalling and factoriation definitions of $\otimes$ and $<$ and their equivalence to causal structures over arbitrary acyclic graphs. These graph types admit a standard form, factoring in all relevant information from the local systems, which reduces the checking of inclusion or compatibility of graph types (with respect to forming deterministic closed scenarios under composition) to simple properties of the graphs like set-wise inclusion or acyclicity.

The question of which string diagrams preserve deterministic realisability (the ability for a processes to be reliably performed with probability $1$) is comparable to asking how causal structures should compose, and the way we choose to solve this is to provide a sound and complete logic (in the sense of the Curry-Howard-Lambek correspondence) for this compositional structure. Such a logical correspondence provides a mechanisable way to test a given string diagram's validity and highlights the structural similarities and differences between different models. The basic operators of causal categories are known to model at least Multiplicative Linear Logic with the MIX rule \cite{Kissinger2019a} and $\bv$ \cite{Simmons2022} (see Appendix \ref{sec:logic_primer} for a brief primer on the relevant logics), but these were not shown to completely describe \textit{all} valid string diagrams. Perhaps most notably, both logics gave no special status to first-order systems, despite the fact that these systems have many useful properties, seemingly arising from the fact that information can only flow in one direction (namely, from negative occurrences of first-order atoms to positive ones).

The second major output of this paper is a sound and complete logic for wiring diagrams in $\caus{\catc}$ based on a proof-net criterion that we build synthetically to match the semantics. This logic is a conservative extension of pomset logic \cite{Retore2020} (adding in directed axioms for first-order atoms) and can itself be faithfully encoded into pomset. From the characterisation theorem, we show that compatibility of causal structures for non-degenerate atomic systems is independent of the choice of local system (e.g. cardinalities of classical variables / Hilbert space dimension) and even the underlying theory $\catc$. Hence, causal consistency in classical and quantum theory are both characterised by the same abstract logical criteria.

\subsection{Related Work}


Categories for deterministic theories are much simpler for first-order processes, just requiring terminality of the unit \cite{Chiribella2010,Coecke2013,Coecke2014}. Second-order processes have also been well-studied recently, characterising them as certain classes of natural transformations \cite{Wilson2022,Wilson2022a} or viewing them as contexts or kinds of diagrams with holes \cite{Hefford2023,Earnshaw2024}. In comparison, the $\caus{-}$ construction sacrifices applications to symmetric monoidal categories that aren't compact closed to achieve a greater flexibility of compositions that escapes the presumption of a definite causal structure over the components of a diagram. Recent work by Hefford and Wilson \cite{Hefford2024} have closed this gap with profunctorial descriptions of causal structures with operators acting like $\otimes$ and $<$ from $\caus{\catc}$. The logical interaction between the operators generally limited, though richer logical content is shown to precisely correspond to decomposition theorems in the base category, e.g. whether higher-order processes decompose into networks of first-order morphisms; in this case we can view $\caus{\catc}$ as an extreme with very strong decompositions (Theorem \ref{thrm:all_graph_types_equivalent}) granting a rich logic. Where the $\caus{-}$ construction conjures an inextricable link between the origin of causal structures and preservation of normalisation conditions, their framework provides causal structures without links to normalisation, though it may require the base category to already be a first-order deterministic theory for this to be meaningful. Promising future directions to investigate the overlap involve applying the profunctor constructions to suitably enriched base categories to express non-signalling conditions and compare them to the existing causal structures on optics.

$\caus{\catc}$ strictly generalises similar inductively-defined higher-order frameworks \cite{Bisio2019,Hoffreumon2022} to permit any mathematically feasible system - this gives more robustness over time as our model does not change with the discovery of additional structure, in contrast to extending an inductively-defined framework with new operators which may invalidate existing results or require extensive reworking of proofs to maintain them.

The graphical proof system $\mathrm{GV}$ \cite{Acclavio2022} was also built to reason about compatibility between causal structures (referred to there as ``logical time''), with the primary applications aimed at event structures for concurrent programming. Despite the obvious similarities between the intents of graph types presented here and the prime directed graph constructors of $\mathrm{GV}$, there are still some differences in how they interact with the logic that we do not yet have reasonable explanations for. Why does one logic conservatively extend pomset and the other only $\mathrm{BV}$? Why do inclusions of edge sets of graphs always induce linear implications in causal logic (see Lemma \ref{lemma:transitive_closure_graph_type}) and not in $\mathrm{GV}$? The bipartite undirected graph for $\parr$ admits the interpretation as a union of directed graphs in each direction, but how should this generalise to other prime undirected graphs? One possible explanation for the logical differences relates to the equivalence between signalling constraints and (affine) factorisations in $\caus{\catc}$ which may not be reflected in other settings that e.g. require states of a tensor system to be strictly local. Having a concrete model that is similarly characterised by $\mathrm{GV}$ would help to test this hypothesis.

Instances of the $\caus{-}$ construction now stand alongside coherence spaces as additional sources of faithful categorical models of pomset logic. The connections between coherence spaces to pomset \cite{Retore2020} and their probabilistic counterparts to $\bv$ \cite{Blute2012} inspired several of the operator definitions and proofs for $\caus{\catc}$. Towards the search for a proper Curry-Howard-Lambek categorical description of pomset logic, we know any solution must capture the essence of both models. Given they both arise from restricting double-glueing by some orthogonality \cite{Hyland2003}, it may be fruitful to examine similar categories to search for a characterisation of pomset models.

\section{The $\caus{-}$ Construction}\label{sec:caus_construction}

Given a category $\catc$ describing an underlying physical theory, one can summarise the causal category $\caus{\catc}$ as a process theory of higher-order deterministic processes through which interpretations of causal structures emerge. In the interest of readers who are unfamiliar with the previous work \cite{Kissinger2019a,Simmons2022}, this section will cover some essential background material, motivate the goals of the construction, derive it, and cover the basic structural operators we will use in the remainder of the paper.

Throughout this paper, we will use string diagrams to depict morphisms in a symmetric monoidal category. A morphism $f : A \to B$ is depicted as a box with zero or more wires coming into the bottom, corresponding to inputs, and zero or more wires coming out the top for outputs. Composition, which should be read from bottom-to-top, is depicted as plugging boxes together, whereas the monoidal product $\otimes$ is depicted as placing boxes side-by-side. For the monoidal unit $I$, we refer to morphisms $\rho : I \to A$ as \textit{states}, morphisms $\pi : A \to I$ as \textit{effects}, and morphisms $\lambda : I \to I$ as \textit{scalars}.

In order to represent higher-order processes, we will make use of the structure of compact closed categories. A symmetric monoidal category is \textit{compact closed} when each object $A$ has a dual $A^*$, and ``cup'' and ``cap'' morphisms
\begin{align}
\eta_A &:= \input{./compact_cup.tikz} &
\epsilon_A &:= \input{./compact_cap.tikz}
\end{align}
that satisfy the following ``yanking'' equations:
\begin{align}
\input{./snake_equation.tikz} &= \begin{tikzpicture}
	\begin{pgfonlayer}{nodelayer}
		\node [style=none, label={[label distance=-3pt]90:$A$}] (5) at (0, 1.5) {};
		\node [style=none, label={[label distance=-3pt]-90:$A$}] (6) at (0, -1.5) {};
	\end{pgfonlayer}
	\begin{pgfonlayer}{edgelayer}
		\draw (5.center) to (6.center);
	\end{pgfonlayer}
\end{tikzpicture}
 &
\input{./snake_equation_dual.tikz} &= \begin{tikzpicture}
	\begin{pgfonlayer}{nodelayer}
		\node [style=none, label={[label distance=-3pt]90:$A^*$}] (5) at (0, 1.5) {};
		\node [style=none, label={[label distance=-3pt]-90:$A^*$}] (6) at (0, -1.5) {};
	\end{pgfonlayer}
	\begin{pgfonlayer}{edgelayer}
		\draw (5.center) to (6.center);
	\end{pgfonlayer}
\end{tikzpicture}

\end{align}

Drawing these morphisms as curved identity wires is convenient for demonstrating how they permit any morphism $f$ on one side to slide over to the other much like naturality of an identity, up to the transposition into $f^*$. This specialises to mapping between states and effects, such as transposing the discard map $\discard$ to give a \textit{uniform state} $\maxmix$ for any object.

\begin{align}
\input{./extranatural_cup.tikz} &= \input{./extranatural_cup1.tikz} &
\input{./transpose_def1.tikz} &:= \input{./transpose_def.tikz} \label{eq:transpose_def} &
\begin{tikzpicture}
	\begin{pgfonlayer}{nodelayer}
		\node [style=maxmix, label={[label distance=-3pt]135:$A$}] (0) at (0, 0) {};
		\node [style=none] (3) at (0, 1) {};
	\end{pgfonlayer}
	\begin{pgfonlayer}{edgelayer}
		\draw (0) to (3.center);
	\end{pgfonlayer}
\end{tikzpicture}
 &:= \input{./maxmix_def.tikz}
\end{align}

When sliding multiple morphisms in this way, their order reverses (i.e. transposition is a contravariant functor $\left( - \right)^* : \catc^{op} \to \catc$). This allows us to dissociate from the idea of a canonical direction of time in our diagrams, or the distinction of input and output roles. One can formalise this kind of time-reversing sliding behaviour in the form of extranatural transformations, generalising the time-directed sense of naturality.

\begin{definition}
Given functors $F : \mathcal{A} \times \mathcal{B} \times \mathcal{B}^{op} \to \mathcal{D}$ and $G : \mathcal{A} \times \mathcal{C} \times \mathcal{C}^{op} \to \mathcal{D}$, a family of morphisms
\begin{equation}
\left\{ \alpha_{A,B,C} : F(A, B, B) \to G(A, C, C) \right\}_{A \in \ob{\mathcal{A}}, B \in \ob{\mathcal{B}}, C \in \ob{\mathcal{C}}}
\end{equation}
is \textit{natural} in $A$ and \textit{extranatural} in $B$ and $C$, if for all $f \in \mathcal{A}(A, A'), g : \mathcal{B}(B, B'), h : \mathcal{C}(C, C')$:
\begin{align}
F(f, \id, \id) \fatsemi \alpha_{A', B, C} &= \alpha_{A, B, C} \fatsemi G(f, \id, \id) \\
F(\id, g, \id) \fatsemi \alpha_{A, B', C} &= F(\id, \id, g) \fatsemi \alpha_{A, B, C} \\
\alpha_{A, B, C} \fatsemi G(\id, h, \id) &= \alpha_{A, B, C'} \fatsemi G(\id, \id, h)
\end{align}
\end{definition}

The flexibility of string diagrams is especially important for ``boxes with holes'' that can now be encoded as genuine morphisms in our category as in Equation \ref{eq:boxes_with_holes_compact}. In the study of higher-order quantum processes, the induced isomorphism between states and transformations (dubbed the Choi-Jamio\l{}kowski isomorphism) is employed to reduce channels and higher-order processes down to simple states which are simpler mathematical objects to work with.

String diagrams additionally grant complete freedom of composition - instead of restricting ourselves to parallel and sequential composition of channels and functional application of higher-order processes, we can arbitrarily connect any two subsystems with matching interfaces. However, permitting arbitrary connections can easily give rise to unphysical process descriptions, encoding paradoxical situations like the Grandfather Paradox or the Bootstrap Paradox which require sending information to the past. It is clear that not every string diagram is going to preserve properties relating to normalisation or information flow (see Equation \ref{eq:paradoxes}). The $\caus{-}$ construction yields theories containing precisely the valid processes in this sense, and hence its internal structure should describe the most flexible general composition scheme that preserves physical realisability of processes.

This brings us to the goals of the $\caus{-}$ construction:
\begin{enumerate}
\item\label{enum:goal_theory_independence} Theory-independence/abstraction - being able to reuse as much as possible between the study of different theories really helps to understand what properties are truly fundamental versus determined by the specifics of a given theory. Physical motivations drive us to at least want versions for both classical probability theory and (finite-dimensional) quantum theory.
\item\label{enum:goal_determinism} Deterministic processes - in this context, we use the term \textit{deterministic} to mean that the process can be reliably performed with probability 1, without any requirement to be free of probabilistic behaviour (for example, a coin flip may generate a random outcome but is still deterministic in this sense since coarse-graining over all outcomes gives a combined probability of 1). This provides an operational foundation with which we can describe any scenario we can prepare in a laboratory. The term \textit{causal} is often used synonymously with this definition elsewhere in the literature, and we will use both interchangeably.
\item\label{enum:goal_free_composition} Freely compositional - we wish to impose minimal a priori constraints on how processes may be composed. In particular, this means we should support as many string diagrams as possible without violating determinism. Having open diagrams with composition only limited by the matching of interfaces means we can make reasonable abstractions such as grouping the operations within a protocol a single process for each participant, or combining complex, interleaved interactions between physical processes occurring over non-convex regions of spacetime into a single interaction.
\end{enumerate}

For Goal \ref{enum:goal_theory_independence}, we start with a base category $\catc$ of unnormalised processes within which we will work (expecting a forgetful functor $\caus{\catc} \to \catc$). We require $\catc$ to be compact-closed to support the string diagrams imposed by Goal \ref{enum:goal_free_composition}. Setting the standard examples of $\MatRp$ to investigate finite-dimensional classical probability theory and $\CPs$ for quantum theory, Definition \ref{def:apc} abstractly presents the properties of these categories we rely on for this work. These properties are satisfied by other settings such as $\MatRaff$ for affine/quasi-probabilistic maps, but notably excludes settings with finite scalars like $\mathrm{Rel}$, infinite-dimensional systems like $\mathrm{Hilb}$, or settings lacking local tomography like real quantum theory.

\begin{definition}\label{def:apc}\cite[Definition 8]{Simmons2022}
Let $\catc$ be a compact closed category with products. Such a category always has biproducts and additive enrichment~\cite{Houston2008}: each homset is a commutative monoid, writing summation of $f, g \in \catc(A, B)$ as $f + g \in \catc(A, B)$ and zero morphisms as $0_{A,B} \in \catc(A, B)$. $\catc$ is an \textit{additive precausal category} if:
\begin{enumerate}
\item[\namedlabel{apc:discard}{APC1}.] $\catc$ has a discarding process $\discard_A \in \catc(A, I)$ for every system $A$, compatible with the monoidal and biproduct structures as below;
\begin{align}
\discard_{A \otimes B} &= \discard_A \otimes \discard_B \label{eq:discard_monoidal} \\
\discard_I &= \id_I \label{eq:discard_unit} \\
\discard_{A \oplus B} &= \left[ \discard_A, \discard_B \right] \label{eq:discard_biproduct}
\end{align}
\item[\namedlabel{apc:dimension}{APC2}.] The dimension $d_A$ is invertible for all non-zero $A$;
\begin{equation}
d_A := \begin{tikzpicture}
\begin{pgfonlayer}{nodelayer}
\node [style=discard, label={[label distance=-3pt]-135:$A$}] (0) at (0, 1) {};
\node [style=maxmix, label={[label distance=-3pt]135:$A$}] (1) at (0, -1) {};
\end{pgfonlayer}
\begin{pgfonlayer}{edgelayer}
\draw (1) to (0);
\end{pgfonlayer}
\end{tikzpicture}

\end{equation}
\item[\namedlabel{apc:basis}{APC3}.] Each object $A \in \ob{\catc}$ has a finite causal basis: a minimal finite set $\{\rho_i\}_i \subseteq \catc\left( I, A \right)$ of causal states which are sufficient to distinguish morphisms;
\begin{align}
\forall i &.\  \input{./causal_state.tikz} = \id_I \\
\forall B .\  \forall f, g \in \catc(A, B) &.\  \left( \forall i .\ \input{./basis_f.tikz} = \input{./basis_g.tikz} \right) \Longrightarrow \input{./basis_f1.tikz} = \input{./basis_g1.tikz}
\end{align}
\item[\namedlabel{apc:scalars}{APC4}.] Addition of scalars is \textit{cancellative} ($\forall x,y,z.\ x + z = y + z \Rightarrow x = y$), \textit{totally pre-ordered} ($\forall x,y .\exists z .\ x=y+z \vee y=x+z$), and all non-zero scalars have a multiplicative inverse.
\item[\namedlabel{apc:complement}{APC5}.] All effects have a complement with respect to discarding: for any $\pi \in \catc(A,I)$, there exists some $\pi' \in \catc(A,I)$ and scalar $\lambda$ such that $\pi + \pi' = \lambda \cdot \discard_A$.
\end{enumerate}
\end{definition}

From these properties, we can already utilise a lot of standard techniques of linear algebra. In particular, even in settings with positivity constraints such as those in $\MatRp$ and $\CPs$ we will benefit from using a larger category without such a constraint (e.g. containing all linear maps). Any additive precausal category $\catc$ has a faithful embedding $\fsub{-} : \catc \to \sub{\catc}$ into the free subtractive closure which preserves almost all of the relevant categorical structure \cite[Propositions 38-42]{Simmons2022}. We will freely make use of this in our proofs, distinguishing between equality $=$ of morphisms in $\catc$ and equivalence $\sim$ of morphisms in $\sub{\catc}$ to make it clear which category we are working in, and with faithfulness of the embedding ($\fsub{f} \sim \fsub{g} \Rightarrow f = g$) allowing us to always draw conclusions in $\catc$.



The following are more useful variants of \ref{apc:basis} and \ref{apc:complement} which follow from the above assumptions.

\begin{proposition}\label{prop:dual_basis}\cite[Lemma 11]{Simmons2022}
Given any set of morphisms in $\catc(A, B)$ that are linearly independent in $\sub{\catc}$, they can be extended to a basis in $\catc$ (in the sense of \ref{apc:basis}) with a dual basis in $\sub{\catc}$.
\end{proposition}

\begin{proposition}\label{prop:complement}\cite[Proposition 12]{Simmons2022}
For any $f : A \to B$ in an additive precausal category, there exists $f' : A \to B$ and scalar $\lambda$ such that:
\begin{equation}\label{eq:complement}
\input{./basis_f1.tikz} + \input{./channel_fp.tikz} = \lambda \input{./noisy_channel.tikz}
\end{equation}
\end{proposition}


Let's return to the remaining goal (\ref{enum:goal_determinism}) and formalise how we will treat determinism. 

\begin{definition}
A \textit{closed scenario} is a process description with no open systems (i.e. a scalar), and a \textit{context} for a morphism $f$ describes a diagram into which $f$ can be placed to form a closed scenario. For any state $\rho \in \catc(I, A)$ and effect $\pi \in \catc(A, I)$, we can interpret the scalar $\rho \fatsemi \pi \in \catc(I, I)$ as the (abstract) probability with which the closed scenario represented by this diagram can be realised, in the vain of the Born rule for quantum theory. When the scalar is $\id_I$, we say that the scenario described by the composite diagram is \textit{deterministic}, or that $\rho$ is \textit{deterministically realisable} in the context of $\pi$ (and vice versa).
\end{definition}


Double-glueing along a hom-functor \cite{Hyland2003} provides an all-purpose construction for enforcing generic normalisation conditions such as determinism.

\begin{definition}\label{def:double_glueing}
Double-glueing along a hom-functor (simplified): given a symmetric monoidal closed $\catc$, the glued category $\mathbf{G}_I(\catc)$ has objects of the form $\mathbf{A} = \left( A \in \ob{\catc}, c_{\mathbf{A}} \subseteq \catc(I, A), d_{\mathbf{A}} \subseteq \catc(A, I) \right)$ and morphisms $f : \mathbf{A} \to \mathbf{B}$ are some $f \in \catc(A, B)$ which preserve both the chosen states and effects:
\begin{align}
\forall \rho \in c_{\mathbf{A}} &. \input{./preserve_state.tikz} \in c_{\mathbf{B}} \label{eq:glueing_forwards} \\
\forall \pi \in d_{\mathbf{B}} &. \input{./preserve_effect.tikz} \in d_{\mathbf{A}} \label{eq:glueing_backwards}
\end{align}
\end{definition}

$\mathbf{G}_I(\catc)$ inherits a lot of the categorical structure from $\catc$ - for each of symmetric monoidal closed structures, $*$-autonomy, finite products, and finite coproducts, if they exist in $\catc$, then they exist in $\mathbf{G}_I(\catc)$ and are preserved by the forgetful functor $\mathbf{G}_I(\catc) \to \catc$ \cite{Hyland2003}. We will see this in more detail across Sections~\ref{sec:monoidal_unit}-\ref{sec:products_coproducts} when reviewing how these structures lift for the $\caus{-}$ construction.

There are multiple intuitions one can use to view the objects of $\mathbf{G}_I(\catc)$ which we will freely appeal to throughout to explain concepts:
\begin{itemize}
\item The direct mathematical interpretation views objects $\mathbf{A}$ as the pair of a set of states $c_{\mathbf{A}}$ and effects $d_{\mathbf{A}}$ which we deem normalised against the criterion of $\mathbf{A}$.
\item Taking an operational perspective, an object describes a black-box system which an external agent may interact with. The black-box has an internal state taken from $c_\mathbf{A}$ which the agent can update by appending morphisms. Monoidal products like $\mathbf{A \otimes B}$ represent composite or multi-partite systems. An agent who only has access to $\mathbf{A}$ can now only apply morphisms of the form $f \otimes \id_{\mathbf{B}}$ as they must commute with every action another agent could take on $\mathbf{B}$. The effects $d_\mathbf{A}$ describe the possible complete sequences of actions the agent can take which yield a closed scenario. In general there may be multiple distinct marginal states at $\mathbf{B}$ induced by the different effects, from which we can define non-signalling conditions and new objects whose states are admissibile with respect to some causal structure. The freedom of picking any state space $c_{\mathbf{A}}$ and effect space $d_{\mathbf{B}}$ allows us to offer agents an artificially restricted set of actions as needed - for example, in $\mathbf{G}_I(\CPs)$ we will not only have a qubit object that supports all (mixed) states, but also an object that restricts states to specific planes or axes of the Bloch sphere or, dually, allows an agent to postselect within a given plane. In Section \ref{sec:products_coproducts}, we will introduce a way to encode observations via binary tests, representing the agent's knowledge as an auxiliary object they also have access to on which we can condition future actions (e.g. to internalise the agent's decision making). With this in mind, the effects $d_{\mathbf{A}}$ involve marginalising away any such knowledge as the other agent on $\mathbf{B}$ would not have access to it.
\item We may view an object as a choice of protocol or interface. We can view a state and effect pair as implementations of two parties participating in the protocol to completion, one as the party in focus and the other as the environment or external party. The analog of causal structure here is the format and ordering requirements of the messages that are sent between the parties as part of the protocol.
\end{itemize}

Applying double-glueing to an additive precausal category is close to fulfilling our goals, but we need to restrict the objects such that the conditions of Equations \ref{eq:glueing_forwards}-\ref{eq:glueing_backwards} give us precisely the deterministically realisable processes. If we consider fixing some set $c \subseteq \catc(I, A)$, we can construct the maximal set of effects that are deterministically realisable in the context of every state in $c$.

\begin{definition}\label{def:dual_set}\cite[Definition 4.1]{Kissinger2019a}
Given $c \subseteq \catc(I, A)$, the \textit{dual set} is 
\begin{equation}
c^* := \left\{ \pi \in \catc(A,I) \middle| \forall \rho \in c . \input{./dual_set.tikz} = \id_I \right\}
\end{equation}
By compact closure and transposing elements, we may choose to equivalently interpret $c$/$c^*$ as either a set of states in $\catc(I,A)$/$\catc(I,A^*)$ or a set of effects $\catc(A^*,I)$/$\catc(A,I)$. A set $c \subseteq \catc(I,A)$ is \textit{closed} if $c = c^{**}$.
\end{definition}

$\left(-\right)^*$ is an instance of a focussed orthogonality \cite{Hyland2003}, which immediately gives $\left(-\right)^{***} = \left(-\right)^*$.

For an object $\left(A, c_{\mathbf{A}}, d_{\mathbf{A}}\right) \in \mathbf{G}_I(\catc)$ to represent deterministic normalisation, we need each combination of a state $\rho \in c_{\mathbf{A}}$ and an effect $\pi \in d_{\mathbf{A}}$ to give $\rho \fatsemi \pi = \id_I$, i.e. $c_{\mathbf{A}} \subseteq d_{\mathbf{A}}^*$ and $d_{\mathbf{A}} \subseteq c_{\mathbf{A}}^*$. If we ask that such sets are maximal (i.e. the only restriction we make on morphisms is determinism), we have $c_{\mathbf{A}} = d_{\mathbf{A}}^* = c_{\mathbf{A}}^{**}$ (the tight orthogonality subcategory). In this case, it is sufficient to just provide one of the sets to identify an object.

The additional symmetry here also makes Equations \ref{eq:glueing_forwards}, \ref{eq:glueing_backwards}, and \ref{eq:deterministic_morphism} equivalent, further driving the idea that deterministic morphisms are those that can be applied with probability $1$ in any closed context.
\begin{equation}\label{eq:deterministic_morphism}
\forall \rho \in c_{\mathbf{A}}, \pi \in d_{\mathbf{B}} \left( = c_{\mathbf{B}}^* \right) .\ \input{./preserve_scalar.tikz} = \id_I
\end{equation}

The only remaining edge case that doesn't necessarily lead to determinism is when either $c_{\mathbf{A}}$ or $c_{\mathbf{A}}^*$ is empty (the other will be the full homset $\catc(I, A)$ or $\catc(A, I)$ by duality, notably including the zero morphism which is not deterministic in any context). The final step to get to the $\caus{-}$ construction is to ask that the states are at least deterministic against some uniform effect (and, dually, the existence of a uniform state).

\begin{definition}\label{def:flat}\cite[Definition 4.2]{Kissinger2019a}
A set $c \subseteq \catc(I,A)$ is \textit{flat} if there exist invertible scalars $\mu, \theta \in \catc(I,I)$ such that $\mu \cdot \maxmix_A \in c$ and $\theta \cdot \discard_A \in c^*$. When these scalars exist, they are necessarily unique.
\end{definition}

\begin{remark}\label{remark:weak_flatness}
Previous work \cite{Simmons2022} had weakened the flatness condition to permit empty sets for the purpose of studying additive units (an initial and terminal object) in $\caus{\catc}$. We will choose to use the original strict version of flatness for this paper for simplicity, though most of the results will follow vacuously for the additional weak cases.
\end{remark}

We will freely make use of the following result throughout the paper to view any flat, closed set $c^{**}$ as the affine combinations over $c$ (when $\catc$ does not natively support negatives, take affine combinations in $\sub{\catc}$ and then restrict to the homsets of $\catc$, e.g. intersecting with the positive cone for $\MatRp$ or $\CPs$).

\begin{definition}\label{def:affine-closure}\cite[Definition 16]{Simmons2022}
For a set of states $c \subseteq \catc(I, A)$, we define sets $\aff(c) \subseteq \sub{\catc}(I, A)$ and $\affp(c) \subseteq \catc(I,A)$ as follows, for $K := \sub{\catc}(I,I)$:
\begin{align*}
\aff(c) & := \left\{ \rho \, \middle| \, \exists \{\rho_i\}_i \subseteq \catc(I, A), \{\lambda_i\}_i \subseteq K . \ \sum_i \lambda_i \sim \id_I, \rho \sim \sum_i \lambda_i \cdot \fsub{\rho_i} \right\} \\
\affp(c) & := \left\{ \rho \, \middle| \, \exists \rho' \in \aff(c).\ \rho' \sim \fsub{\rho} \right\}
\end{align*}
\end{definition}

\begin{proposition}\label{prop:affine}\cite[Theorem 17]{Simmons2022}
Given any flat set $c \subseteq \catc(I,A)$ for a non-zero $A$, $c^{**} = \affp(c)$.
\end{proposition}


Combining the steps so far, we can define $\caus{\catc}$ as the full subcategory of $\mathbf{G}_I(\catc)$ on objects $\left( A, c_{\mathbf{A}}, d_{\mathbf{A}} \right)$ with $c_{\mathbf{A}} = d_{\mathbf{A}}^* = c_{\mathbf{A}}^{**}$ flat. For simplicitly, we will use the following explicit definition:

\begin{definition}\label{def:caus_category}\cite[Definition 4.3]{Kissinger2019a}
Given an additive precausal category $\catc$, the causal category $\caus{\catc}$ has as objects pairs $\mathbf{A} := \left( A \in \ob{\catc}, c_\mathbf{A} \subseteq \catc(I,A) \right)$ where $c_\mathbf{A}$ is closed and flat. A morphism $f : \mathbf{A} \to \mathbf{B}$ is a morphism $f : A \to B$ in $\catc$ such that $\forall \rho \in c_\mathbf{A} . \rho \fatsemi f \in c_\mathbf{B}$.

This has an obvious forgetful functor $\mathcal{U} : \caus{\catc} \to \catc$ which drops the sets from the objects and is identity on morphisms. We may use $\rho : \mathbf{A}$ to say $\rho \in c_\mathbf{A}$, and we introduce notation $\mu_{\mathbf{A}}, \theta_{\mathbf{A}} \in \catc(I, I)$, $\nmaxmix_{\mathbf{A}} := \mu_{\mathbf{A}} \cdot \maxmix_A \in c_\mathbf{A}$, $\ndiscard_{\mathbf{A}} := \theta_{\mathbf{A}} \cdot \discard_A \in c_\mathbf{A}^*$ for the normalisation constants and normalised uniform state and effect required by flatness.
\end{definition}

\begin{remark}\label{remark:scalar_multiples}
For any object $\mathbf{A} \in \ob{\caus{\catc}}$ and non-zero scalar $\lambda \in \catc(I, I)$, we also have the isomorphic object $\lambda \mathbf{A} \left( \cong \mathbf{A} \right) \in \ob{\caus{\catc}}$ where we scale the states $c_{\lambda \mathbf{A}} := \left\{ \lambda \cdot \rho \middle| \rho \in c_{\mathbf{A}} \right\}$ (and inversely scale effects $c_{\lambda \mathbf{A}}^*$). The scalar is meaningless for physical interpretations as it will commute with all actions or observations made by an agent, and will always meet its inverse in a closed scenario.
\end{remark}

For the remainder of this section, we will summarise the structural operators for $\caus{\catc}$ discovered across previous work~\cite{Kissinger2019a,Simmons2022} which make it a $\bv$-category (generalisation of a $*$-autonomous category with the sequence operator of $\bv$) with finite products and coproducts. Many of these are the standard constructions for preserving categorical structure of the base category by double glueing with an orthogonality \cite{Hyland2003}. For brevity, we will omit proofs of preserving closure and flatness or any equivalences of definitions, each of which is either straightforward using duality or shown in~\cite{Kissinger2019a,Simmons2022}.

\subsection{Monoidal Unit}\label{sec:monoidal_unit}

The monoidal unit of $\caus{\catc}$ is lifted from the unit $I$ of $\catc$.

\begin{equation}\label{eq:monoidal_unit_def}
\mathbf{I} := \left( I, \left\{ \id_I \right\} \right)
\end{equation}

$\mathbf{I}$ has a single valid state, which we typically view as an empty diagram. This means that for any $\mathbf{A}$, its states are precisely the homset $\caus{\catc}\left(\mathbf{I}, \mathbf{A} \right) = c_{\mathbf{A}}$ and effects are $\caus{\catc}\left( \mathbf{A}, \mathbf{I} \right) = c_{\mathbf{A}}^*$. Again, this reinforces the concept of deterministic normalisation since the composition of any state $\rho : \mathbf{I} \to \mathbf{A}$ and effect $\pi : \mathbf{A} \to \mathbf{I}$ must be $\id_I$, the unit (abstract) probability.


\subsection{Dual Objects}

Compact closure of $\catc$ allows us to transpose any morphism $f \in \catc(A, B)$ to $f^* \in \catc(B^*, A^*)$ as defined in Equation \ref{eq:transpose_def}, giving us a time-symmetric view of our processes. Transposition exchanges notions of states and effects and formally expresses a duality between viewing morphisms as transforming states or transforming effects. This duality is lifted to $\caus{\catc}$ as a similar (strong monoidal) functor $(-)^* : \caus{\catc}^{op} \to \caus{\catc}$.

\begin{equation}\label{eq:dual_object_def}
\mathbf{A^*} := \left( A^*, c_{\mathbf{A}}^* \right)
\end{equation}

When viewing an object as describing a set of processes or implementations of an interface, the dual object describes the set of effects/contexts that can be applied to them deterministically or ways to consume or interact with a black-box presenting that interface.

\subsection{First-order Objects}

We say that an object $\mathbf{A}$ of $\caus{\catc}$ is first-order when there is a unique causal effect $\left| c_{\mathbf{A}}^* \right| = 1$. By the assumption of flatness, the unique effect must be $\ndiscard_{\mathbf{A}}$, i.e. $\discard_A$ up to some invertible scalar - by Remark \ref{remark:scalar_multiples}, this gives a unique first-order object for each $A \in \ob{\catc}$ up to isomorphism. We denote a first-order object on $A$ as $\mathbf{A^1}$, with the following canonical construction:

\begin{equation}\label{eq:fo_object_def}
\mathbf{A^1} := \left(A, \left\{\discard_A\right\}^*\right)
\end{equation}

First-order systems $\mathbf{A^1}$ and their duals $\mathbf{(A^1)^*}$ are the simplest atomic components of a causal scenario, capturing degenerate interfaces that respectively represent outputs with (in general) multiple distinguishable states but only one choice of effect $\discard_A$ that an external agent can apply, or inputs where an external agent can choose how to interact with the system but they cannot receive information from it. In most theories, states of a first-order system will correspond to basic data (distributions over a finite set) or descriptions of a physical system (density matrices of a finite-dimensional Hilbert space) at a single point in time.

If a system is degenerate in both senses (i.e. it has a single state and effect), it is isomorphic to the trivial system $\mathbf{I}$.

\begin{remark}\label{remark:fo_subcategory}
We denote the full subcategory of first-order objects as $\focat{\caus{\catc}}$. Within this subcategory, $\mathbf{I}$ is terminal ($\caus{\catc}\left( \mathbf{A^1}, \mathbf{I} \right) = \left\{ \discard_A \right\}$), following the usual treatment of causality as the existence of a unique effect for each system.
\end{remark}

\subsection{Tensor, Sequence, and Par}\label{sec:monoidal_products}

In $\catc$, $A \otimes B$ represents parallel composition of two systems. There are multiple ways of lifting this to $\caus{\catc}$ depending on what kinds of states and effects we permit on the joint system.

\begin{align}
\mathbf{A \otimes B} := & \left( A \otimes B, \left\{ \input{./tensor_def.tikz} \middle| \rho_A \in c_{\mathbf{A}}, \rho_B \in c_{\mathbf{B}} \right\}^{**} \right) \label{eq:tensor_object_separable_def} \\
= & \left( A \otimes B, \left\{ h \middle| \begin{array}{rl} & \forall \pi_B \in c_{\mathbf{B}}^* . \input{./no_sig_rl_eff.tikz} = \input{./no_sig_rl_eff2.tikz} \in c_{\mathbf{A}} \\
\wedge & \forall \pi_A \in c_{\mathbf{A}}^* . \input{./no_sig_lr_eff.tikz} = \input{./no_sig_lr_eff2.tikz} \in c_{\mathbf{B}} \end{array} \right\} \right) \label{eq:tensor_object_non_signalling_def}
\end{align}

$\mathbf{A \otimes B}$ is the smallest closed space containing all separable states over the joint system (Equation \ref{eq:tensor_object_separable_def}). Even though this includes some states that are not separable or even representable as local processes with some shared history (as in Equation \ref{eq:localisable}), the separable ones still span the entire space under affine combination. This generalises similar equivalences for first-order channels~\cite{Gutoski2008,Chiribella2013,Cavalcanti2022} to arbitrary affine-closed systems.

States of this kind will present an interface that acts like $\mathbf{A}$ in parallel with another that acts like $\mathbf{B}$ such that there is no signalling between them - the marginal state at $\mathbf{A}$ is independent of any local context applied at the $\mathbf{B}$ side and vice-versa (Equation \ref{eq:tensor_object_non_signalling_def}). For example, if $\mathbf{A}$ and $\mathbf{B}$ each describe implementations of some multi-round communication protocol, none of the inputs on the $\mathbf{A}$ side can cause any observable difference to the outputs on the $\mathbf{B}$ side. There may still be correlation or even entanglement between the two sides so long as this does not enable communication between local agents on each side.

We can think of the $\mathbf{A}$ and $\mathbf{B}$ interfaces as being simultaneously accessible, so contexts can freely pass information between them. For example, if $\mathbf{B} = \mathbf{A^*}$, then a separable state is precisely a pair of a state and effect for $\mathbf{A}$ which we can plug together since the cap $\epsilon_A \in \catc \left( A \otimes A^*, I \right)$ of the compact structure is causal $\epsilon_A \in \caus{\catc} \left( \mathbf{A \otimes A^*}, \mathbf{I} \right)$. We extend this intuition to contracting non-separable states by treating them as affine combinations of separable states and applying linearity. In the multi-round communication protocol picture, $\epsilon_A$ would play a copycat, forwarding each output from the $\mathbf{A}$ side to the corresponding input on the $\mathbf{A^*}$ side and vice versa.

Dropping the non-signalling conditions, we define another operator on objects giving the full space of bipartite processes that act locally like $\mathbf{A}$ and $\mathbf{B}$ but can permit arbitrary signalling of information between them, yielding objects with more states but fewer effects.

\begin{align}
\mathbf{A \parr B} := & \left( A \otimes B, \left\{ h \middle| \forall \pi_A \in c_{\mathbf{A}}^*, \pi_B \in c_{\mathbf{B}}^* . \input{./parr_local_eff.tikz} = \id_I \right\} \right) \label{eq:parr_object_local_effect_def} \\
= & \left( A \otimes B, \left\{ h \middle| \forall \pi_A \in c_{\mathbf{A}}^* . \input{./no_sig_lr_eff.tikz} \in c_{\mathbf{B}} \right\} \right) \label{eq:parr_object_left_control_def} \\
= & \left( A \otimes B, \left\{ h \middle| \forall \pi_B \in c_{\mathbf{B}}^* . \input{./no_sig_rl_eff.tikz} \in c_{\mathbf{A}} \right\} \right) \label{eq:parr_object_right_control_def}
\end{align}

These three definitions are straightforwardly equivalent by Definition \ref{def:dual_set}.

Whilst the local marginals show that this acts like an $\mathbf{A}$ and $\mathbf{B}$ in parallel, causal contexts can only act locally on them. This allows us to model distributed protocols or physical systems separated by a substantial distance in space or time, where there may not be a way to present the $\mathbf{A}$ and $\mathbf{B}$ interfaces simultaneously for a localised agent to interact with both. This is exhibited in the de Morgan duality between $\otimes$ and $\parr$, since the contexts for $\mathbf{A \parr B}$ must be non-signalling, whereas contexts for $\mathbf{A \otimes B}$ can feature arbitrary signalling:

\begin{equation}\label{eq:tensor_parr_de_morgan}
\mathbf{\left( A \otimes B \right)^*} = \mathbf{A^* \parr B^*}
\end{equation}

We can understand $\parr$ a little better by looking at the internal hom $\mathbf{A \multimap B}$ which represents the space of transformations from $\mathbf{A}$ to $\mathbf{B}$.

\begin{align}
\mathbf{A \multimap B} := & \mathbf{A^* \parr B} \\
= & \left( A^* \otimes B, \left\{ h \middle| \forall \rho_A \in c_{\mathbf{A}} . \input{./internal_hom_on_states1.tikz} \in c_{\mathbf{B}} \right\} \right) \label{eq:internal_hom_left_control_def} \\
= & \left( A^* \otimes B, \left\{ h \middle| \forall \pi_B \in c_{\mathbf{B}}^* . \input{./internal_hom_on_effects1.tikz} \in c_{\mathbf{A}}^* \right\} \right) \label{eq:internal_hom_right_control_def}
\end{align}

Each state of $\mathbf{A \multimap B}$ is an encoding of a transformation which we can view in either direction: composing it with any state of $\mathbf{A}$ gives a state of $\mathbf{B}$, or dually applying any effect $\mathbf{B^*}$ yields an effect $\mathbf{A^*}$. $\multimap$ makes $\caus{\catc}$ monoidal closed wrt $\otimes$, with encodings and evaluations given by composing with the cup $\eta_A : \mathbf{I} \to \mathbf{A^* \parr A}$ and cap $\epsilon_A : \mathbf{A \otimes A^*} \to \mathbf{I}$. For example, if $\catc = \CPs$, $\eta_A$ and $\epsilon_A$ are a maximally entangled state and effect (e.g. the Bell state/effect for qubits) capturing the Choi-Jamio\l{}kovski isomorphism from quantum theory.

Finally, if we drop just one of the non-signalling conditions from $\otimes$, we obtain a non-commutative monoidal product $<$ (and similarly $>$ by symmetry).

\begin{align}
\mathbf{A < B} := & \left( A \otimes B, \left\{ h \middle| \forall \pi_B \in c_{\mathbf{B}}^* . \input{./no_sig_rl_eff.tikz} = \input{./no_sig_rl_eff2.tikz} \in c_{\mathbf{A}} \right\} \right) \label{eq:seq_object_one_way_def} \\
= & \left( A \otimes B, \left\{ \input{./semi_local_HO_named.tikz} \middle| \begin{array}{l} h_{AZ} \in c_{\mathbf{A \parr Z^1}}, \\ h_{ZB} \in c_{\mathbf{\left(Z^1\right)^* \parr B}} \end{array} \right\}^{**} \right) \label{eq:seq_object_semi_localisable_def} \\
= & \left( A \otimes B, \left\{ h \middle| \begin{array}{l} \exists \mathcal{I}, \{f_i\}_{i \in \mathcal{I}} \subseteq \sub{\catc}(I, A \otimes B), \\ \{g_i\}_{i \in \mathcal{I}} \subseteq c_\mathbf{B}, f \in c_\mathbf{A} . \\ \fsub{h} \sim \sum_{i \in \mathcal{I}} f_i \otimes \fsub{g_i} \wedge \fsub{f} \sim \sum_{i \in \mathcal{I}} f_i \end{array} \right\} \right) \label{eq:seq_object_local_sum_def}
\end{align}

Just like $\otimes$, Equation \ref{eq:seq_object_semi_localisable_def} provides an equivalent definition based on decomposition into an affine combination of simple morphisms. In this case, one-way signalling processes are spanned by \textit{semi-localisable processes} - those which factorise into local processes connected by a side-channel passing state in just one direction. If $\mathbf{A}$ and $\mathbf{B}$ are spaces of first-order quantum channels, we know that semi-localisable processes are already closed since every one-way signalling channel is semi-localisable using purifications \cite{Chiribella2010}, though it is not known if we can more generally drop the closure on Equation \ref{eq:seq_object_semi_localisable_def} for any higher-order $\mathbf{A}$, $\mathbf{B}$ and any base category $\catc$. Equation \ref{eq:seq_object_local_sum_def} is provided for completeness, showing that one-way signalling processes can be written as a sum of causal states of $\mathbf{B}$ and (possibly negative) terms which combine to a causal state of $\mathbf{A}$, though we will not make use of this form in this paper outside of proofs.

We can think of a process of $\mathbf{A < B}$ as something that presents an $\mathbf{A}$ interface, and then presents a $\mathbf{B}$ interface at some later time, possibly requiring the $\mathbf{A}$ protocol to be completed beforehand. Our ways to interact with this are therefore to first completely consume the $\mathbf{A}$, before later consuming the $\mathbf{B}$.

\begin{equation}\label{eq:seq_de_morgan}
\mathbf{\left( A < B \right)^*} = \mathbf{A^* < B^*}
\end{equation}

These constructions come equipped with a number of equations, including those standard to $\bv$ logic \cite{Simmons2022}, which incorporates linear logic with isomorphic monoidal units, as well as the following natural transformations for mixing and interchanges with the sequence operator.

\begin{align}
\mathbf{A \otimes B} &\Rightarrow \mathbf{A < B} \Rightarrow \mathbf{A \parr B} \label{eq:mixing_nat_trans} \\
\mathbf{\left( A < B \right) \otimes \left( C < D \right)} &\Rightarrow \mathbf{\left( A \otimes C \right) < \left( B \otimes D \right)} \label{eq:weak_interchange_nat_trans}
\end{align}

Additionally, we recall the special role of first-order objects giving additional known equations.

\begin{proposition}\label{prop:first_order_equations}\cite[Corollary 5.5]{Kissinger2019a}\cite[Theorem 32]{Simmons2022}
$\mathbf{A^* \parr A} = \mathbf{A^* < A} \Leftrightarrow \left| c_{\mathbf{A}}^* \right| = 1$. Consequently, the following equations hold for any first-order objects $\mathbf{A^1}, \mathbf{A'^1}$ and any object $\mathbf{B}$:
\begin{align}
\mathbf{A^1 \parr B} &= \mathbf{A^1 > B} \label{eq:fo_parr_is_one_way} \\
\mathbf{A^1 \otimes B} &= \mathbf{A^1 < B} \label{eq:fo_tensor_is_one_way} \\
\mathbf{A^1 \parr A'^1} = \mathbf{A^1 < A'^1} &= \mathbf{A^1 > A'^1} = \mathbf{A^1 \otimes A'^1} \label{eq:fo_fo_parr_is_tensor} \\
\mathbf{\left(A^1\right)^* \parr B} &= \mathbf{\left(A^1\right)^* < B} \label{eq:fod_parr_is_one_way} \\
\mathbf{\left(A^1\right)^* \otimes B} &= \mathbf{\left(A^1\right)^* > B} \label{eq:fod_tensor_is_one_way} \\
\mathbf{\left(A^1\right)^* \parr \left(A'^1\right)^*} = \mathbf{\left(A^1\right)^* < \left(A'^1\right)^*} &= \mathbf{\left(A^1\right)^* > \left(A'^1\right)^*} = \mathbf{\left(A^1\right)^* \otimes \left(A'^1\right)^*} \label{eq:fod_fod_parr_is_tensor}
\end{align}
\end{proposition}

The signalling conditions of $\otimes$ and $<$ may not look exactly like those in Equation \ref{eq:non-signalling}, but when we consider local systems to be first-order channels $\mathbf{A^1 \multimap B^1}$ we can show them to be the same. First, consider what the effects $\pi \in c_{\mathbf{A}^1 \multimap \mathbf{B}^1}^*$ look like. As these are dual to the set of first-order processes $\mathbf{A}^1 \multimap \mathbf{B}^1$, we can visualise them this way:
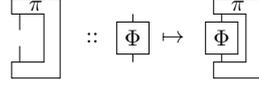
\ctikzfig{proc-effect}
Consequently, the condition in the definition~\eqref{eq:seq_object_one_way_def} applied to $(\mathbf{A}^1 \multimap \mathbf{B}^1) < (\mathbf{C}^1 \multimap \mathbf{D}^1)$ says that for any causal effect $\pi \in c_{\mathbf{C}^1 \multimap \mathbf{D}^1}$ we have:
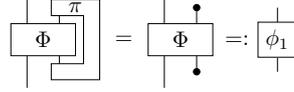
\ctikzfig{seq-effect-reduce}
However, as was shown in~\cite{Kissinger2019a}, the effects in $c_{\mathbf{C}^1 \multimap \mathbf{D}^1}^*$ are of a very restricted form: namely plugging a causal state of $\mathbf{C}^1$ into the process and discarding the result. Hence, the equation above reduces to say that for all states $\rho \in c_{\mathbf{C}^1}$, there exists $\phi_1$ such that:
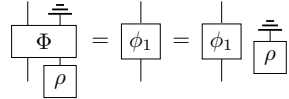
\ctikzfig{seq-effect-reduce-2}
which, by APC3 from Definition~\ref{def:apc}, recovers the first non-signalling equation in~\eqref{eq:non-signalling}:
\ctikzfig{one-way-sig}
Similarly, the symmetric conditions present in the definition of $\otimes$ recovers both equations in~\eqref{eq:non-signalling}.

By iterating these connectives, we obtain various objects that have previously been studied in the context of higher order (quantum) causal structures, such as combs \cite{Gutoski2007,Chiribella2009} and process matrices \cite{Oreshkov2012}, as demonstrated by the examples in Table \ref{table:causal_concepts}. A notable example is that of \textit{process matrices}, which are the dual of non-signalling processes. Intuitively, the two ``holes'' in this process can be thought as locations where some agent recieves an input and produces an output. The type of process matrices allows for these two locations to be causally ordered in either direction, or even be in unknown or (quantum) indefinite causal orderings.

\begin{table}
\centering
\begin{tabular}{l@{}c@{}c}
Bipartite channel & $\mathbf{(A^1 \multimap B^1) \parr (C^1 \multimap D^1)}$ & \input{./parr_channel.tikz} \\
One-way & $\mathbf{(A^1 \multimap B^1) < (C^1 \multimap D^1)}$ & \input{./seq_channel.tikz} \\
Non-signalling & $\mathbf{(A^1 \multimap B^1) \otimes (C^1 \multimap D^1)}$ & \input{./tensor_channel.tikz} \\
2-comb & $\mathbf{(A^1 \multimap B^1) \multimap (C^1 \multimap D^1)}$ & \input{./2comb_side_channel.tikz} \\
Process matrix & $\mathbf{\left( (A^1 \multimap B^1) \otimes (C^1 \multimap D^1) \right)^*}$ & \input{./process_matrix.tikz}
\end{tabular}
\caption{Encodings of concepts from quantum causality as objects in $\caus{\catc}$. The pictorial presentations are informal suggestions of ways to intuitively think of the concepts in which we view time/information as flowing up the wires and in any direction within a box. The state sets of the objects are closed under affine combinations of such diagrams.}
\label{table:causal_concepts}
\end{table}

\begin{example}\label{ex:signalling_condition}
The encoding of a non-signalling condition via $<$ can be extended to non-signalling between subsets of parties in the multi-partite case. Fix a set $V$ of agents and objects $\left\{ \mathbf{A}_v \right\}_{v \in V} \subseteq \ob{\caus{\catc}}$ describing the local system provided to each agent. A multi-partite process $\Phi$ is a state of $\bigparr_{v \in V} \mathbf{A}_v$ precisely when it acts locally according to $\mathbf{A}_v$ from the perspective of each agent $v$, regardless of the actions of the other agents. Given disjoint subsets $P, Q \subseteq V$, $\Phi$ is non-signalling from $P$ to $Q$ if, regardless of the actions chosen by $V \setminus (P \cup Q)$, the local behaviour at $Q$ is independent of any choices of actions by agents in $P$.
\begin{equation}
\begin{split}
\forall \left\{ \pi_r \right\}_{r \in V \setminus (P \cup Q)} .& \forall \left\{ \pi_p \right\}_{p \in P} . \\
\input{./non_signalling_multi_party.tikz} &= \input{./non_signalling_multi_party1.tikz}
\end{split}
\end{equation}
Combining the definitions of Equations \ref{eq:parr_object_left_control_def} and \ref{eq:seq_object_one_way_def}, this is precisely captured by the typing:
\begin{equation}
\Phi : \left( \bigparr_{r \in V \setminus (P \cup Q)} \mathbf{A}_r \right) \parr \left( \left( \bigparr_{p \in P} \mathbf{A}_p \right) > \left( \bigparr_{q \in Q} \mathbf{A}_q \right) \right)
\end{equation}
Again, if we pick each $\mathbf{A}_v$ to be a first-order channel $\mathbf{B^1}_v \multimap \mathbf{C^1}_v$ this exactly coincides with the usual presentation, and taking affine combinations of local actions at $P$ permits correllated and entangled states at their inputs.
\begin{equation}
\begin{split}
\forall \rho &: \bigotimes_{p \in P} \mathbf{B^1}_p . \\
\input{./non_signalling_multi_channel.tikz} &= \input{./non_signalling_multi_channel1.tikz}
\end{split}
\end{equation}
In Section \ref{sec:union_intersection}, we introduce a mechanism to express multiple such non-signalling properties for a single process. We then use this in Section \ref{sec:graph_types} to describe types for causal structures given by graphs in which every such non-signalling condition will have $P \cup Q = V$, partitioning the vertices into a logical past $Q$ and a logical future $P$ via a cut through the graph.
\end{example}

\subsection{Product and Coproduct}\label{sec:products_coproducts}

As well as the monoidal product of $\catc$ lifting to multiple non-degenerate monoidal products on $\caus{\catc}$, we can lift the biproduct structure to distinct products and coproducts, modelling the additives of linear logic.

\begin{align}
\mathbf{A \times B} := & \left( A \oplus B, \left( \begin{array}{rl} & \left\{p_A \fatsemi \pi_A \middle| \pi_A \in c_\mathbf{A}^* \subseteq \catc(A,I) \right\} \\ \cup & \left\{p_B \fatsemi \pi_B \middle| \pi_B \in c_\mathbf{B}^* \subseteq \catc(B,I) \right\} \end{array} \right)^* \right) \\
= & \left( A \oplus B, \left\{ \langle\rho_A, \rho_B\rangle \middle| \rho_A \in c_\mathbf{A}, \rho_B \in c_\mathbf{B} \right\} \right) \\
\mathbf{A \oplus B} := & \left( A \oplus B, \left( \left\{ \rho_A \fatsemi \iota_A \middle| \rho_A \in c_\mathbf{A} \right\} \cup \left\{ \rho_B \fatsemi \iota_B \middle| \rho_B \in c_\mathbf{B} \right\} \right)^{**} \right) \\
= & \left( A \oplus B, \left\{ \left[ \pi_A, \pi_B \right] \middle| \pi_A \in c_\mathbf{A}^* \subseteq \catc(A,I), \pi_B \in c_\mathbf{B}^* \subseteq \catc(B,I) \right\}^* \right)
\end{align}

Instead of presenting a pair of interfaces in parallel, $\mathbf{A \oplus B}$ and $\mathbf{A \times B}$ present a single interface which can be chosen to act either like $\mathbf{A}$ or $\mathbf{B}$. The distinction is whether this choice is made in advance during state preparation for $\mathbf{A \oplus B}$, or if the context is allowed to make the choice of projecting into $\mathbf{A}$ or $\mathbf{B}$ for $\mathbf{A \times B}$. As a linear resource, the choice is made to be exactly one of these, i.e. we can't choose to interact with both the $\mathbf{A}$ and $\mathbf{B}$ components of $\mathbf{A \times B}$ simultaneously, but we can make the choice probabilistically. This interpretation allows us to view $\mathbf{A \oplus B}$ instead as a kind of probabilistic test which sometimes gives one result and prepares an $\mathbf{A}$, and otherwise prepares a $\mathbf{B}$; dually, we could view $\mathbf{A \times B}$ as a system which accepts a binary test result and conditionally prepares an $\mathbf{A}$ or a $\mathbf{B}$. These operators are key to encoding (abstract) probability distributions in arbitrary causal categories.


A particular first-order object that will play a special role in this paper is the binary system $\mathbf{2} := \mathbf{I \oplus I}$. This is the smallest non-trivial state space (i.e. a space with multiple distinguishable states), and hence the simplest system that can be used to send information. By definition, the states of this object are affine mixtures of the left and right injections which we can model as two-dimensional vectors of scalars summing to $1$. In our standard examples for $\catc$, the scalars for $\MatRp$ and $\CPs$ are $\mathbb{R^+}$ and the (complete) positivity requirement for the category constrains each element of the vector to be positive, giving two-outcome probability distributions. For $\MatRaff$ we have $\mathbb{R}$ without any such positivity restriction, giving two-outcome pseudo-probabilities. In light of this interpretation, we will denote the two pure states as $\tbool := \iota_1$ and $\fbool := \iota_2$ and use this to interpret Proposition \ref{prop:complement} as the fact that any morphism from $\catc$ can be realised probabilistically (i.e. in the $\tbool$ component of a causal morphism into $\mathbf{2}$, up to a non-zero scalar).

\begin{proposition}\label{prop:complement_binary_test}
For any $f \in \catc(A, B)$ and $\mathbf{A} = \left( A, c_{\mathbf{A}} \right), \mathbf{B} = \left( B, c_{\mathbf{B}} \right) \in \ob{\caus{\catc}}$, there exists some non-zero scalar $\kappa \in \catc(I, I)$, morphism $f' \in \catc(A, B)$ and causal morphism $t_f \in \caus{\catc}\left( \mathbf{A}, \mathbf{B \parr 2} \right)$ such that $t_f = (\kappa \cdot f \otimes \tbool) + (f' \otimes \fbool)$.
\end{proposition}

\begin{remark}
Informally, distributivity between the additives and multiplicatives holds whenever the parallel content cannot influence the internal/external choices. Taking a step back to $\catc$, the monoidal and biproduct structures distribute nicely, so we don't need to be cautious about distributivity at the level of morphisms.
\begin{align}
\delta &:= \left[ \id_A \otimes \iota_B, \id_A \otimes \iota_C \right] = p_{A \otimes B} \fatsemi \left( \id_A \otimes \iota_B \right) + p_{A \otimes C} \fatsemi \left( \id_A \otimes \iota_C \right) \\
&: (A \otimes B) \oplus (A \otimes C) \to A \otimes (B \oplus C) \\
\delta^{-1} &:= \langle \id_A \otimes p_B, \id_A \otimes p_C \rangle = \left( \id_A \otimes p_B \right) \fatsemi \iota_{A \otimes B} + \left( \id_A \otimes p_C \right) \fatsemi \iota_{A \otimes C} \\
&: A \otimes (B \oplus C) \to (A \otimes B) \oplus (A \otimes C)
\end{align}
When does this lift up to $\caus{\catc}$? For $\otimes$ and $\parr$, it follows the regular expectations of $\mathrm{MALL}$ ($\mll$ with additives). Any $*$-autonomous category has an adjunction between $\otimes$ and $\multimap$, so the left-adjunct $\otimes$ preserves all colimits and dually $\parr$ will preserve all limits.
\begin{align}
\mathbf{\left( A \otimes B \right) \oplus \left( A \otimes C \right)} &\cong \mathbf{A \otimes \left( B \oplus C \right)} \\
\mathbf{\left( A \parr B \right) \times \left( A \parr C \right)} &\cong \mathbf{A \parr \left( B \times C \right)}
\end{align}
As for the other combinations, $\delta$ and $\delta^{-1}$ are still causal for one direction.
\begin{align}
\delta &: \mathbf{(A \parr B) \oplus (A \parr C)} \to \mathbf{A \parr (B \oplus C)} \\
\delta^{-1} &: \mathbf{A \otimes (B \times C)} \to \mathbf{(A \otimes B) \times (A \otimes C)}
\end{align}
However, when we look at the final cases (the distributivity linear implications that are not provable in $\mathrm{MALL}$), they do not preserve causality. For a concrete example, let's take $\mathbf{A} = \mathbf{2^*}$ and $\mathbf{B} = \mathbf{C} = \mathbf{I}$ and consider the following state:
\begin{equation}
\rho := \tbool \otimes \tbool + \fbool \otimes \fbool : \mathbf{2^* \parr (I \oplus I)} = \mathbf{2 \multimap 2}
\end{equation}
This just encodes an identity channel on a classical bit, so the choice between $\mathbf{B}$ and $\mathbf{C}$ will clearly be dependent on external choices at $\mathbf{A}$. This is in contrast with $\mathbf{(A \parr B) \oplus (A \parr C)}$ where the (probabilistic) choice between $\mathbf{B}$ and $\mathbf{C}$ must be fixed for any state and hence independent of actions at $\mathbf{A}$. This is clear when we apply $\delta^{-1}$ to $\rho$ and compare the result with $c_{\mathbf{(A \parr B) \oplus (A \parr C)}}$.
\begin{align}
\rho \fatsemi \delta^{-1} &= \left( \tbool \otimes \id_I \right) \fatsemi \iota_{A \otimes B} + \left( \fbool \otimes \id_I \right) \fatsemi \iota_{A \otimes C} \\
&\notin c_{\mathbf{(A \parr B) \oplus (A \parr C)}} = \left\{ \left( \maxmix_{2^*} \otimes \id_I \right) \fatsemi \iota_{A \otimes B}, \left( \maxmix_{2^*} \otimes \id_I \right) \fatsemi \iota_{A \otimes C} \right\}^{**}
\end{align}
The argument that $\delta$ is not causal for $\mathbf{(A \otimes B) \parr (A \otimes C)} \to \mathbf{A \otimes (B \times C)}$ follows by duality.

As for $<$, the non-signalling condition means we always preserve internal and external choices in the past.
\begin{align}
\mathbf{(A < C) \oplus (B < C)} &\cong \mathbf{(A \oplus B) < C} \\
\mathbf{(A < C) \times (B < C)} &\cong \mathbf{(A \times B) < C}
\end{align}
But for choices in the future, we similarly have $\delta$ and $\delta^{-1}$ only preserving causality in one direction (using the same counterexample as for $\parr$).
\begin{align}
\delta &: \mathbf{(A < B) \oplus (A < C)} \to \mathbf{A < (B \oplus C)} \\
\delta^{-1} &: \mathbf{A < (B \times C)} \to \mathbf{(A < B) \times (A < C)}
\end{align}
\end{remark}

\subsection{Union and Intersection}\label{sec:union_intersection}


Previous work on causal categories and similar theories have presented intersections of objects either as pullbacks from embeddings of state sets \cite{Kissinger2019a} or by combining projectors \cite{Hoffreumon2022}. Since neither of these are well-defined for every pair of objects, we first define when it is possible to take intersections.

\begin{definition}\label{def:union_intersection}
Objects $\mathbf{A}, \mathbf{A'} \in \ob{\caus{\catc}}$ are \textit{set-compatible} $\mathbf{A} \pitchfork \mathbf{A'}$ when they share the same carrier object $A \in \ob{\catc}$ and normalisation scalars $\mu_{\mathbf{A}} = \mu_{\mathbf{A'}}$ and $\theta_{\mathbf{A}} = \theta_{\mathbf{A'}}$, i.e. $c_\mathbf{A} \cap c_{\mathbf{A'}} \neq \emptyset$ (see Proposition \ref{prop:set_compatible_by_intersection}).

Given set-compatible $\mathbf{A} \pitchfork \mathbf{A'}$, we define their intersection and closed union as:
\begin{align}
\mathbf{A \cap A'} &:= \left( A, c_\mathbf{A} \cap c_{\mathbf{A'}} \right) \\
\mathbf{A \cup A'} &:= \left( A, \left( c_{\mathbf{A}} \cup c_{\mathbf{A'}} \right)^{**} \right)
\end{align}
\end{definition}

\begin{lemma}\label{lemma:union_intersection_rules}
Given objects $\mathbf{A}, \mathbf{A'}, \mathbf{B} \in \ob{\caus{\catc}}$ with $\mathbf{A} \pitchfork \mathbf{A'}$, we have the following for each $\Box \in \left\{ \otimes, \parr, <, >, \times, \oplus \right\}$ and $\lozenge \in \left\{ \cup, \cap \right\}$:
\begin{align}
\mathbf{(A \cup A')^*} &= \mathbf{A^* \cap A'^*} \label{eq:union_intersection_de_morgan} \\
\mathbf{A \otimes B} &= \mathbf{(A < B) \cap (A > B)} \label{eq:tensor_intersection_decomposition} \\
\mathbf{A \parr B} &= \mathbf{(A < B) \cup (A > B)} \label{eq:parr_union_decomposition} \\
\mathbf{(A \lozenge A') \Box B} &= \mathbf{(A \Box B) \lozenge (A' \Box B)} \label{eq:iu_monoidal_distribution}
\end{align}
\end{lemma}

The de Morgan duality of Equation \ref{eq:union_intersection_de_morgan} is not too surprising here. Equation \ref{eq:tensor_intersection_decomposition} comes from the characterisation of the non-signalling space as the affine combination of local operators \cite{Simmons2022,Hoffreumon2022,Cavalcanti2022,Chiribella2013,Gutoski2008,Al-Safi2013}. By duality, we get Equation \ref{eq:parr_union_decomposition} which, at a high level, demonstrates the principle we will elaborate on in Section \ref{sec:logic} that even settings with arbitrary communication (including random, dynamic, or indefinite causal orderings) can be expressed as an affine combination of terms with definite causal orders. Finally, Equation \ref{eq:iu_monoidal_distribution} shows that both union and intersection distribute over any of the monoidal products (a similar equation holds for $\multimap$, but the mixed variance may switch between union and intersection in line with the de Morgan duality). Proofs of these equations are given in Appendix \ref{sec:caus_construction_proofs} (some are lifted from Hoffreumon and Oreshkov's type system \cite{Hoffreumon2022} and some are new for this paper).


\begin{remark}\label{remark:no_distribution_of_union_intersection}
Whilst Hoffreumon and Oreshkov's type system permitted a distribution law $(A \cup B) \cap C = (A \cap C) \cup (B \cap C)$ \cite{Hoffreumon2022}, this doesn't generalise to our causal categories here. Consider the space $3 = I \oplus I \oplus I$ in $\MatRp$, which has a dual basis formed by the injections and projections. We can define three flat, closed spaces by $\mathbf{A_i} = \left( 3, \left\{ \iota_i, \tfrac{1}{3} \maxmix_3 \right\}^{**} \right)$ ($i \in \{1, 2, 3\}$). The intersection between any pair of these spaces is the singleton space $\left( 3, \left\{ \tfrac{1}{3} \maxmix_3 \right\} \right) \cong \mathbf{3^{1*}}$ and the union of any pair is the full first-order space $\mathbf{3^1}$. Then $\mathbf{(A_1 \cup A_2) \cap A_3} = \mathbf{3^1 \cap A_3} = \mathbf{A_3}$ but $\mathbf{(A_1 \cap A_3) \cup (A_2 \cap A_3)} \cong \mathbf{3^{1*} \cup 3^{1*}} = \mathbf{3^{1*}}$.
\end{remark}

\section{Graph Types}\label{sec:graph_types}

In causality literature, a causal structure may be defined either by factorisation (a process is compatible with a graph if it factorises into local functions dependent only on the predecessors in the graph) or signalling (a process is compatible with a graph if it is non-signalling between any two sets of parties unless connected by the successor relation). It is often crucial to distinguish between these as they are not equivalent notions: whilst every process that is localisable according to a graph satisfies the appropriate non-signalling constraints, even in the simple case of a bipartite system there exist non-signalling processes that are not localisable \cite{Popescu1994}. In this section, we will build types in $\caus{\catc}$ to directly encode these notions of causal structures up to affine closure and examine their behaviour.

\subsection{The Equivalence of Factorisations, Signalling, and Orderings}\label{sec:equivalences}


We can consider $\otimes$ and $<$ as starting points for describing causal structures, where they represent the two vertex graphs with zero and one edge respectively. However, the numerous equivalent definitions give us multiple starting points from which we can generalise the constructions. The first direction we will look at generalises the (semi-)localisable definitions of Equations \ref{eq:tensor_object_separable_def} and \ref{eq:seq_object_semi_localisable_def}. These show the spaces as spanned (under affine combination) by terms which factorise into local causal processes connected by wires along the edges of a graph (trivially for $\otimes$). In Equation \ref{eq:seq_object_semi_localisable_def} in particular, we require the intermediate system to be first-order, allowing information to travel from $\mathbf{A}$ to $\mathbf{B}$ but never in the opposite direction. This concept is one that generalises straightforwardly to factorising with wires given by the edges of a graph.

\begin{definition}\label{def:graph_morphism}
Let $G = (V, E)$ be a directed graph, and suppose $V$ is ordered. A \textit{local interpretation for $G$} is a function $\Gamma : V \to \ob{\caus{\catc}}$ assigning a causal object to each vertex. An \textit{edge interpretation for $G$} is a function $\Delta : E \to \ob{\focat{\caus{\catc}}}$. Together, these determine a \textit{component typing for $G$} $\mathrm{Comp}_G^{\Gamma;\Delta} : V \to \ob{\caus{\catc}}$ by:
\begin{equation}\label{eq:component_type_def}
\mathrm{Comp}_G^{\Gamma;\Delta}(v) := \left( \begin{array}{r@{}l} & \Gamma(v) \\ \parr & \left( \bigparr_{(u \to v) \in E} \Delta(u \to v)^* \right) \\ \parr & \left( \bigparr_{(v \to w) \in E} \Delta(v \to w) \right) \end{array} \right)
\end{equation}
The \textit{contraction morphism for $G$} is the morphism
\begin{equation}\label{eq:contraction_morphism_def}
\epsilon_G^{\Gamma;\Delta} \in \catc \left(\bigotimes_{v \in V} \mathcal{U}\left( \mathrm{Comp}_G^{\Gamma;\Delta}(v) \right), \bigotimes_{v \in V} \mathcal{U} \left( \Gamma(v) \right) \right)
\end{equation}
formed by the following rules:
\begin{itemize}
\item For each edge $(u \to v) \in E$, apply a cap $\epsilon_{\Delta(u \to v)}  : \Delta(u \to v) \otimes \Delta(u \to v)^* \to \mathbf{I}$ between the $(u \to v)$ component of $\mathrm{Comp}_G^{\Gamma;\Delta}(u)$ and the $(u \to v)$ component of $\mathrm{Comp}_G^{\Gamma;\Delta}(v)$.
\item For each vertex $v \in V$, apply an identity on the $\Gamma(v)$ within $\mathrm{Comp}_G^{\Gamma;\Delta}(v)$.
\end{itemize}
A \textit{graph state over $G$} is a morphism $g : I \to \bigotimes_{v \in V} \mathcal{U}\left( \Gamma(v) \right)$ which factorises as $g = \left( \bigotimes_{v \in V} g_v \right) \fatsemi \epsilon_G^\Gamma$ with the component at each vertex matching the component typing $\forall v \in V . g_v : \mathrm{Comp}_G^{\Gamma;\Delta}(v)$.
\end{definition}

\begin{example}
Let $N$ be the following graph.
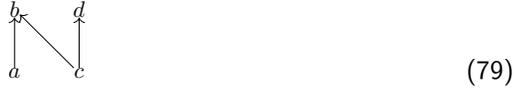
\begin{equation}
\input{./n_graph.tikz}
\end{equation}
For simplicity, consider constant interpretations $\Gamma$, $\Delta$ that map vertices to bit channels $\mathbf{2 \multimap 2}$ and each edge to an a single bit $\mathbf{2}$. The graph states over $N$ are then morphisms of the following form where each component $g_v$ is a causal state of $\mathrm{Comp}_N^{\Gamma;\Delta}(v)$:
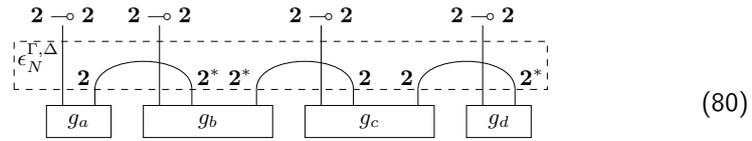
\begin{equation}
\input{./n_graph_state.tikz}
\end{equation}
Unpacking each component $g_v$ into the corresponding channel $f_v$, we see more clearly that the graph states encode networks that factorise according to the graph $N$.
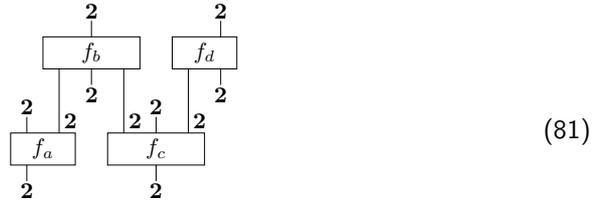
\begin{equation}
\input{./n_graph_channel.tikz}
\end{equation}
\end{example}

\begin{lemma}\label{lemma:acyclic_graph_morphisms_flat}
The set of graph states over $G$ with a fixed edge interpretation $\Delta$ is flat iff in every cycle of $G$ there is some edge $(u \to v)$ for which $\Delta(u \to v) \cong \mathbf{I}$.
\end{lemma}

The first candidate definition of a graph type considers taking affine combinations of graph states with arbitrary edge interpretations. We will also consider artificially restricting all intermediate edges to the object $\mathbf{2}$, the simplest object that can still carry information, as we will eventually see that this is sufficient to recover arbitrary systems under affine combination.

\begin{definition}\label{def:local_graph_types}
Given a DAG $G = (V, E)$ with ordered vertices and a local interpretation $\Gamma$, the \textit{local graph type over $G$} is
\begin{equation}\label{eq:local_graph_type_def}
\mathbf{LoGr}_G^\Gamma := \left( \bigotimes_{v \in V} \mathcal{U}\left( \Gamma(v) \right), \left\{ g \middle| \begin{array}{l} g \text{ is a graph state} \\ \text{over } G \text{ for some edge} \\ \text{interpretation } \Delta \end{array} \right\}^{**} \right)
\end{equation}
and the \textit{$\mathbf{2}$-local graph type over $G$} is
\begin{equation}\label{eq:2local_graph_type_def}
\mathbf{LoGr2}_G^{\Gamma} := \left( \bigotimes_{v \in V} \mathcal{U}\left( \Gamma(v) \right), \left\{ g \middle| \begin{array}{l} g \text{ is a graph state} \\ \text{over } G \text{ using } \\ \Delta :: (u \to v) \mapsto \mathbf{2} \end{array} \right\}^{**} \right)
\end{equation}
\end{definition}

Since all of the connections between components in a graph state are fixed to precisely the shape of the graph itself, local graph types give a good representation for \textit{definite causal structures} (black box processes placed in a circuit in a definite order).



The next definitions we will look at build up a graph type as the intersection of elementary signalling constraints like those of Example \ref{ex:signalling_condition}. Each constraint considers a cut through the graph that respects the directions of edges, partitioning the set of vertices into past and future; it imposes that no choice of local effects applied at the future vertices have any impact on the joint state of the past vertices. We specify the cut just by giving the past set, with validity of the cut given by asking for the past set to be \textit{down-closed} (i.e. containing all ancestors in the graph).

We provide two definitions of this kind: one which imposes no constraints on the kinds of marginal states observed in the past, and one which recursively asks the past marginals to continue respecting the signalling conditions of the graph.

\begin{definition}\label{def:signalling_graph_type}
Let $G = (V, E)$ be a DAG with ordered vertices and fix a local interpretation $\Gamma$. A subset of vertices $U \subseteq V$ is \textit{down-closed} if $\forall (u \to v) \in E . v \in U \Rightarrow u \in U$. The \textit{signalling graph type over $G$} is defined as:
\begin{equation}\label{eq:signalling_graph_type_def}
\mathbf{SiGr}_G^\Gamma := \bigcap_{U \subset V \text{ down-closed}} \mathrm{perm}_V \left( \bigparr_{u \in U} \Gamma(u) < \bigparr_{v \in V \setminus U} \Gamma(v) \right)
\end{equation}

The \textit{recursive signalling graph type over $G$} is defined as:
\begin{equation}\label{eq:recursive_signalling_def_main}
\mathbf{RSiGr}_G^\Gamma := \bigcap_{U \subset V \text{ down-closed}} \mathrm{perm}_V \left( \mathbf{RSiGr}_{G[U]}^{\Gamma[U]} < \bigparr_{v \in V \setminus U} \Gamma(v) \right)
\end{equation}
where $G[U]$ is the induced subgraph over $U$ and $\Gamma[U]$ similarly restricts the domain of $\Gamma$ to $U$. $\mathrm{perm}_V$ permutes the components wrt the tensor product $\otimes$ of $\catc$ to match the standard ordering of $V$ (and, in the case of $U = \emptyset$, applies a unitor to absorb the trivial side of the $<$) so that each term in the intersection has the same carrier object $\bigotimes_{v \in V} \mathcal{U} \left( \Gamma (v) \right)$.
\end{definition}

In the seminal paper on the $\caus{-}$ construction, Kissinger and Uijlen \cite{Kissinger2019a} explored a definition for compatibility of a multi-partite first-order channel with a partial order over the parties similar to the signalling graph type. The division of each local system into an input and an output made it possible to consider elementary constraints stating that a subset of outputs are only dependent on their inputs of their ancestors in the graph. If we split each vertex in two for input and output, we can encode this as above by considering our past set to contain the chosen set of outputs and their ancestors inputs.

They notably showed that compatibility with a partial order in this sense was equivalent to compatibility with every totalisation of the partial order \cite[Theorem 6.12]{Kissinger2019a}. This brings us to the final graph type definition we will consider:

\begin{definition}\label{def:ordered_graph_type}
The \textit{ordered graph type over $G$} is defined as:
\begin{equation}\label{eq:ordered_graph_type_def}
\mathbf{OrGr}_G^\Gamma := \bigcap_{\left[v_1, \ldots, v_n\right] \in \mathrm{sort}(G)} \mathrm{perm}_V \left( \Gamma\left(v_1\right) < \cdots < \Gamma\left(v_n\right) \right)
\end{equation}
where $\mathrm{sort}(G)$ is the set of all topological sorts of the graph $G$.
\end{definition}

Here, the graph acts as a partial order in the same sense as event orderings in distributed systems \cite{Lamport1978}, representing concurrent events via their multiple possible linear orderings. 

\begin{remark}
Relating to the category theory literature, the ordered graph type extends $\otimes$ and $<$ from series-parallel orders to arbitrary partial orders over the vertices using the pullback construction of Shapiro \& Spivak \cite{Shapiro2022}, elevating the duoidal structure to a dependence structure ($\caus{\catc}$ is a pseudo-algebra of the categorical symmetric operad of finite posets).
\end{remark}

For readers familiar with the classical causality literature, the signalling graph type condition (Equation \ref{eq:signalling_graph_type_def}) closely resembles the causal Markov condition: a probability distribution is Markov relative to a DAG iff every variable is independent of its non-descendants conditional on its parents. Since $\catc$ might not support the ability to freely construct conditional distributions, we instead phrase it in terms of signalling (whether the choice of local effects, which may cover marginalisation and some classes of interventions, has any observable effect on the remaining state), and we look at all ancestors since graph states can permit transitive signalling. The ordered graph type condition (Equation \ref{eq:ordered_graph_type_def}) resembles the ordered Markov condition with the same adjustments made.



We bring our attention the first key Theorem for this paper: the equivalence of all of these definitions. Focussing on $\mathbf{SiGr}_G^\Gamma$ and $\mathbf{OrGr}_G^\Gamma$, this lifts the result of Kissinger and Uijlen \cite[Theorem 6.12]{Kissinger2019a} from first-order channels to arbirary local systems, and between $\mathbf{LoGr}_G^\Gamma$ and $\mathbf{OrGr}_G^\Gamma$ it generalises Equations \ref{eq:tensor_object_separable_def}-\ref{eq:tensor_object_non_signalling_def} and \ref{eq:seq_object_one_way_def}-\ref{eq:seq_object_semi_localisable_def} from the empty graph for $\otimes$ or linear graphs for $<$ to arbitrary DAGs.

\begin{theorem}\label{thrm:all_graph_types_equivalent}
For any DAG $G$ with ordered vertices and local interpretation $\Gamma$, $\mathbf{LoGr}_G^\Gamma = \mathbf{LoGr2}_G^\Gamma = \mathbf{RSiGr}_G^\Gamma = \mathbf{SiGr}_G^\Gamma = \mathbf{OrGr}_G^\Gamma$.
\end{theorem}

This equivalence offers us great flexibility for proving results about graph types as we can freely switch between many different characterising properties. Many of the properties in the remainder of this paper have multiple proofs using the different kinds of graph types which may be interesting and provide insight in their own right. Now the equivalence is established, we may use $\mathbf{Gr}_G^\Gamma$ to refer to graph types in a canonical way, choosing between the alternative definitions as we see fit. We can observe that the following special cases hold by the definitions:

\begin{align}
\mathbf{Gr}_{\left( \emptyset, \emptyset \right)}^{\emptyset} &= \mathbf{I} \label{eq:empty_graph_type} \\
\mathbf{Gr}_{\left( \left\{ a \right\}, \emptyset \right)}^{a \mapsto \mathbf{A}} &= \mathbf{A} \label{eq:singleton_graph_type} \\
\mathbf{Gr}_{\left( \left\{ a, b \right\}, \emptyset \right)}^{a \mapsto \mathbf{A}, b \mapsto \mathbf{B}} &= \mathbf{A \otimes B} \label{eq:tensor_as_graph_type} \\
\mathbf{Gr}_{\left( \left\{ a, b \right\}, \left\{ a \to b \right\} \right)}^{a \mapsto \mathbf{A}, b \mapsto \mathbf{B}} &= \mathbf{A < B} \label{eq:seq_as_graph_type}
\end{align}


The equivalence of local and $\mathbf{2}$-local graph types demonstrates that any channel with arbitrary capacity over any kind of data (e.g. classical data, quantum states, or other exotic systems) can be simulated by an affine combination of classical channels with 1 bit capacity. This is not too surprising since affine combinations need not preserve information capacity - whilst taking convex combinations of pure channels will only preserve or decrease capacity, affine combinations can invert this allowing us to recover pure channels of higher capacity from mixed channels of lower capacity. The $\mathbf{2}$-local graph type definition gives an additional alternative definition for sequence types, refining Equation \ref{eq:seq_object_semi_localisable_def} to fix the side-channel system to $\mathbf{2}$.

\begin{corollary}\label{corollary:seq_as_2local}
$\mathbf{A < B} = \left( A \otimes B, \left\{ \input{./semi_local_HO_named_2.tikz} \middle| \begin{array}{l} h_{A2} \in c_{\mathbf{A \parr 2}}, \\ h_{2B} \in c_{\mathbf{2^* \parr B}} \end{array} \right\}^{**} \right)$
\end{corollary}

Looking at other specific examples of graph types beyond $\mathbf{A \otimes B}$ and $\mathbf{A < B}$, we find an exact correspondence between any object generated by $\left\{ \otimes, < \right\}$ and graph types of series-parallel graphs.

\begin{corollary}\label{corollary:series_parallel_graph_types}
Given graph types over non-empty disjoint vertex sets $V$ and $V'$,
\begin{align}
\mathbf{Gr}_{\left( V, E \right)}^{\Gamma} \otimes \mathbf{Gr}_{\left( V', E' \right)}^{\Gamma'} &= \mathbf{Gr}_{\left( V \cup V', E \cup E' \right)}^{\Gamma, \Gamma'} \label{eq:tensor_of_graph_types} \\
\mathbf{Gr}_{\left( V, E \right)}^{\Gamma} < \mathbf{Gr}_{\left( V', E' \right)}^{\Gamma'} &= \mathbf{Gr}_{\left( V \cup V', E \cup E' \cup \left\{ v \to v' \middle| v \in V, v' \in V' \right\} \right)} \label{eq:seq_of_graph_types}
\end{align}
\end{corollary}

\begin{example}\label{example:one_time_pad}
Graph types on their own are not universally expressive forms of causal structures, as there exist combinations of non-signalling relations which do not coincide with a single graph type.

Consider the encoder for the one-time-pad scheme $\mathrm{enc} : 2 \to 2 \otimes 2$:
\begin{align}
\mathrm{enc}(\tbool) &= \tfrac{1}{2}(\tbool \otimes \fbool + \fbool \otimes \tbool) \\
\mathrm{enc}(\fbool) &= \tfrac{1}{2}(\tbool \otimes \tbool + \fbool \otimes \fbool)
\end{align}
This is a famous example to show that signalling relations between individual parties are not sufficient to infer the signalling relations between subsets of parties. The input cannot signal information successfully to either one of the outputs (i.e. marginalising out the other output means we can no longer recover any information about the input), but it does signal to the joint system over both outputs since we can recover the input by XOR.

We can show that $\mathrm{enc}$ is in the intersection of the graph types $\mathbf{Gr}_{\mathbf{In} \to \mathbf{Out_1}}^\Gamma \cap \mathbf{Gr}_{\mathbf{In} \to \mathbf{Out_2}}^\Gamma$ with a single edge from $\mathbf{In}$ to one of the outputs and the other output disconnected.
\begin{equation}
\mathrm{enc} = \frac{1}{2} \left( \input{./otp_l_t.tikz} + \input{./otp_l_f.tikz} \right) = \frac{1}{2} \left( \input{./otp_r_t.tikz} + \input{./otp_r_f.tikz} \right)
\end{equation}
Expanding the signalling graph type definitions, $\mathrm{enc}$ satisfies the following non-signalling conditions:
\begin{align}
\mathbf{Out_1 < \left( In \parr Out_2 \right)} \label{eq:enc_types_1} \\
\mathbf{Out_2 < \left( In \parr Out_1 \right)} \\
\mathbf{In < \left( Out_1 \parr Out_2 \right)} \\
\mathbf{\left( In \parr Out_1 \right) < Out_2} \\
\mathbf{\left( In \parr Out_2 \right) < Out_1} \label{eq:enc_types_5}
\end{align}

However, whilst traditionally we would think of the intersection of these graphs as being the graph with no edges, $\mathrm{enc}$ is not compatible with the corresponding graph type $\mathbf{Gr}_\emptyset^\Gamma$. This type requires all of \ref{eq:enc_types_1}-\ref{eq:enc_types_5} as well as $\mathbf{\left( Out_1 \parr Out_2 \right) < In}$ which $\mathrm{enc}$ fails to satisfy. Moreover, there is no graph type that exactly matches just the intersection of \ref{eq:enc_types_1}-\ref{eq:enc_types_5}.

We conjecture that the distribution of union and intersection of Remark \ref{remark:no_distribution_of_union_intersection} does still hold over graph types with the same local interpretation, meaning Hoffreumon and Oreshkov's union of intersections of linear sequential orderings from \cite{Hoffreumon2022} remains general enough to capture all causal data.
\end{example}

\subsection{Standard Forms}\label{sec:standard_forms}

Even after showing that all definitions of graph types coincide exactly, graph types built from different graphs may coincide. In this section, we will work towards a standard form for graph types which are unique and characterise inclusion and composition. Firstly, instead of the exact graph we find that only the partial order induced by it (i.e. its transitive closure) matters.

\begin{lemma}\label{lemma:transitive_closure_graph_type_without_pruning}
$\mathbf{Gr}_G^\Gamma = \mathbf{Gr}_{G^+}^\Gamma$ where $G^+$ is the transitive closure of $G$.
\end{lemma}


Next, we look towards details in the local interpretations and how they induce equivalences between graph types. For example, if it maps every vertex to $\mathbf{I}$ (i.e. the graph states are scalars), then there is a unique state.

\begin{corollary}\label{corollary:scalar_graph_types}
$\mathbf{Gr}_{\left( V, E \right)}^{\left\{ v \mapsto \mathbf{I} \right\}_{v \in V}} \cong \mathbf{I}$
\end{corollary}

Moreover, this can be observed piecewise by the ability to remove vertices with trivial local interpretation. We still need to take the transitive closure so we don't lose any pathways that passed through this vertex.

\begin{lemma}\label{lemma:unit_elimination_from_graph_types}
$\mathbf{Gr}_{G}^{\Gamma, v \mapsto I} \cong \mathbf{Gr}_{G^+ \setminus \left\{ v \right\}}^{\Gamma}$
\end{lemma}

Suppose, instead, that a vertex is interpreted as a first-order object. Because of the degeneracy of having a single causal effect, there is no way for a local external agent to signal information to other vertices, so we can prune away any outgoing edges from that vertex and preserve the graph type. Again, we need to consider transitive closures so we don't disturb indirect signalling pathways. Dually, for first-order dual objects no information can be observed locally so we can prune away incoming edges. Combining the two recovers Lemma \ref{lemma:unit_elimination_from_graph_types} by removing both incoming and outgoing edges from a vertex.

Taking the transitive closure and then pruning edges at first-order (dual) objects yields a procedure that reduces the graph $G$ to a standard form $\overline{G}$. The final two lemmas for this section show that this form is a canonical normal form in the sense that it uniquely identifies a graph type and characterises both inclusion of graph types and dual inclusion (i.e. whether we can always compose graph states to yield the unit scalar).


\begin{lemma}\label{lemma:transitive_closure_graph_type}
Given a DAG $G = (V, E)$ and a local interpretation $\Gamma$, let $\overline{G} = (V, \overline{E})$ be the DAG with $(u \to v) \in \overline{E} \Leftrightarrow (u \to v) \in E^+ \wedge \left| c_{\Gamma(u)}^* \right| > 1 \wedge \left| c_{\Gamma(v)} \right| > 1$. Then $\mathbf{Gr}_G^\Gamma = \mathbf{Gr}_{\overline{G}}^\Gamma$. Furthermore, for any $G' = (V, E')$ over the same vertices, $c_{\mathbf{Gr}_G^\Gamma} \subseteq c_{\mathbf{Gr}_{G'}^\Gamma}$ iff $\overline{E} \subseteq \overline{E'}$.
\end{lemma}


\begin{lemma}\label{lemma:transitive_closure_dual_graph_type}
Let $G = (V, E)$ and $G' = (V, E')$ be two DAGs over the same set of vertices, and $\Gamma, \Gamma^* : V \to \ob{\caus{\catc}}$ be dual local interpretation functions, $\forall v \in V . \Gamma^*(v) = \Gamma(v)^*$. Then $c_{\mathbf{Gr}_G^\Gamma} \subseteq c_{\left( \mathbf{Gr}_{G'}^{\Gamma^*} \right)^*}$ iff $(V, \overline{E} \cup \overline{E'})$ is acyclic.
\end{lemma}

One can intuitively read this as saying that the graph types are compatible exactly when there is a total causal ordering between the vertices that respects both graph types. This falls in line with an exact characterisation of the dual space of a graph type as a union of topological orderings.

\begin{proposition}\label{prop:dual_graph_type_as_union}
$\left( \mathbf{Gr}_G^\Gamma \right)^* = \bigcup_{[v_1, \ldots, v_n] \in \mathrm{sort}(G)} \mathrm{perm}_V \left( \Gamma(v_1)^* < \cdots < \Gamma(v_n)^* \right)$
\end{proposition}

\begin{example}\label{example:transitive_closure_graphs}
Suppose we have graph types described by the following graphs, with the vertex labels picking out some local interpretation:
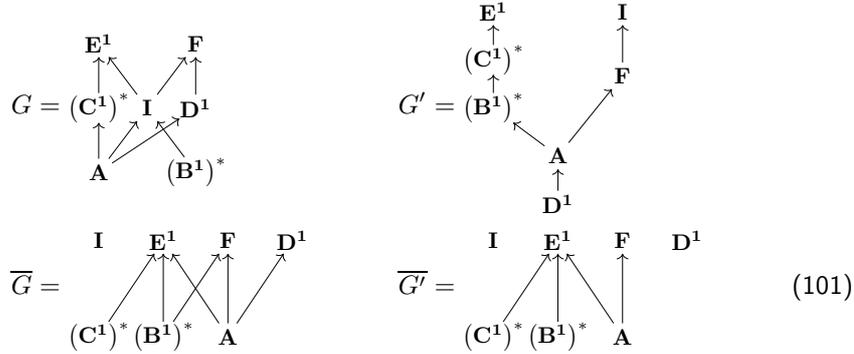
\begin{align}
G &= \input{./canonical_graph_example0.tikz} & G' &= \input{./canonical_graph_example2.tikz} \nonumber \\
\overline{G} &= \input{./canonical_graph_example1.tikz} & \overline{G'} &= \input{./canonical_graph_example3.tikz} \label{eq:transitive_closure_graphs}
\end{align}
By just looking at $G$ and $G'$, it may not be immediately obvious whether one structure generalises the other, whereas we can easily see that $\overline{G'}$ is a subgraph of $\overline{G}$. To get from $G$ to $\overline{G}$, the arcs out of $\mathbf{A}$ and $\mathbf{\left(B^1\right)^*}$ are introduced by taking the transitive closure. Then the first-order $\mathbf{D^1}$, $\mathbf{E^1}$, and $\mathbf{I}$ act as simple outputs of the network, so we remove their outgoing edges; it is impossible for local agents at them to signal to anything else since they only have a single choice of effect to apply (discarding). Similarly, the first-order dual $\mathbf{\left(B^1\right)^*}$, $\mathbf{\left(C^1\right)^*}$, and $\mathbf{I}$ act as simple inputs; an agent would see that the network acts locally like discarding regardless of actions made at other vertices, so nothing can ever signal to them.
\end{example}

\subsection{Preservation of Local Structure}\label{sec:local_structure}


The remainder of this section will look at further interplay between graph types and the local interpretation functions. In particular, if the local intepretation maps a vertex to an object with some structure, we will see when that structure can be lifted to the level of the graph types themselves. For example, graph types will preserve unions and intersections in the local interpretations.

\begin{proposition}\label{prop:graph_types_preserve_union_intersection}
\begin{align}
\mathbf{Gr}_{G}^{\Gamma, v \mapsto \mathbf{A \cup B}} &= \mathbf{Gr}_{G}^{\Gamma, v \mapsto \mathbf{A}} \cup \mathbf{Gr}_{G}^{\Gamma, v \mapsto \mathbf{B}} \label{eq:graph_types_preserve_union} \\
\mathbf{Gr}_{G}^{\Gamma, v \mapsto \mathbf{A \cap B}} &= \mathbf{Gr}_{G}^{\Gamma, v \mapsto \mathbf{A}} \cap \mathbf{Gr}_{G}^{\Gamma, v \mapsto \mathbf{B}} \label{eq:graph_types_preserve_intersection}
\end{align}
\end{proposition}

\begin{remark}
Recall from Example \ref{example:one_time_pad} that graph types are not closed under union and intersection. Even if we can recover a union of intersections form \cite{Hoffreumon2022} by having a limited distribution law, that form still won't be unique in general. For example, a graph type on a graph with $n$ vertices and no edges can be written as an intersection of all $n!$ orderings using the ordered graph type definition, but by treating it as an $n$-fold tensor product we can recursively apply Equation \ref{eq:tensor_intersection_decomposition} to write it as an intersection of only $2^n$ terms. It would be interesting to determine exactly when terms can be dropped from an intersection or union in this way.
\end{remark}


More generally, if a local interpretation maps a vertex to another graph type, then we can flatten the nesting of graphs into a single graph type by graph substitution.

\begin{proposition}\label{prop:graph_substitution_in_graph_types}
Let $G = (V, E)$ and $G' = (V', E')$ be two DAGs with disjoint vertices $V \cap V' = \emptyset$ and local interpretations $\Gamma$ and $\Gamma'$. Suppose that for some vertex $v$, $\Gamma(v) = \mathbf{Gr}_{G'}^{\Gamma'}$. Then $\mathbf{Gr}_{G}^\Gamma = \mathbf{Gr}_{G''}^{\Gamma \setminus \left\{ v \right\}, \Gamma'}$ where $G'' = \left(\left( V \setminus \left\{ v \right\} \right) \cup V', E''\right)$ with $(u \to w) \in E'' \Leftrightarrow \begin{cases}
u \to w \in E & u, w \in V \setminus \left\{ v \right\} \\
u \to v \in E & u \in V \setminus \left\{ v \right\}, w \in V' \\
v \to w \in E & u \in V', w \in V \setminus \left\{ v \right\} \\
u \to w \in E' & u, w \in V'
\end{cases}$
\end{proposition}

This result shows when it is safe to map between different levels of abstraction for a causal scenario by grouping parties together or splitting them up. We can see this as generalising Corollary \ref{corollary:series_parallel_graph_types}, since both series and parallel composition of graphs are instances of graph substitution. In terms of dependence categories \cite{Shapiro2022}, this substitution result is exactly the natural associativity isomorphism from the pseudo-algebra definition.

\section{Logical Characterisation}\label{sec:logic}

Expanding on the relationships between the operators shown in Section \ref{sec:monoidal_products}, previous work has showed that the $\caus{-}$ construction inherits enough structure from double-glueing to give a model of $\mathrm{ISOMIX}$ logic \cite{Kissinger2019a}, which can be extended to $\bv$ logic when $<$ is introduced \cite{Simmons2022}. However, neither of these logics were shown to be complete and they failed to account for the additional equations satisfied by first-order systems such as Equations \ref{eq:fo_parr_is_one_way} and \ref{eq:fo_tensor_is_one_way}. Over the rest of this paper, we will find an appropriate extension of these logics that completely describes all such equations of causal categories or, more specifically, a logic that characterises the closed diagrams of black boxes and wirings.

For examples of such diagrams, recall the encoding of process matrices from Table \ref{table:causal_concepts}. We wish to verify that, given a pair of process matrices, we can always apply them locally on either side of a single bipartite channel
\begin{equation}
\input{./process_matrix_pair_ok.tikz}
\end{equation}
\begin{equation}
\begin{split}
\left( \left( A^1 \multimap B^1 \right) \otimes \left( C^1 \multimap D^1 \right) \right)^*, \left( \left( P^1 \multimap Q^1 \right) \otimes \left( R^1 \multimap S^1 \right) \right)^*, & \\
A^1 \multimap B^1, \left( C^1 \multimap D^1 \right) \parr \left( P^1 \multimap Q^1 \right), R^1 \multimap S^1 & \Vdash
\end{split}
\end{equation}
but applying this to two bipartite channels simultaneously can break causality by the possible introduction of paradoxical scenarios, such as the cycle induced by the red paths.
\begin{equation}
\input{./process_matrix_pair_bad.tikz}
\end{equation}
\begin{equation}
\begin{split}
\left( \left( A^1 \multimap B^1 \right) \otimes \left( C^1 \multimap D^1 \right) \right)^*, \left( \left( P^1 \multimap Q^1 \right) \otimes \left( R^1 \multimap S^1 \right) \right)^*, & \\
\left( C^1 \multimap D^1 \right) \parr \left( P^1 \multimap Q^1 \right), \left( A^1 \multimap B^1 \right) \parr \left( R^1 \multimap S^1 \right) & \not\Vdash
\end{split}
\end{equation}

To refine this notion towards something more concrete, let's start with the most general form of two-sided multi-sequents, considering formulae $G_1, \ldots, G_m$ and $F_1, \ldots, F_n$ which are mapped to objects of $\caus{\catc}$ by some interpretation function $\Phi$. Suppose we have a collection of black boxes representing states $\left\{ \rho_j : \mathbf{I} \to \Phi(G_j) \right\}_j$ and effects $\left\{ \pi_i : \Phi(F_i) \to \mathbf{I} \right\}_i$ which we know to be causal, but we don't know the exact choice of state/effect implemented by each box. Furthermore, suppose the interfaces to each box (the atoms of the formulae) are labelled in a balanced way, so there is a unique way to wire up all of the interfaces by matching labels, allowing us to draw a closed/scalar diagram by wiring together the boxes. We then define $G_1, \ldots, G_m \Vdash_\catc^\Phi F_1, \ldots, F_n$ to mean that wiring the boxes together in this way is always causal (the resulting scalar diagram is $\id_I$) regardless of the contents of the boxes.

\begin{equation}\label{eq:two_sided_multi_sequent_def}
\begin{array}{l} \forall \left\{ \rho_j : \mathbf{I} \to \Phi(G_j) \right\}_j, \left\{ \pi_i : \Phi(F_i) \to \mathbf{I} \right\}_i . \\ \input{./twosided_multi_sequent_clean.tikz} = \input{./empty_diag.tikz} \end{array}
\end{equation}

By exploiting $*$-autonomy, we could always treat the effects of each $\Phi(F_i)$ as states of $\Phi(F_i)^*$, or similarly states of $\Phi(G_j)$ as effects of $\Phi(G_j)^*$, immediately giving a duality between the left and right sides of the sequent.

\begin{equation}
\begin{split}
G_1, \ldots, G_m \Vdash_\catc^\Phi F_1, \ldots, F_n \Longleftrightarrow& G_1, \ldots, G_m, F_1^*, \ldots, F_n^* \Vdash_\catc^\Phi \\
\Longleftrightarrow& \Vdash_\catc^\Phi G_1^*, \ldots, G_m^*, F_1, \ldots, F_n
\end{split}
\end{equation}

The final step we can do to simplify this is combining terms together. Recall from Equation \ref{eq:tensor_object_separable_def} that the separable states $\bigotimes_j \rho_j$ generate all states of $\bigotimes_j \Phi(G_j)$ under affine combination and similarly effects $\bigotimes_i \pi_i$ generate the effects of $\bigparr_i \Phi(F_i)$. This correctness condition is therefore equivalent to asking for the wiring diagram to be a causal morphism $\bigotimes_j \Phi(G_j) \to \bigparr_i \Phi(F_i)$.

\begin{equation}
\begin{split}
& G_1, \ldots, G_m \Vdash_\catc^\Phi F_1, \ldots, F_n \\
\Longleftrightarrow& \input{./twosided_multi_sequent_clean_morphism.tikz} : \bigotimes_j \Phi(G_j) \to \bigparr_i \Phi(F_i) \\
\Longleftrightarrow& \bigotimes_j G_j \Vdash_\catc^\Phi \bigparr_i F_i \\
\Longleftrightarrow& \left( \bigotimes_j G_j \right) \otimes \left( \bigotimes_i F_i^* \right) \Vdash_\catc^\Phi \\
\Longleftrightarrow& \Vdash_\catc^\Phi \left( \bigparr_j G_j^* \right) \parr \left( \bigparr_i F_i \right)
\end{split}
\end{equation}

Observe that, in particular, such sequents can be used to give us all inclusions between affine-closed spaces generated by the multiplicative operators: $c_{\Phi(G)} \subseteq c_{\Phi(F)} \Leftrightarrow G \Vdash_\catc^\Phi F$, i.e. the identity is causal $\Phi(G) \to \Phi(F)$.

Building towards a logic that characterises this property of causal categories, suppose we collect all terms into a single right-sided formula, i.e. we consider a single black box which the wiring diagram completely contracts away like the contraction morphisms of Definition \ref{def:graph_morphism}. This will both simplify our definitions and grant easier comparisons with existing logics, at the cost of accepting that this will involve dualising our interpretations at points to view the black boxes as effects.

The intuition for building the logic is to use Equation \ref{eq:parr_union_decomposition} to both break the type into a union of graph types and break the contraction morphism into an affine combination of graph states, then use Lemma \ref{lemma:transitive_closure_dual_graph_type} to reduce the question of compatibility to checking for acyclicity in the composite graphs.


Now that we have this intuition, we will build up the definitions formally. For now, we will focus on formulae with just the multiplicative operators and consider extensions in Section \ref{sec:logic_extensions}. Given the special behaviour of first-order and first-order dual objects (Proposition \ref{prop:first_order_equations}, Lemmas \ref{lemma:transitive_closure_graph_type} and \ref{lemma:transitive_closure_dual_graph_type}), we will allow our formulae to specify when variables are mapped to first-order objects and require that the interpretations \textit{strictly} respect this. Any objects that are both first-order and first-order dual are isomorphic to $\mathbf{I}$, so we will also allow units to appear in our formulae.

\begin{definition}\label{def:causal_formula}
We define causal formulae (in negation normal form) by the grammar $F, G ::= A | A^* | A^1 | \left(A^1\right)^* | I | F \otimes G | F < G | F \parr G$. The \textit{negation} $F^*$ of a formula is defined inductively:
\begin{equation}\label{eq:formula_negation}
\begin{split}
A^{**} &= A \\
\left(A^1\right)^{**} &= A^1 \\
I^* &= I \\
(F \otimes G)^* &= F^* \parr G^* \\
(F < G)^* &= F^* < G^* \\
(F \parr G)^* &= F^* \otimes G^*
\end{split}
\end{equation}
A formula $F$ is \textit{balanced} if any atom $A$/$A^1$ in $F$ appears exactly once in each of positive and negative form. An \textit{interpretation} $\Phi : \Var \to \ob{\caus{\catc}}$ \textit{FO-respects} $F$ when $\Phi(A)$ is first-order iff $A$ appears in $F$ as $A^1$. An interpretation can be extended to formulae inductively:
\begin{equation}\label{eq:interpretation_of_formulae}
\begin{split}
\Phi\left(A^*\right) &= \Phi(A)^* \\
\Phi\left(A^1\right) &= \Phi(A) \\
\Phi\left(\left(A^1\right)^*\right) &= \Phi\left(A^1\right)^* \\
\Phi(I) &= \mathbf{I} \\
\Phi(F \otimes G) &= \Phi(F) \otimes \Phi(G) \\
\Phi(F < G) &= \Phi(F) < \Phi(G) \\
\Phi(F \parr G) &= \Phi(F) \parr \Phi(G)
\end{split}
\end{equation}
\end{definition}

\begin{definition}\label{def:causal_consistency}
Given a balanced formula $F$ and an interpretation $\Phi$ which FO-respects $F$, the \textit{contraction morphism of $F$} is the morphism $\epsilon_F^\Phi \in \catc\left( \mathcal{U} \left( \Phi(F)^* \right), I \right)$ formed by applying a cap $\epsilon_{\Phi(A)} : \Phi(A^*)^* \otimes \Phi(A)^* \to \mathbf{I}$ between the components for each variable (regular and first-order). We say that \textit{$F$ is causally consistent for $\catc$ under $\Phi$} ($\Vdash_{\catc}^{\Phi} F$) when $\epsilon_F^\Phi$ is causal $\Phi(F)^* \to \mathbf{I}$, i.e. $\epsilon_F^\Phi \in c_{\Phi(F)}$.
\end{definition}

\subsection{Causal Proof-Nets}\label{sec:proof_nets}


Now that we have a formal property of our categories to examine, we give a proof-net definition of a logic that we argue soundly and completely characterises causal consistency. It will capture this idea of decomposing into a union of graph types via its switchings.

\begin{definition}\label{def:proof_structure}
A \textit{causal proof-structure} $P$ is a graph inductively defined by the following links:
\begin{itemize}
\item Axiom links with no premises and a pair of conclusions $A$, $A^*$.
\item FO-axiom links with no premises and a pair of conclusions $A^1$, $\left(A^1\right)^*$.
\item Unit links with no premises and conclusion $I$.
\item Cut links with premises $F$, $F^*$, and no conclusions.
\item Tensor links with premises $F$, $G$, and conclusion $F \otimes G$.
\item Seq links with premises $F$, $G$, and conclusion $F < G$.
\item Par links with premises $F$, $G$, and conclusion $F \parr G$.
\end{itemize}
A balanced formula $F$ identifies a unique cut-free causal proof-structure $P_F$ with conclusion $F$ by replacing each node in its syntax tree with the corresponding link and joining pairs of matching atoms with axiom links. Conversely, for any causal proof-structure $P$ with conclusions $\left\{ C_i \right\}_i$ we define its corresponding (possibly unbalanced) formula $F_P = C_1 \parr \cdots \parr C_n$.
\end{definition}

\begin{figure}
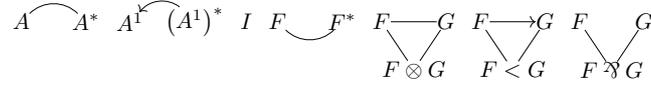

\centering
\input{./link_axiom.tikz}
\input{./link_fo_axiom.tikz}
\input{./link_unit.tikz}
\input{./link_cut.tikz}
\input{./link_tensor.tikz}
\input{./link_seq.tikz}
\input{./link_par.tikz}
\caption{A summary of the kinds of links in a causal proof-structure. The use of undirected edges is just to indicate that the switchings can induce connectivity across the link in either direction, whereas directed edges indicate that connectivity will be induced in only one direction.}\label{fig:links}
\end{figure}

\begin{definition}\label{def:switching}
An \textit{up-down switching} $s \in \mathcal{S}_P$ over a causal proof-structure $P$ is a choice of one option for each link dependent on its type from the options below, which together define the \textit{switching graph} $G_s$ over the vertices of $P$ built using the subgraphs of Figure \ref{fig:switching}:
\begin{itemize}
\item For axiom links, one of $\{la, ra\}$.
\item For FO-axiom links, the choice is fixed $\{fo\}$.
\item For unit links, the choice is fixed $\{i\}$.
\item For cut links, one of $\{lc, rc\}$.
\item For tensor links, one of $\{ul, ur, dl, dr\}$.
\item For seq links, one of $\{us, ds\}$.
\item For par links, one of $\{up, dp\}$.
\end{itemize}
An up-down switching may similarly be defined over a syntax tree of a formula, since this matches a cut-free causal proof-structure with atomic axioms in place of axiom links.
\end{definition}

\begin{figure}
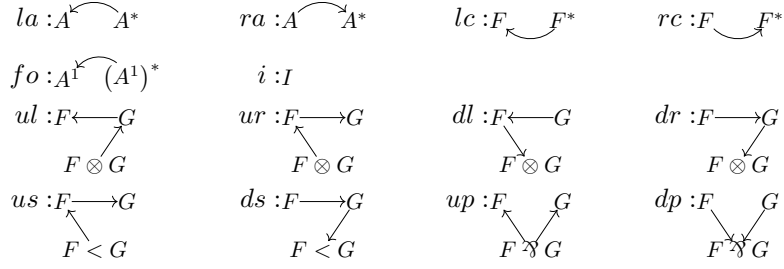

\begin{align*}
la:& \input{./switch_la.tikz} &
ra:& \input{./switch_ra.tikz} &
lc:& \input{./switch_lc.tikz} &
rc:& \input{./switch_rc.tikz} \\
fo:& \input{./link_fo_axiom.tikz} &
i:& \input{./link_unit.tikz} \\
ul:& \input{./switch_ul.tikz} &
ur:& \input{./switch_ur.tikz} &
dl:& \input{./switch_dl.tikz} &
dr:& \input{./switch_dr.tikz} \\
us:& \input{./switch_us.tikz} &
ds:& \input{./switch_ds.tikz} &
up:& \input{./switch_up.tikz} &
dp:& \input{./switch_dp.tikz}
\end{align*}
\caption{A summary of the subgraphs of a switching graph induced by the options chosen at each link.}\label{fig:switching}
\end{figure}

\begin{definition}\label{def:proof_net}
A causal proof-structure $P$ is a \textit{causal proof-net} when, for every up-down switching $s \in \mathcal{S}_P$, the corresponding switching graph $G_s$ is acyclic.
\end{definition}

There are notable similarities amongst the switchings and correctness criteria between causal proof-nets and other acyclicity conditions for $\mllmix$ (and extensions) such as long-trip, Danos-Regnier, or $\mathrm{R\&B}$-graph. Taking a $\parr$ link as an example, each of these switching methods (or the alternating colour restriction for $\mathrm{R\&B}$-graphs) allow a path between either premise and the conclusion in both directions but never a path between the premises. Rather than implementing this by connecting the conclusion to only one premise at a time, up-down switchings restrict the directions of paths to either up from the conclusion to the premises ($up$) or down from the premises to the conclusions ($dp$). Similarly, the switchings over $<$ only disallow the path between the premises in one direction, and $\otimes$ allows all paths within the link (though we still split this into multiple choices of switchings to leave directed switching graphs, similar to long-trip switchings). Using this observation, it is immediately obvious that causal logic precisely coincides with $\mllmix$ and pomset logic over the corresponding fragments because any falsifying cycle for one kind of proof-net induces a cycle for the others.

\begin{proposition}\label{prop:pomset_extension}
The logic of causal proof-nets is a conservative extension of pomset logic. Specifically, given a formula in the fragment $F, G ::= A | A^* | F \otimes G | F < G | F \parr G$, there exists a causal proof-net for $F$ iff there exists a pomset proof-net for $F$.
\end{proposition}

We can give an interpretation of the switching conditions in Figure~\ref{fig:switching} in terms of the flows of information between their premises. As we saw in Definition \ref{def:causal_consistency}, causal consistency $\Vdash_{\catc}^{\Phi} F$ means the canonical ``wiring map'' can be treated as map $\epsilon_F^\Phi : \Phi(F)^* \to \mathbf{I}$ which assigns every morphism in $c_{\Phi(F)^*}$ to the scalar 1. Crucial to the proof of the characterisation, this fails to be the case precisely when plugging $\epsilon_F^\Phi$ into $\Phi(F)^*$ could result in a directed cycle of signalling relations. The signalling relations introduced by $\epsilon_F^\Phi$ are relatively simple to characterise: a wire between $X^*$ and $X$ for a generic type $X$ can introduce signalling in either direction, whereas information can only flow in one direction for first order types (from $(A^1)^*$ to $A^1$). This follows from the fact that only thing we can ``plug in'' to $A^1$ is the unique causal effect $\discard_{A^1}$, hence there is no way to use the choice of effect to send any non-trivial information. On the other hand, we can plug any causal state of $\Phi(A^1)$ into $(A^1)^*$, of which there will typically be many.

We can then combine this with the signalling relations allowed by $\Phi(F)^*$. Since $\Phi(F)$ appears under a $(-)^*$, the roles of $\otimes$ and $\parr$, with respect to (non-)signalling relations, are reversed. Namely, an occurance of $X \otimes Y$ in $\Phi(F)^*$ means that signalling can occur in either direction between the premises $X$ and $Y$, generalising the case of process matrices highlighted in Table~\ref{table:causal_concepts}. Similarly, an occurance of $X < Y$ in $\Phi(F)^*$ means that signalling can only occur from $X$ to $Y$, whereas $X \parr Y$ in $\Phi(F)^*$ doesn't allow any signalling between premises.

In addition to fixing a flow of information between premises, we also have to consider the flow of information between a subexpression and its environment (i.e. the rest of the syntax tree): either ``up'' toward the leaves of the tree or ``down'' toward the root. Fixing this direction, plus a signalling direction between premises yields all of the possible choices for connectives (4 choices for $\otimes$, 2 choices each for $<$ and $\parr$). The ``up'' vs. ``down'' signalling direction is slightly harder to think about intuitively, but we can make this precise using (unions of) graph types, which we will do in the next section.

\subsection{Characterising Causal Objects}\label{sec:characterisation}


Recall that we can view Equation \ref{eq:parr_union_decomposition} as a way to replace any $\mathbf{A \parr B}$ with a union of sequence types. Applying this recursively along with distributivity (Equation \ref{eq:iu_monoidal_distribution}) we can therefore reduce any formula involving $\parr$ to a union of graph types. The up-down switchings are defined to respect this decomposition so the switching graphs over the syntax tree of a formula $F$ inductively yield the graph types generating $\Phi(F)^*$.

\begin{lemma}\label{lemma:switching_decomposition_of_causal_type}
Given a balanced formula $F$ and an interpretation $\Phi$ which FO-respects $F$ and an additional object $\mathbf{E} \in \ob{\caus{\catc}}$, let $V$ be the vertices of the syntax tree of $F$ (in negation normal form) and let $\Gamma_{\Phi,\mathbf{E}}^* : V \to \ob{\caus{\catc}}$ be the function:
\begin{equation}
\Gamma_{\Phi,\mathbf{E}}^* (v) = \begin{cases}
\Phi(F)^* \mathbf{\parr E} & v \text{ is the only vertex } F ::= A | A^* | A^1 | \left( A^1 \right)^* | I \\
\Phi(a)^* & v \text{ is a leaf } a ::= A | A^* | A^1 | \left(A^1\right)^* | I \\
\mathbf{E} & v \text{ is the root } $F$ \\
I & \text{otherwise}
\end{cases}
\end{equation}
Then $\Phi(F)^* \mathbf{\parr E} \cong \bigcup_{s \in \mathcal{S}_F} \mathbf{Gr}_{G_s}^{\mathbf{\Gamma}_\Phi^*}$. In words with $\mathbf{E} = \mathbf{I}$, the state space of $\Phi(F)^*$ coincides with the space generated by graph states over the switchings of the syntax tree of $F$.
\end{lemma}

\begin{example}
Consider $F = A < (I \parr A^*)$, interpreting $\Phi(A) = \mathbf{A}$ and fix some environment object $\mathbf{E}$. Then this lemma gives the following decomposition:
\begin{equation}
\begin{split}
\Phi(F)^* \parr \mathbf{E} =& \mathbf{\left( A^* < \left( I \otimes A \right) \right) \parr E} \\
\cong& \mathbf{Gr}\left(\input{./switch_ex1.tikz}\right) \cup \mathbf{Gr}\left( \input{./switch_ex2.tikz}\right) \\ & \cup \mathbf{Gr}\left(\input{./switch_ex3.tikz}\right) \cup \mathbf{Gr}\left(\input{./switch_ex4.tikz}\right)
\end{split}
\end{equation}
\end{example}


This presents a great simplification of characterising causal consistency from checking for valid contractions of arbitrary types to contractions of graph types. Physically, this is significant because some types may contain processes with indefinite causal structure \cite{Oreshkov2012,Chiribella2013}, where the patterns of influence may be inconsistent with any graph type no matter how we divide up our subsystems, yet they appear to add no new behaviour that is relevant for determining causal consistency.

As for the axiom links, we look at the contractions themselves and apply the same decomposition
\begin{equation}
\epsilon_{\Phi(A)} : \Phi(A^*) \parr \Phi(A) = \left( \Phi(A^*) < \Phi(A) \right) \cup \left( \Phi(A^*) > \Phi(A) \right)
\end{equation}
giving us directed switchings over the axiom links. The FO-axiom links only have a single switching option since 
\begin{equation}
\epsilon_{\Phi\left(A^1\right)} : \Phi\left(\left(A^1\right)^*\right) \parr \Phi(A^1) = \Phi\left(\left(A^1\right)^*\right) < \Phi(A^1)
\end{equation}
is already one-way signalling.

Combining these into a single picture, we obtain the switching graphs over the whole proof-structure. If there is a cycle, we can devise an example implementation for the black box which, when contracted, encodes a paradoxical situation (i.e. an information cycle which breaks normalisation) to disprove causal consistency; otherwise, there is some linear ordering of the vertices that all information flow respects, from which we can prove that normalisation is preserved. This covers the intuition for the proof of the characterisation theorem.

\begin{theorem}\label{thrm:causal_logic_characterisation}
Given a balanced formula $F$ and an interpretation $\Phi$ which FO-respects $F$, $\Vdash_{\catc}^\Phi F$ iff $P_F$ is a causal proof-net.
\end{theorem}

The fact that the proof-net criterion captures causal consistency for \textit{any} FO-respecting interpretation is significant: the only information about atomic objects which is relevant for determining causal relations is whether or not they are first-order or first-order dual. Whilst we already knew from Equations \ref{eq:fo_parr_is_one_way} and \ref{eq:fo_tensor_is_one_way} that they satisfied additional equations between causal types, we can now conclude that they are the \textit{only} objects that do so. This observation comes as a result of our definition of FO-respecting interpretations requiring that regular atoms are specifically mapped to objects that are neither first-order or first-order dual. If, instead, we only restricted first-order atoms and allowed regular atoms to be mapped to any object, then we would need to quantify over all interpretations ($P_F$ is a causal proof-net iff $\forall \Phi . \Vdash_\catc^\Phi F$).

In addition to interpretation-independence, the proof-net criterion is independent of our chosen base category $\catc$. Not only are the equations between causal types independent of the atomic objects, but they are also completely theory-independent: the inclusions between causal structures are the same for higher-order quantum theory ($\caus{\CPs}$), classical probability theory ($\caus{\MatRp}$), and pseudo-probability theory ($\caus{\MatRaff}$).

Recall from Proposition \ref{prop:pomset_extension} that causal logic conservatively extends pomset logic. A major consequence of this comes in terms of the computational complexity of checking causal consistency, since it must be at least as hard as verifying a pomset proof-net which is $\mathrm{coNP}$-complete \cite{Nguyen2022a}. $\mathrm{coNP}$ fits nicely with our picture of causality, since failure of consistency amounts to giving a cycle of information flow through the system which can be efficiently verified.

\begin{corollary}\label{corollary:complexity_of_causal_consistency}
Causal consistency of a balanced formula is $\mathrm{coNP}$-complete.
\end{corollary}

\subsection{Sufficient Fragments}\label{sec:fragments}

Causal logic extends pomset with first-order atoms. Proposition \ref{prop:first_order_equations} characterises first-order objects as precisely those objects over which the identity process (and hence the contraction morphism $\epsilon_{\Phi(A^1)}$ corresponding to an axiom link) is one-way signalling. It turns out that, instead of using the directed FO-axiom link, we can use this characterisation to just describe that the two atoms are linked by a one-way channel. Consider the following component of a proof-structure:

\begin{equation}
\input{./rewrite_fo_to_seq2.tikz}
\end{equation}

Enumerating the switching graphs over this component, there are no internal cycles possible and the only induced connectivity is from $Y$ to $X$, matching exactly the connectivity of an FO-axiom link. Because the connectivity matches, we can convert any cycle in a proof-structure containing an FO-axiom link to a cycle in the same proof-structure with the FO-axiom link replaced by this component and vice versa, so this replacement must preserve the correctness of the proof-net criterion.

The following definition introduces this formally as a rewrite procedure between causal proof-structures (and subsequently their corresponding formulae), which always reduces down to the fragment $F, G ::= A | A^* | F \otimes G | F < G | F \parr G$ used by pomset logic.

\begin{definition}\label{def:pom_encoding}
Any proof-structure $P$ induces a proof-structure $\mathrm{pom}\left(P\right)$ with no unit or FO-axiom links by applying the following rewrites exhaustively:
\begin{align}
\begin{tikzpicture}
	\begin{pgfonlayer}{nodelayer}
		\node [style=none] (0) at (0, 1) {$I$};
		\node [style=none] (1) at (0, 0) {$\vdots$};
		\node [style=none] (2) at (0, -1) {$C(I)$};
	\end{pgfonlayer}
\end{tikzpicture}
 &\mapsto \input{./rewrite_unit_theorem2.tikz} \\
\input{./rewrite_fo_to_seq.tikz} &\mapsto \input{./rewrite_fo_to_seq2.tikz}
\end{align}
where $X$ and $Y$ are fresh variable names for each instance of the rewrite and $C$ denotes the labels at any other link in the proof-structure as a function of the given subterms.
\end{definition}

One can similarly construct a faithful reduction down to the fragment $F, G ::= A^1 | \left( A^1 \right)^* | F \otimes G | F \parr G$ inductively defined by first-order systems. We can eliminate unit axioms in the same way, so it remains to give encodings for the regular (bidirectional) axiom and seq links.

We have already shown that causal consistency of a formula is independent of the specific objects chosen by an interpretation function (Theorem \ref{thrm:causal_logic_characterisation}). Therefore, it should be sufficient to consider interpretations that map regular variables to channels over first-order objects $\mathbf{Y^1 \multimap X^1} \cong \mathbf{\left( Y^1 \right)^* \parr X^1}$ since the channel types are neither first-order nor first-order dual. In terms of proof-structures, we can replace any regular axiom link with a pair of FO-axiom links in opposite directions and recover connectivity in both directions.

One-way signalling processes are spanned by those which factorise by the first-order system $\mathbf{2}$ (Corollary \ref{corollary:seq_as_2local}); that is, any state of $\mathbf{A < B}$ can be represented as a state of $\mathbf{\left( A \parr 2 \right) \otimes \left( 2^* \parr B \right)}$ contracted along $\mathbf{2}$. Making this contraction explicit in our proof-structure introduces a new FO-axiom link to control the direction of information flow between $\mathbf{A}$ and $\mathbf{B}$. In working with right-sided sequents and black box effects, we dualise this to replace $\mathbf{A < B}$ with $\mathbf{\left( A \otimes 2^* \right) \parr \left( 2 \otimes B \right)}$.

Again, we formalise these ideas as a rewrite procedure.

\begin{definition}\label{def:fo_encoding}
Any proof-structure $P$ induces a proof-structure $\mathrm{fo}\left(P\right)$ with no unit, regular axiom, or seq links by applying the following rewrites exhaustively:
\begin{align}
\begin{tikzpicture}
	\begin{pgfonlayer}{nodelayer}
		\node [style=none] (0) at (0, 1) {$I$};
		\node [style=none] (1) at (0, 0) {$\vdots$};
		\node [style=none] (2) at (0, -1) {$C(I)$};
	\end{pgfonlayer}
\end{tikzpicture}
 &\mapsto \input{./rewrite_unit_theorem3.tikz} \\
\input{./rewrite_ax_to_channel.tikz} &\mapsto \input{./rewrite_ax_to_channel2.tikz} \\
\input{./rewrite_seq_to_fo.tikz} &\mapsto \input{./rewrite_seq_to_fo2.tikz}
\end{align}
where $X^1$ and $Y^1$ are fresh first-order variable names for each instance of the rewrite.
\end{definition}

We are now able to formally state the faithfulness property of these encodings.

\begin{proposition}\label{prop:equivalent_fragments}
For any causal proof-structure $P$, the following are equivalent:
\begin{itemize}
\item $P$ is a causal proof-net;
\item $\mathrm{pom}\left(P\right)$ is a causal proof-net;
\item $\mathrm{fo}\left(P\right)$ is a causal proof-net.
\end{itemize}
\end{proposition}

This draws an even closer link between the logic of causal proof-nets and pomset. Not only do we have a conservative extension, but the theorems of pomset logic completely determine the truth of all statements in causal logic. Whilst there is no new behaviour that cannot be derived from pomset logic, there are still practical benefits to the model from keeping first-order systems as part of the logic due to their important physical semantics.

The other sufficient fragment corresponds to the subcategory of $\caus{\catc}$ inductively defined by first-order systems and $\multimap$ (obtaining duals and hence $\otimes$ and $\parr$ using the obvious encodings). Such inductively-defined constructions have received attention in the study of higher-order quantum processes \cite{Bisio2019,Hoffreumon2022}. Sufficiency shows that no more logical or structural behaviour exists in $\caus{\catc}$ than in these other frameworks. Despite this, $\caus{\catc}$ retains its flexibility over these alternatives because some of the additional objects still have important physical or computational interpretations - for example, in Measurement-Based Quantum Computing, each qubit is annotated with a particular axis or plane of the Block sphere for which measurement errors can be corrected, essentially allowing for postselection within that axis/plane, and in $\caus{\CPs}$ we can build qubit objects where the causal effects coincide with postselection within the chosen axis/plane (and convex combinations of these) which cannot be built inductively from any of the operators of Section \ref{sec:caus_construction}.

Proof-nets over the first-order inductive fragment resembles a variant of polarised $\mllmix$ where only the atoms are polarised. Despite linear time verification of $\mllmix$ proof-nets \cite{Nguyen2020}, by faithfully encoding all causal formulae (and hence pomset) we find that the seemingly small change of directing the axiom links jumps the complexity all the way to $\mathrm{coNP}$-complete again.

\subsection{Extensions}\label{sec:logic_extensions}

In the other direction from looking at taking fragments of the grammar of causal formulae, we could look at expanding the grammar to include arbitrary graph types, additives, unions, and intersections, and ask how causal logic should be extended to remain sound and complete.

For unions and intersections, we can always lift these to the top level of a formula (Lemma \ref{lemma:union_intersection_rules} and Proposition \ref{prop:graph_types_preserve_union_intersection}) at which point they are just existential and universal quantification over other formulae/proof-structures. Semantically sensible formulae containing unions and intersections may involve duplicates of atoms of the same name, so we would need to carefully adjust the notion of balanced formulae to guarantee that each branch of this quantification yields balanced causal formulae.

For graph types, we can either add them directly as primitives (prime graphs would be sufficient as in $\mathrm{GV}$ \cite{Acclavio2022}) or derive them as intersections of existing types. To add graph types as primitive $n$-ary operators, we can generalise seq links to $n$ premises and switch up or down to the conclusion along with switching over every linear ordering of the graph (in line with Proposition \ref{prop:dual_graph_type_as_union} and generalising the switchings of tensor links). For their dual types, we again switch up or down to the conclusion, but have a static version of the graph connecting the premises (generalising the switchings of seq links). The proof of Lemma \ref{lemma:switching_decomposition_of_causal_type} generalises using Proposition \ref{prop:graph_substitution_in_graph_types}.

We leave additives as unsolved for now. Whilst Retor\'{e}'s original paper on pomset logic \cite{Retore1997} claimed it would be easy to integrate them using proof-nets with boxes, Horne \cite{Horne2015} showed that a similar claim for unifying $\bv$ and $\mathrm{MALL}$ required an additional \textit{medial} rule $\left( (F \times G) < (K \times L) \right) \multimap \left( (F < K) \times (G < L) \right)$ (the lax interchange in the duoidal structure between $<$ and $\times$ \cite{Shapiro2022}). We therefore advocate for a similar study here to verify that pomset with additives does indeed follow straightforwardly. Whilst $\times$ and $\oplus$ are categorical products and coproducts in $\caus{\catc}$, the fact that $\mathbf{A \oplus B}$ permits affine combinations of $\mathbf{A}$ and $\mathbf{B}$ states means it is possible they may share some behaviour with probabilistic sub-additives \cite{Horne2019} instead of acting like the normal additives of linear logic. Future investigation is clearly required in order to understand the role of these additives in causal logic better.

\section{Discussion}\label{sec:discussion}

The remaining sections address some smaller follow-on remarks, building on our understanding of causal logic.

\subsection{Canonical Morphisms of Proofs}\label{sec:pomset}

We can see notable similarity between Theorem \ref{thrm:causal_logic_characterisation} and Retor\'{e}'s result that pomset logic proof-nets have a sound and faithful interpretation in coherence spaces \cite{Retore1994}. In that result, it is shown that correctness of a proof-net is not only implied by the existence of a clique for every interpretation of the atoms, but it is sufficient to only consider interpretations that map each atom to some four-token coherence spaces $\mathrm{N}$ and $\mathrm{Z} = \mathrm{N^\bot}$. By exploiting the interpretation-independence of Theorem \ref{thrm:causal_logic_characterisation}, we can do a similar thing here by picking an interpretation $\Phi$ and taking the contraction morphism $\epsilon_F^\Phi$ as a canonical representation of the proof.

Because of our distinction of different classes of atoms, we need more elementary objects for canonical interpretations of proofs. $\mathbf{2}$ is a good candidate for first-order variables, it being the simplest object which is first-order but not first-order dual.

For regular variables, we should associate them with an object that is neither first-order or first-order dual, giving us a couple of options that stand out. $\mathbf{A_1} = \left( 3, \left\{ \iota_1, \tfrac{1}{3} \maxmix_3 \right\}^{**} \right)$ from Remark \ref{remark:no_distribution_of_union_intersection} is similar to $\mathrm{N}$ in that it is the object with lowest dimensions (with respect to both the carrier object and state set) which is not isomorphic to any object definable from $\mathbf{I}$ with $\mathbf{\oplus}, \mathbf{\times}, \mathbf{(-)^*}, \mathbf{\otimes}, \mathbf{<}$.

Alternatively, we could apply the $\mathrm{fo}(-)$ procedure of Definition \ref{def:fo_encoding} to relate the variable to a simple channel over first-order systems, naturally leading to us picking $\mathbf{2 \multimap 2}$ for our canonical interpretation. This is appealing since we can then see the interpretation as built entirely with $\mathbf{I \multimap I} (\cong \mathbf{I})$, $\mathbf{I \multimap 2} (\cong \mathbf{2})$, $\mathbf{2 \multimap I} (\cong \mathbf{2^*})$, and $\mathbf{2 \multimap 2}$, exhibiting the sufficiency of stochastic maps on classical bits for studying causal structures.

\subsection{Consistency as Extranatural Transformations}\label{sec:extranatural}


If we actually want to draw conclusions about the coherent structure of causal categories, we may wish for an alternative definition of causal consistency that determines a property of the category as a whole, rather than being dependent on a particular interpretation function. We can achieve this by lifting formulae to functors and interpreting a sequent as the existence of an extranatural transformation between them.

\begin{definition}\label{def:causal_functor}
The \textit{causal functor $\mathcal{F}_F$} of a formula $F$ is defined inductively:
\begin{equation}\label{eq:causal_functor_def}
\begin{split}
\mathcal{F}_A &= 1_{\caus{\catc}} : \caus{\catc} \to \caus{\catc} \\
\mathcal{F}_{A^*} &= (-)^* : \caus{\catc}^{op} \to \caus{\catc} \\
\mathcal{F}_{A^1} &= \iota : \focat{\caus{\catc}} \hookrightarrow \caus{\catc} \\
\mathcal{F}_{\left(A^1\right)^*} &= \iota^{op} \fatsemi (-)^* : \left(\focat{\caus{\catc}}\right)^{op} \to \caus{\catc} \\
\mathcal{F}_{I} &= \mathcal{I} : 1 \to \caus{\catc} \\
\mathcal{F}_{F \otimes G} &= \left( \mathcal{F}_F \times \mathcal{F}_G \right) \fatsemi \mathbf{\otimes} : \mathrm{dom}(\mathcal{F}_F) \times \mathrm{dom}(\mathcal{F}_G) \to \caus{\catc} \\
\mathcal{F}_{F < G} &= \left( \mathcal{F}_F \times \mathcal{F}_G \right) \fatsemi \mathbf{<} : \mathrm{dom}(\mathcal{F}_F) \times \mathrm{dom}(\mathcal{F}_G) \to \caus{\catc} \\
\mathcal{F}_{F \parr G} &= \left( \mathcal{F}_F \times \mathcal{F}_G \right) \fatsemi \mathbf{\parr} : \mathrm{dom}(\mathcal{F}_F) \times \mathrm{dom}(\mathcal{F}_G) \to \caus{\catc}
\end{split}
\end{equation}
where $\mathcal{I}$ picks out $\mathcal{I}(\star) = \mathbf{I}$, $\mathcal{I}(\id_\star) = \id_i$. Note that the objects of the domain of $\mathcal{F}_F$ coincide with interpretation functions over the free variables of $F$.
\end{definition}

Because each morphism in $\caus{\catc}$ is a morphism in $\catc$, we can use the existence of a basis of states/effects \ref{apc:basis} to identify each term of an extranatural transformation between causal functors and determine that it is necessarily just a wiring diagram.

\begin{lemma}\label{lemma:unique_causal_extranatural}
Given a balanced formula $F$, there is at most one extranatural transformation $\eta_F : \mathcal{I} \to \mathcal{F}_F$ which is extranatural across the pairs of indices given by matching atoms in $F$. When it exists, the term attributed to an interpretation $\Phi : \Var \to \ob{\caus{\catc}}$ is precisely $\left( \epsilon_F^\Phi \right)^*$, i.e. a selection of compact cups connecting matching atoms.
\end{lemma}

\begin{definition}\label{def:causal_consistency_by_extranaturality}
A balanced formula $F$ is \textit{causally consistent for $\catc$ by extranatural transformation} ($\Vdash_\catc F$) when there exists an extranatural transformation $\eta_F : \mathcal{I} \to \mathcal{F}_F$ which is extranatural across the pairs of indices given by matching atoms if $F$.
\end{definition}

It is always possible to construct $\left( \epsilon_F^\Phi \right)^*$ for any $\Phi$, and they will always collectively satisfy extranaturality by sliding morphisms around the cups in $\catc$. This means the only thing left that determines $\Vdash_\catc F$ is whether each $\epsilon_F^\Phi$ is causal, i.e. $\Vdash_\catc^\Phi F$.

\begin{proposition}\label{prop:equivalent_defs_of_causal_consistency}
For any balanced formula $F$, $\Vdash_\catc F$ iff for all interpretations $\Phi$ which FO-respect $F$ we have $\Vdash_\catc^\Phi F$, i.e. iff $P_F$ is a causal proof-net.
\end{proposition}

This gives us a more category-theoretic way to characterise the structure of causal categories. At the time of this paper, there is no known category-theoretic characterisation of pomset logic, but the existence of such extranatural transformations for all provable formulae makes a convincing case that $\caus{\catc}$ should be an instance of a ``pomset category'' for any suitable definition.

We could similarly extend this idea to two-sided sequents, having $G \Vdash_\catc F$ hold when there exists an extranatural transformation $\mathcal{F}_G \to \mathcal{F}_F$. We conjecture that adding unions, intersections, or additives as per Section \ref{sec:logic_extensions} will require looking at dinatural transformations instead to account for duplication of atom names.

\subsection{Separating $\bv$ and Pomset}\label{sec:separating_statement}


Our results here show that pomset logic has a tight correspondence with consistent formulae, improving on the previous lower bound of $\bv$ \cite{Simmons2022}. This allows us, in particluar, to look at the separating statements that are theorems of pomset logic but not provable in $\bv$ and attribute them with a physical intuition. For example, Nguy\textipa{\~{\^{e}}}n and Straßburger \cite{Nguyen2022} proved that the following is a theorem of pomset logic but not of $\bv$.
\begin{equation}\label{eq:pomset_bv_separating_theorem}
\begin{split}
\vdash ((A < B) \otimes (C < D)) \parr ((E < F) \otimes (G < H)) \parr (A^* < H^*) \\ \parr (E^* < B^*) \parr (G^* < D^*) \parr (C^* < F^*)
\end{split}
\end{equation}

One can take this sequent and interpret it directly via causal consistency as a valid wiring diagram of black boxes (dualising each formula to describe states).

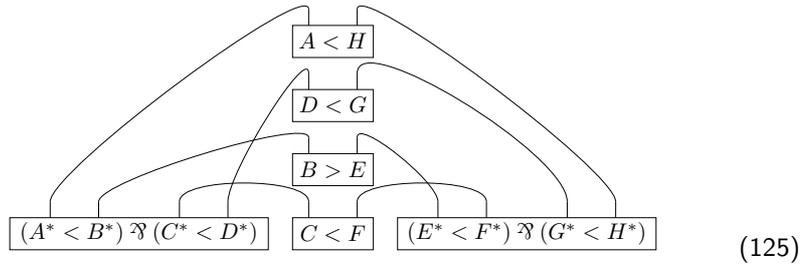
\begin{equation}
\input{./separating_statement.tikz}
\end{equation}

There are still a lot of boxes here, making it quite cumbersome to attempt read this intuitively. However, we recognise that this sequent is the outcome of applying the $\mathrm{pom}(-)$ operation to the following sequent (using new labels to distinguish the new first-order atoms from the previous non-degenerate atoms):
\begin{equation}\label{eq:pomset_bv_separating_theorem_fo}
\vdash \left( \left( \left( P^1 \right)^* < Q^1 \right) \otimes \left( \left( R^1 \right)^* < S^1 \right) \right) \parr \left( \left( \left( Q^1 \right)^* < R^1 \right) \otimes \left( \left( S^1 \right)^* < P^1 \right) \right)
\end{equation}

This has reduced all of the atomic systems down to first-order. By Proposition \ref{prop:first_order_equations}, $<$ is equivalent to $\parr$ (and $\multimap$ up to the appropriate duals) in these instances, giving our final sequent:
\begin{equation}\label{eq:pomset_bv_separating_theorem_simple}
\left( \left( P^1 \multimap Q^1 \right) \otimes \left( R^1 \multimap S^1 \right) \right) ^* \vdash \left( Q^1 \multimap R^1 \right) \otimes \left( S^1 \multimap P^1 \right)
\end{equation}

Looking back at Table \ref{table:causal_concepts} for our common concepts, we recognise that this is precisely the statement that any process matrix can be represented as a non-signalling first-order channel by swapping the outputs.

\begin{equation}\label{eq:process_matrix_non_signalling}
\input{./process_matrix_non_signalling_fo.tikz} \sim \sum_i \alpha_i \input{./process_matrix_non_signalling_fo1.tikz}
\end{equation}

Since the $\mathrm{pom}(-)$ operation preserves correctness (Proposition \ref{prop:equivalent_fragments}), this really is equivalent to the original separating statement of Equation \ref{eq:pomset_bv_separating_theorem} but far easier to interpret as something physically meaningful.

\section{Conclusion}\label{sec:conclusion}

Working with the $\caus{-}$ construction as a very general model of deterministic higher-order processes from an arbitrary base theory with maximal freedom of composition, we have reintroduced definite (DAG) causal structures from traditional causality frameworks into this model via graph types and shown them to be a flexible and powerful tool in studying the structure of these categories. Much of this comes from the unification of definitions of compatibility with a graph based on non-signalling constraints and factorisations up to affine combinations (Theorem \ref{thrm:all_graph_types_equivalent}), bolstered by preserving any local structure through graph substitution (Proposition \ref{prop:graph_substitution_in_graph_types}) and standard forms that characterise inclusions and compatibility for composition (Lemmas \ref{lemma:transitive_closure_graph_type} and \ref{lemma:transitive_closure_dual_graph_type}), which respectively amount to checking possible signalling paths between individual parties or the existence of any signalling loops that can break normalisation by encoding a paradoxical scenario.

This maximal freedom of composition for black boxes can be formally characterised by a sound and complete logic for wiring diagrams in $\caus{\catc}$ (Theorem \ref{thrm:causal_logic_characterisation}), built intentionally from the model in such a way that switchings over proof-structures enumerate graph types that span the state space of the black boxes under union (Lemma \ref{lemma:switching_decomposition_of_causal_type}). Note that even though the processes in the category may exhibit indefinite causal structure, this shows that such processes can always be decomposed into an affine combination of deterministic processes with definite causal structure. Furthermore, independence of the logic from local systems and base theory is very powerful, showing that no more behaviour (with respect to deterministic composition of typed black boxes) exists in higher-order quantum processes than simple channels over single classical bits.

The logic can be faithfully reduced down to two relevant fragments: one showing that all structure is determined by the sub-theory inductively defined from first-order systems, and the other, pomset \cite{Retore2020}, giving a $\mathrm{coNP}$-complete complexity for checking causal consistency of a wiring diagram. Connecting pomset with a physical model allowed us to show that Nguy\textipa{\~{\^{e}}}n and Straßburger's separating statement between $\bv$ and pomset \cite{Nguyen2022} exactly describes that bipartite process matrices can be represented as non-signalling channels up to permutation of the outputs.

We can identify a few areas of future work that receive motivation from the results in this paper:

\begin{itemize}
\item Whilst Theorem \ref{thrm:all_graph_types_equivalent} unifies a number of different ways to express compatibility with a causal structure, it by no means covers every such way. One notable missing case is quantum causal models \cite{Barrett2019} and their classical predecessors \cite{Pearl}. Categorifications of classical causal models \cite{Fritz2020,Lorenz2023} tend to assume the existence of a copy-delete comonoid but there is no completely-positive trace-preserving comonoid to extend this to the quantum case. We expect that one could use the existence of such comonoids in $\sub{\catc}$ to present a similar construction and show it is equivalent to the graph types presented here.
\item In Section \ref{sec:logic_extensions} we discussed how to extend the logic to richer grammars, incorporating additional operators present in the $\caus{-}$ construction. In particular, the additives may require a more thorough treatment to reveal how they work in the semantics of causal categories and to ensure that all desirable behaviour is captured when combined with existing pomset techniques.
\item Since the failure of causal consistency for some scenario amounts to the existence of a causal loop, one can analogise this as checking for deadlock-freedom in a concurrent programming model. This raises some curiosity as to whether there is any correspondence between deadlock-free $\pi$-calculus programs and pomset or a related logic that can be proved using similar techniques.
\end{itemize}

\subsection*{Acknowledgements}

The authors would like to thank the members of the quantum group at University of Oxford and attendees of the Mathematical Foundations of Computation seminar at University of Bath and the CHoCoLa seminar at ENS Lyon for discussions and feedback, especially to Alessio Guglielmi, Chris Barrett, Tito Nguy\textipa{\~{\^{e}}}n, and Jonathan Barrett. Additional thanks to Robin Lorenz and Sean Tull for feedback on drafts of this paper. WS was supported by Quantinuum Ltd. AK acknowledges support from Grant No. 61466 from the John Templeton Foundation as part of the QISS project. The opinions expressed in this publication are those of the authors and do not necessarily reflect the views of the John Templeton Foundation.

\bibliography{library}

\begin{thebibliography}{10}

\bibitem{Acclavio2022}
Matteo Acclavio, Ross Horne, Sjouke Mauw, and Lutz Stra{\ss}burger.
\newblock {A Graphical Proof Theory of Logical Time}.
\newblock In {\em Leibniz International Proceedings in Informatics, LIPIcs},
  volume 228, pages 22:1--22:0, 2022.
\newblock \href
  {https://doi.org/10.4230/LIPIcs.FSCD.2022.22}
  {\path{doi:10.4230/LIPIcs.FSCD.2022.22}}.

\bibitem{Al-Safi2013}
Sabri~W. Al-Safi and Anthony~J. Short.
\newblock {Simulating all nonsignaling correlations via classical or quantum
  theory with negative probabilities}.
\newblock {\em Physical Review Letters}, 111(17), 2013.
\newblock \href
  {https://doi.org/10.1103/PhysRevLett.111.170403}
  {\path{doi:10.1103/PhysRevLett.111.170403}}.

\bibitem{Apadula2022}
Luca Apadula, Alessandro Bisio, and Paolo Perinotti.
\newblock {No-signalling constrains quantum computation with indefinite causal
  structure}.
\newblock {\em Quantum}, 8:1241, feb 2024.
\newblock \href
  {https://doi.org/10.22331/q-2024-02-05-1241}
  {\path{doi:10.22331/q-2024-02-05-1241}}.

\bibitem{Barrett2019}
Jonathan Barrett, Robin Lorenz, and Ognyan Oreshkov.
\newblock {Quantum Causal Models}.
\newblock jun 2019.
\newblock URL: \href
  {http://arxiv.org/abs/1906.10726} {\path{arXiv:1906.10726}}.

\bibitem{BaumelerParadox}
{\"{A}}min Baumeler and Eleftherios Tselentis.
\newblock {Equivalence of Grandfather and Information Antinomy Under
  Intervention}.
\newblock In Beno{\^{i}}t Valiron, Shane Mansfield, Pablo Arrighi, and Prakash
  Panangaden, editors, {\em Electronic Proceedings in Theoretical Computer
  Science}, volume 340 of {\em Electronic Proceedings in Theoretical Computer
  Science}, pages 1--12. Open Publishing Association, sep 2021.
\newblock \href
  {https://doi.org/10.4204/EPTCS.340.1} {\path{doi:10.4204/EPTCS.340.1}}.

\bibitem{Bisio2019}
Alessandro Bisio and Paolo Perinotti.
\newblock {Theoretical framework for higher-order quantum theory}.
\newblock {\em Proceedings of the Royal Society A: Mathematical, Physical and
  Engineering Sciences}, 475(2225), may 2019.
\newblock \href
  {https://doi.org/10.1098/rspa.2018.0706} {\path{doi:10.1098/rspa.2018.0706}}.

\bibitem{Blute2012}
Richard Blute, Prakash Panangaden, and Sergey Slavnov.
\newblock {Deep inference and probabilistic coherence spaces}.
\newblock {\em Applied Categorical Structures}, 2012.
\newblock \href {https://doi.org/10.1007/s10485-010-9241-0}
  {\path{doi:10.1007/s10485-010-9241-0}}.

\bibitem{Cavalcanti2022}
Paulo~J. Cavalcanti, John~H. Selby, Jamie Sikora, and Ana~Bel{\'{e}}n Sainz.
\newblock {Decomposing all multipartite non-signalling channels via
  quasiprobabilistic mixtures of local channels in generalised probabilistic
  theories}.
\newblock {\em Journal of Physics A: Mathematical and Theoretical},
  55(40):404001, oct 2022.
\newblock \href
  {https://doi.org/10.1088/1751-8121/ac8ea4}
  {\path{doi:10.1088/1751-8121/ac8ea4}}.

\bibitem{Chiribella2009}
Giulio Chiribella, Giacomo~Mauro D'Ariano, and Paolo Perinotti.
\newblock {Theoretical framework for quantum networks}.
\newblock {\em Physical Review A}, 80(2):022339, aug 2009.
\newblock \href
  {https://doi.org/10.1103/PhysRevA.80.022339}
  {\path{doi:10.1103/PhysRevA.80.022339}}.

\bibitem{Chiribella2010}
Giulio Chiribella, Giacomo~Mauro D'Ariano, and Paolo Perinotti.
\newblock {Probabilistic theories with purification}.
\newblock {\em Physical Review A}, 81(6):062348, jun 2010.
\newblock \href
  {https://doi.org/10.1103/PhysRevA.81.062348}
  {\path{doi:10.1103/PhysRevA.81.062348}}.

\bibitem{Chiribella2013}
Giulio Chiribella, Giacomo~Mauro D'Ariano, Paolo Perinotti, and Benoit Valiron.
\newblock {Quantum computations without definite causal structure}.
\newblock {\em Physical Review A}, 88(2):022318, aug 2013.
\newblock \href
  {https://doi.org/10.1103/PhysRevA.88.022318}
  {\path{doi:10.1103/PhysRevA.88.022318}}.

\bibitem{Cockett1997a}
J.~R.B. Cockett and R.~A.G. Seely.
\newblock {Proof theory for full intuitionistic linear logic, bilinear logic,
  and MIX categories}.
\newblock {\em Theory and Applications of Categories}, 3(5):85--131, 1997.

\bibitem{Coecke2014}
Bob Coecke.
\newblock {Terminality implies non-signalling}.
\newblock In {\em Electronic Proceedings in Theoretical Computer Science,
  EPTCS}, volume 172, pages 27--35, 2014.
\newblock \href
  {https://doi.org/10.4204/EPTCS.172.3} {\path{doi:10.4204/EPTCS.172.3}}.

\bibitem{Coecke2013}
Bob Coecke and Raymond Lal.
\newblock {Causal Categories: Relativistically Interacting Processes}.
\newblock {\em Foundations of Physics}, 43(4):458--501, apr 2013.
\newblock \href {https://doi.org/10.1007/s10701-012-9646-8}
  {\path{doi:10.1007/s10701-012-9646-8}}.

\bibitem{Earnshaw2024}
Matt Earnshaw, James Hefford, and Mario Rom{\'{a}}n.
\newblock {The Produoidal Algebra of Process Decomposition}.
\newblock In Aniello Murano and Alexandra Silva, editors, {\em 32nd EACSL
  Annual Conference on Computer Science Logic (CSL 2024)}, volume 288 of {\em
  Leibniz International Proceedings in Informatics (LIPIcs)}, pages
  25:1----25:19, Dagstuhl, Germany, 2024. Schloss Dagstuhl -- Leibniz-Zentrum
  f{\"{u}}r Informatik.
\newblock \href {https://doi.org/10.4230/LIPIcs.CSL.2024.25}
  {\path{doi:10.4230/LIPIcs.CSL.2024.25}}.

\bibitem{Fritz2020}
Tobias Fritz.
\newblock {A synthetic approach to Markov kernels, conditional independence and
  theorems on sufficient statistics}.
\newblock {\em Advances in Mathematics}, 370, 2020.
\newblock \href
  {https://doi.org/10.1016/j.aim.2020.107239}
  {\path{doi:10.1016/j.aim.2020.107239}}.

\bibitem{Gutoski2008}
Gus Gutoski.
\newblock {Properties of Local Quantum Operations with Shared Entanglement}.
\newblock {\em Quantum Information and Computation}, 9(9-10):0739--0764, may
  2008.
\newblock \href
  {http://arxiv.org/abs/0805.2209} {\path{arXiv:0805.2209}}.

\bibitem{Gutoski2007}
Gus Gutoski and John Watrous.
\newblock {Toward a general theory of quantum games}.
\newblock {\em Proceedings of the Annual ACM Symposium on Theory of Computing},
  pages 565--574, 2007.
\newblock \href {https://doi.org/10.1145/1250790.1250873}
  {\path{doi:10.1145/1250790.1250873}}.

\bibitem{Hefford2023}
James Hefford and Cole Comfort.
\newblock {Coend Optics for Quantum Combs}.
\newblock In {\em Electronic Proceedings in Theoretical Computer Science,
  EPTCS}, volume 380, pages 63--76. Open Publishing Association, aug 2023.
\newblock \href
  {https://doi.org/10.4204/EPTCS.380.4} {\path{doi:10.4204/EPTCS.380.4}}.

\bibitem{Hefford2024}
James Hefford and Matt Wilson.
\newblock {A Profunctorial Semantics for Quantum Supermaps}.
\newblock feb 2024.
\newblock \href
  {http://arxiv.org/abs/2402.02997} {\path{arXiv:2402.02997}}.

\bibitem{Hoffreumon2022}
Timoth{\'{e}}e Hoffreumon and Ognyan Oreshkov.
\newblock {Projective characterization of higher-order quantum
  transformations}.
\newblock {\em arxiv.org}, 2022.
\newblock \href {http://arxiv.org/abs/2206.06206}
  {\path{arXiv:2206.06206}}.

\bibitem{Horne2015}
Ross Horne.
\newblock {The consistency and complexity of multiplicative additive system
  virtual}.
\newblock {\em Scientific Annals of Computer Science}, 25(2):245--316, 2015.
\newblock \href {https://doi.org/10.7561/SACS.2015.2.245}
  {\path{doi:10.7561/SACS.2015.2.245}}.

\bibitem{Horne2019}
Ross Horne.
\newblock {The sub-additives: A proof theory for probabilistic choice extending
  linear logic}.
\newblock In {\em Leibniz International Proceedings in Informatics, LIPIcs},
  volume 131, 2019.
\newblock \href
  {https://doi.org/10.4230/LIPIcs.FSCD.2019.23}
  {\path{doi:10.4230/LIPIcs.FSCD.2019.23}}.

\bibitem{Houston2008}
Robin Houston.
\newblock {Finite products are biproducts in a compact closed category}.
\newblock {\em Journal of Pure and Applied Algebra}, 2008.
\newblock \href {https://doi.org/10.1016/j.jpaa.2007.05.021}
  {\path{doi:10.1016/j.jpaa.2007.05.021}}.

\bibitem{Hyland2003}
Martin Hyland and Andrea Schalk.
\newblock {Glueing and orthogonality for models of linear logic}.
\newblock In {\em Theoretical Computer Science}, volume 294, pages 183--231,
  2003.
\newblock \href {https://doi.org/10.1016/S0304-3975(01)00241-9}
  {\path{doi:10.1016/S0304-3975(01)00241-9}}.

\bibitem{jacobs2019causal}
Bart Jacobs, Aleks Kissinger, and Fabio Zanasi.
\newblock {Causal Inference by String Diagram Surgery}.
\newblock In {\em Foundations of Software Science and Computation Structures:
  22nd International Conference, FOSSACS 2019, Held as Part of the European
  Joint Conferences on Theory and Practice of Software, ETAPS 2019, Prague,
  Czech Republic, April 6--11, 2019, Proceedings}, pages 313--329. Springer,
  2019.
\newblock \href {https://doi.org/10.1007/978-3-030-17127-8_18}
  {\path{doi:10.1007/978-3-030-17127-8_18}}.

\bibitem{Jia2018}
Ding Jia and Nitica Sakharwade.
\newblock {Tensor products of process matrices with indefinite causal
  structure}.
\newblock {\em Physical Review A}, 97(3):032110, mar 2018.
\newblock \href
  {https://doi.org/10.1103/PhysRevA.97.032110}
  {\path{doi:10.1103/PhysRevA.97.032110}}.

\bibitem{Kahramanogullari2008}
Ozan Kahramanoǧullari.
\newblock {System BV is NP-complete}.
\newblock {\em Annals of Pure and Applied Logic}, 152(1-3):107--121, 2008.
\newblock \href {https://doi.org/10.1016/j.apal.2007.11.005}
  {\path{doi:10.1016/j.apal.2007.11.005}}.

\bibitem{Kissinger2017a}
Aleks Kissinger, Matty Hoban, and Bob Coecke.
\newblock {Equivalence of relativistic causal structure and process
  terminality}.
\newblock aug 2017.
\newblock \href
  {http://arxiv.org/abs/1708.04118} {\path{arXiv:1708.04118}}.

\bibitem{Kissinger2019a}
Aleks Kissinger and Sander Uijlen.
\newblock {A categorical semantics for causal structure}.
\newblock {\em Logical Methods in Computer Science}, 2019.
\newblock \href {https://doi.org/10.23638/LMCS-15(3:15)2019}
  {\path{doi:10.23638/LMCS-15(3:15)2019}}.

\bibitem{Lamport1978}
Leslie Lamport.
\newblock {Time, Clocks, and the Ordering of Events in a Distributed System}.
\newblock {\em Communications of the ACM}, 21(7):558--565, jul 1978.
\newblock \href
  {https://doi.org/10.1145/359545.359563} {\path{doi:10.1145/359545.359563}}.

\bibitem{Lorenz2023}
Robin Lorenz and Sean Tull.
\newblock {Causal models in string diagrams}.
\newblock 2023.
\newblock \href
  {http://arxiv.org/abs/2304.07638} {\path{arXiv:2304.07638}}.

\bibitem{Mellies2010}
Paul-Andr{\'{e}} Melli{\`{e}}s.
\newblock {A Topological Correctness Criterion for Multiplicative
  Non-Commutative Logic}.
\newblock In {\em Linear Logic in Computer Science}, volume 316, pages
  283--322. London Mathematical Society, 2010.
\newblock \href
  {https://doi.org/10.1017/cbo9780511550850.009}
  {\path{doi:10.1017/cbo9780511550850.009}}.

\bibitem{Murawski2006}
Andrzej~S. Murawski and C.~H.Luke Ong.
\newblock {Fast verification of MLL proof nets via IMLL}.
\newblock {\em ACM Transactions on Computational Logic}, 7(3):473--498, jul
  2006.
\newblock \href
  {https://doi.org/10.1145/1149114.1149116}
  {\path{doi:10.1145/1149114.1149116}}.

\bibitem{Nguyen2020}
L{\^{e}}~Th{\`{a}}nh~Dũng Nguy{\^{e}}n.
\newblock {Unique perfect matchings, forbidden transitions and proof nets for
  linear logic with mix}.
\newblock {\em Logical Methods in Computer Science}, 16(1):27:1--27:31, 2020.
\newblock \href
  {https://doi.org/10.23638/LMCS-16(1:27)2020}
  {\path{doi:10.23638/LMCS-16(1:27)2020}}.

\bibitem{Nguyen2022}
L{\^{e}}~Th{\`{a}}nh~Dũng Nguy{\^{e}}n and Lutz Stra{\ss}burger.
\newblock {BV and Pomset Logic Are Not the Same}.
\newblock In {\em Leibniz International Proceedings in Informatics, LIPIcs},
  volume 216, pages 1--32, 2022.
\newblock \href {https://doi.org/10.4230/LIPIcs.CSL.2022.32}
  {\path{doi:10.4230/LIPIcs.CSL.2022.32}}.

\bibitem{Nguyen2022a}
L{\^{e}}~Th{\`{a}}nh~Dũng Nguy{\^{e}}n and Lutz Stra{\ss}burger.
\newblock {A System of Interaction and Structure III: The Complexity of BV and
  Pomset Logic}.
\newblock {\em Logical Methods in Computer Science}, Volume 19,(4):25:1--25:60,
  dec 2023.
\newblock \href {https://doi.org/10.46298/lmcs-19(4:25)2023}
  {\path{doi:10.46298/lmcs-19(4:25)2023}}.

\bibitem{Oreshkov2012}
Ognyan Oreshkov, Fabio Costa, and {\v{C}}aslav Brukner.
\newblock {Quantum correlations with no causal order}.
\newblock {\em Nature Communications}, 3(1):1--8, oct 2012.
\newblock \href
  {https://doi.org/10.1038/ncomms2076} {\path{doi:10.1038/ncomms2076}}.

\bibitem{Pearl}
Judea Pearl.
\newblock {\em {Causality: Models, reasoning and inference}}.
\newblock Cambridge University Press, 2nd edition.

\bibitem{Popescu1994}
Sandu Popescu and Daniel Rohrlich.
\newblock {Quantum nonlocality as an axiom}.
\newblock {\em Foundations of Physics}, 24(3):379--385, mar 1994.
\newblock \href {https://doi.org/10.1007/BF02058098}
  {\path{doi:10.1007/BF02058098}}.

\bibitem{Retore1994}
Christian Retor{\'{e}}.
\newblock {On the relation between coherence semantics and multiplicative proof
  nets .}
\newblock Technical Report RR-2430, INRIA, dec 1994.
\newblock URL: \url{https://inria.hal.science/inria-00074245}.

\bibitem{Retore1997}
Christian Retor{\'{e}}.
\newblock {Pomset logic: A non-commutative extension of classical linear
  logic}.
\newblock In {\em Lecture Notes in Computer Science (including subseries
  Lecture Notes in Artificial Intelligence and Lecture Notes in
  Bioinformatics)}, volume 1210, pages 300--318, 1997.
\newblock \href
  {https://doi.org/10.1007/3-540-62688-3_43}
  {\path{doi:10.1007/3-540-62688-3_43}}.

\bibitem{Retore2003}
Christian Retor{\'{e}}.
\newblock {Handsome proof-nets: Perfect matchings and cographs}.
\newblock In {\em Theoretical Computer Science}, volume 294, pages 473--488,
  2003.
\newblock \href
  {https://doi.org/10.1016/S0304-3975(01)00175-X}
  {\path{doi:10.1016/S0304-3975(01)00175-X}}.

\bibitem{Retore2020}
Christian Retor{\'{e}}.
\newblock {Pomset logic: a logical and grammatical alternative to the Lambek
  calculus}.
\newblock jan 2020.
\newblock \href
  {http://arxiv.org/abs/2001.02155} {\path{arXiv:2001.02155}}.

\bibitem{Shapiro2022}
Brandon~T. Shapiro and David~I. Spivak.
\newblock {Duoidal Structures for Compositional Dependence}.
\newblock oct 2022.
\newblock \href
  {http://arxiv.org/abs/2210.01962} {\path{arXiv:2210.01962}}.

\bibitem{Simmons2022}
Will Simmons and Aleks Kissinger.
\newblock {Higher-Order Causal Theories Are Models of BV-Logic}.
\newblock In Stefan Szeider, Robert Ganian, and Alexandra Silva, editors, {\em
  47th International Symposium on Mathematical Foundations of Computer Science
  (MFCS 2022)}, volume 241 of {\em Leibniz International Proceedings in
  Informatics (LIPIcs)}, pages 80:1----80:14, Dagstuhl, Germany, 2022. Schloss
  Dagstuhl -- Leibniz-Zentrum f{\"{u}}r Informatik.
\newblock \href {https://doi.org/10.4230/LIPIcs.MFCS.2022.80}
  {\path{doi:10.4230/LIPIcs.MFCS.2022.80}}.

\bibitem{Slavnov2019}
Sergey Slavnov.
\newblock {On noncommutative extensions of linear logic}.
\newblock {\em Logical Methods in Computer Science}, 15(3):25, 2019.
\newblock \href
  {https://doi.org/10.23638/LMCS-15(3:30)2019}
  {\path{doi:10.23638/LMCS-15(3:30)2019}}.

\bibitem{Tiu2006}
Alwen Tiu.
\newblock {A system of interaction and structure II: The need for deep
  inference}.
\newblock {\em Logical Methods in Computer Science}, 2(2):1--24, 2006.
\newblock \href
  {https://doi.org/10.2168/LMCS-2(2:4)2006}
  {\path{doi:10.2168/LMCS-2(2:4)2006}}.

\bibitem{Wilson2022a}
Matt Wilson and Giulio Chiribella.
\newblock {Free Polycategories for Unitary Supermaps of Arbitrary Dimension}.
\newblock 2022.
\newblock \href
  {http://arxiv.org/abs/2207.09180} {\path{arXiv:2207.09180}}.

\bibitem{Wilson2022}
Matt Wilson, Giulio Chiribella, and Aleks Kissinger.
\newblock {Quantum Supermaps are Characterized by Locality}.
\newblock may 2022.
\newblock \href
  {http://arxiv.org/abs/2205.09844} {\path{arXiv:2205.09844}}.

\end{thebibliography}
\bibliographystyle{plainurl}

\appendix

\section{Essentials of $\mll$, $\bv$, and Pomset Logic}\label{sec:logic_primer}

To account for readers from a range of backgrounds, this appendix will serve as a minimal primer for these logics to see the significance towards the content of this paper.

Formal logics aim to provide a canonical syntax for mathematical reasoning, where a goal statement is shown to be provable in the logic by giving a derivation from simple elementary rules or some other proof object which can be verified by some simple consistency condition. The choices of what grammar to build terms from, how one builds statements from terms (e.g. one- or two-sided, collections of terms as sets/multisets/ordered lists, etc.), and which rules to permit or consistency condition to enforce gives the logic its flavour, determining the structure of proofs and what semantic models exist for it.

Linear logics expose the hidden assumption of propositional logic that everything is copyable: a proof of the implication $F \Rightarrow G$ may use its premise $F$ multiple (or even zero) times to derive the conclusion $G$. The rules are ``resource-aware'', ensuring that each premise is used precisely once, in line with the reading of linear implication $F \multimap G$ as ``one use of $F$ can generate a use of $G$''. Sequents use multisets $\Gamma, \Delta$ of terms to explicitly track the quantities of any duplicate terms. The operators of linear logic are typically grouped into the following classes:
\begin{itemize}
\item Duality $\left( - \right)^\bot$ as the correspondent of negation, which is commonly interpreted in models as the duality between producing or consuming a use of a term. Syntactic De Morgan equations can always push duals inside formulae to yield a negation normal form (where duals are only applied on atoms).
\item Multiplicatives $\left\{ \otimes, \parr, \multimap \right\}$ introduce multiple uses of terms in parallel, possibly with some connection or relationship between them.
\item Additives $\left\{ \times, \oplus \right\}$ capture a single-use choice between two terms, distinguished as external choice (both are provided and the choice selected at consumption) versus internal choice (only one is provided, so a consumer must be able to handle either).
\item Exponentials $\left\{ !, ? \right\}$ annotate terms to reintroduce resources that can be used multiple times.
\item Units $\left\{ 1, \bot, 0, \top \right\}$ for the multiplicative and additive operators represent trivial or degenerate cases. Many extensions of linear logic will unify some of these units.
\end{itemize}
We retain the usual distinction between classical and intuitionistic variants of the logic (whether negation/duality is involutive $F^{\bot\bot} \multimap F$). For this presentation, we will work with the classical logics that are reflected by the models in the main text.

\subsection{Sequent Calculus for $\mll$}

Multiplicative Linear Logic ($\mll$) is the fragment just concerning the multiplicative operators and their units. In the classical variant, linear implication $F \multimap G$ is syntactically identified with $F^\bot \parr G$. The standard presentation is as a sequent calculus, inductively defined by the following rules:

\noindent\begin{minipage}{\textwidth}
\begin{minipage}[b]{.5\textwidth}
\begin{prooftree}
\AxiomC{}
\LeftLabel{$\mathrm{Ax}$}
\UnaryInfC{$\vdash A^\bot, A$}
\end{prooftree}
\end{minipage}\hfill
\begin{minipage}[b]{.5\textwidth}
\begin{prooftree}
\AxiomC{$\vdash \Gamma, F$}
\AxiomC{$\vdash F^\bot, \Delta$}
\LeftLabel{$\mathrm{Cut}$}
\BinaryInfC{$\vdash \Gamma, \Delta$}
\end{prooftree}
\end{minipage}
\end{minipage}

\vspace{1em}

\noindent\begin{minipage}{\textwidth}
\begin{minipage}[b]{.5\textwidth}
\begin{prooftree}
\AxiomC{$\vdash \Gamma, F$}
\AxiomC{$\vdash \Delta, G$}
\LeftLabel{$\otimes$}
\BinaryInfC{$\vdash \Gamma, \Delta, F \otimes G$}
\end{prooftree}
\end{minipage}\hfill
\begin{minipage}[b]{.5\textwidth}
\begin{prooftree}
\AxiomC{$\vdash \Gamma, F, G$}
\LeftLabel{$\parr$}
\UnaryInfC{$\vdash \Gamma, F \parr G$}
\end{prooftree}
\end{minipage}
\end{minipage}

\vspace{1em}

\noindent\begin{minipage}{\textwidth}
\begin{minipage}[b]{.5\textwidth}
\begin{prooftree}
\AxiomC{}
\LeftLabel{$1$}
\UnaryInfC{$\vdash 1$}
\end{prooftree}
\end{minipage}\hfill
\begin{minipage}[b]{.5\textwidth}
\begin{prooftree}
\AxiomC{$\vdash \Gamma$}
\LeftLabel{$\bot$}
\UnaryInfC{$\vdash \Gamma, \bot$}
\end{prooftree}
\end{minipage}
\end{minipage}
\vspace{1em}

There exists a cut-elimination procedure which, given a derivation of a sequent using the $\mathrm{Cut}$ rule, yields a derivation with the same conclusion that does not use $\mathrm{Cut}$ (a \textit{cut-free} derivation).

\begin{example}
Observe that the following derivation tree yields one of the properties needed for associativity of $\otimes$:

\begin{prooftree}
\AxiomC{}
\LeftLabel{$\mathrm{Ax}$}
\UnaryInfC{$\vdash A^\bot, A$}
\AxiomC{}
\LeftLabel{$\mathrm{Ax}$}
\UnaryInfC{$\vdash B^\bot, B$}
\LeftLabel{$\otimes$}
\BinaryInfC{$\vdash A^\bot, B^\bot, A \otimes B$}
\AxiomC{}
\LeftLabel{$\mathrm{Ax}$}
\UnaryInfC{$\vdash C^\bot, C$}
\LeftLabel{$\otimes$}
\BinaryInfC{$\vdash A^\bot, B^\bot, C^\bot, (A \otimes B) \otimes C$}
\LeftLabel{$\parr$}
\UnaryInfC{$\vdash A^\bot, B^\bot \parr C^\bot, (A \otimes B) \otimes C$}
\LeftLabel{$\parr$}
\UnaryInfC{$\vdash A^\bot \parr (B^\bot \parr C^\bot), (A \otimes B) \otimes C$}
\LeftLabel{$\parr$}
\UnaryInfC{$\vdash A \otimes (B \otimes C) \multimap (A \otimes B) \otimes C$}
\end{prooftree}

However, if we search for a cut-free derivation of $A \multimap A \otimes \bot$ (i.e. using the unit of $\parr$ with $\otimes$ instead) the syntax-directed nature of the sequents forces any solution to finish with the following rules:
\begin{prooftree}
\AxiomC{}
\LeftLabel{$\mathrm{Ax}$}
\UnaryInfC{$\vdash A^\bot, A$}
\AxiomC{$\vdots$}
\UnaryInfC{$\vdash$}
\LeftLabel{$\bot$}
\UnaryInfC{$\vdash \bot$}
\LeftLabel{$\otimes$}
\BinaryInfC{$\vdash A^\bot, A \otimes \bot$}
\LeftLabel{$\parr$}
\UnaryInfC{$\vdash A \multimap A \otimes \bot$}
\end{prooftree}
But it is obvious from the rules that no cut-free proof exists for the empty sequent, so we conclude $\not\vdash A \multimap A \otimes \bot$.
\end{example}

Some settings will consider adding additional ``mixing'' rules to the logic, giving rise to $\mllmix$:

\vspace{1em}
\noindent\begin{minipage}{\textwidth}
\begin{minipage}[b]{.5\textwidth}
\begin{prooftree}
\AxiomC{}
\LeftLabel{$\mathrm{Mix1}$}
\UnaryInfC{$\vdash$}
\end{prooftree}
\end{minipage}\hfill
\begin{minipage}[b]{.5\textwidth}
\begin{prooftree}
\AxiomC{$\vdash \Gamma$}
\AxiomC{$\vdash \Delta$}
\LeftLabel{$\mathrm{Mix2}$}
\BinaryInfC{$\vdash \Gamma, \Delta$}
\end{prooftree}
\end{minipage}
\end{minipage}
\vspace{1em}

These are equivalent to adding axioms for $\vdash \bot, \bot$ (which gives $\vdash 1 \multimap \bot$) and $\vdash 1, 1$ (giving $\vdash \bot \multimap 1$) respectively, making the two units interchangeable \cite{Cockett1997a}. $\mathrm{Mix2}$ is also equivalent to the binary case $\vdash A^\bot, B^\bot, A, B$, which can be summarised as $\otimes$ ``embedding into'' $\parr$ as $\vdash A \otimes B \multimap A \parr B$.

\subsection{Proof-nets for $\mll$}

The primary alternative to sequent derivations as a presentation of linear logic proofs is \textit{proof-nets}. These consist of graphical objects annotated with terms where correctness of the proof is determined by some global property of the graph, in contrast with the local correctness of each rule in a derivation tree. One can translate proofs back and forth, with each derivation tree yielding a unique proof-net but each proof-net matching multiple derivation trees.

To specify a proof-net framework, we first define \textit{proof-structures} as a particular grammar of graphical objects under consideration which are composites of elementary graphs called \textit{links}. For simplicity of this introduction of proof-nets, we will not be including units.

\begin{definition}
\textit{Links} for $\mll$ are graphs of the following forms, relating some premises to some conclusions:
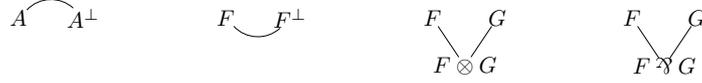
\begin{align*}
\begin{tikzpicture}
	\begin{pgfonlayer}{nodelayer}
		\node [style=none] (0) at (0, 0) {$A$};
		\node [style=none] (1) at (2, 0) {$A^\bot$};
	\end{pgfonlayer}
	\begin{pgfonlayer}{edgelayer}
		\draw [bend left=45] (0) to (1);
	\end{pgfonlayer}
\end{tikzpicture}
 &&
\begin{tikzpicture}
	\begin{pgfonlayer}{nodelayer}
		\node [style=none] (0) at (0, 0) {$F$};
		\node [style=none] (1) at (2, 0) {$F^\bot$};
	\end{pgfonlayer}
	\begin{pgfonlayer}{edgelayer}
		\draw [bend right=45] (0) to (1);
	\end{pgfonlayer}
\end{tikzpicture}
 &&
\begin{tikzpicture}
	\begin{pgfonlayer}{nodelayer}
		\node [style=none] (0) at (0, 0) {$F$};
		\node [style=none] (1) at (2, 0) {$G$};
		\node [style=none] (2) at (1, -1.5) {$F \otimes G$};
	\end{pgfonlayer}
	\begin{pgfonlayer}{edgelayer}
		\draw (0) to (2);
		\draw (2) to (1);
	\end{pgfonlayer}
\end{tikzpicture}
 &&
\begin{tikzpicture}
	\begin{pgfonlayer}{nodelayer}
		\node [style=none] (0) at (0, 0) {$F$};
		\node [style=none] (1) at (2, 0) {$G$};
		\node [style=none] (2) at (1, -1.5) {$F \parr G$};
	\end{pgfonlayer}
	\begin{pgfonlayer}{edgelayer}
		\draw (0) to (2);
		\draw (2) to (1);
	\end{pgfonlayer}
\end{tikzpicture}

\end{align*}
\begin{itemize}
\item Axiom links with no premises and a pair of conclusions $A$, $A^\bot$.
\item Cut links with a pair of premises $F$, $F^\bot$ and no conclusions.
\item Tensor links with a pair of premises $F, G$ and conclusion $F \otimes G$.
\item Par links with a pair of premises $F, G$ and conclusion $F \parr G$.
\end{itemize}
\end{definition}

\begin{definition}
A \textit{proof-structure} is a graph composed of links such that every occurrence of a formula (i.e. vertex) is a conclusion of exactly one link and a premise of at most one link. The \textit{conclusions} of the proof-structure are the formulae that are not the premises of any link.
\end{definition}

Proof-nets are those proof-structures that correspond to ``correct'' proofs in the sense that they satisfy a particular correctness criterion. The Danos-Regnier criterion derives a collection of variants of the proof-structure's graph, called \textit{switching graphs}, and asks each is acyclic.

\begin{definition}
A \textit{switching} for a proof-structure $P$ is a choice of one edge from each par link. The corresponding \textit{switching graph} is a copy of $P$ in which we omit the edges from par links that were not selected.
\end{definition}

\begin{definition}
A proof-structure is a $\mll$ proof-net when every switching graph is connected and acyclic. It is a $\mllmix$ proof-net when every switching graph is acyclic but need not be connected.
\end{definition}

Proof-nets also admit cut-elimination. The syntax-directed nature of the links means that cut-free proof-structures will resemble the syntax tree of the conclusion with some choice of matching for the axiom links.

\begin{example}
Consider the following proof-structure:
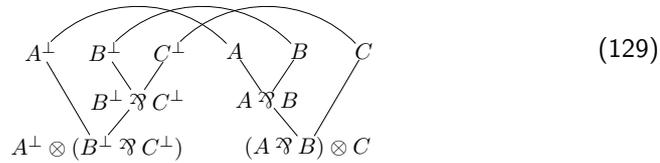
\begin{equation}
\input{./proof_structure_fail.tikz}
\end{equation}
By looking at the switching graph where we take the paths to $C^\bot$ and $A$, we can construct a cycle and conclude this is not a valid proof-net.

On the other hand, the following proof-structure representing the mixing rule is a proof-net for $\mllmix$ but not $\mll$ as every switching graph is disconnected.
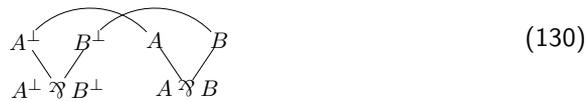
\begin{equation}
\input{./proof_structure_mix.tikz}
\end{equation}
\end{example}

There are multiple, equivalent notions of proof-nets for $\mll$ or $\mllmix$, each of which can provide a unique intuition for the logics. For Girard's long-trip condition, switchings define a route for traversing over the proof-structure and requires that the entire proof is traversed. Topological criteria \cite{Mellies2010} can be very visually intuitive. Handsome proof-nets \cite{Retore2003} even forgo switchings, instead building proof-structures as $\mathrm{R\&B}$-graphs and obtaining a more canonical representation of proofs in which associativity of operators is a simple equality.

\subsection{$\bv$ and Pomset Logic}

Many settings will interpret $\otimes$ and $\parr$ as a form of parallel composition, but may also carry a notion of time or sequential composition. $\bv$ and pomset are two candidate solutions to this, where both extend $\mllmix$ with a self-dual non-commutative operator $<$, reading $F < G$ as ``$F$ before $G$''. They were originally speculated to express the same logic, but pomset contains strictly more theorems than $\bv$ \cite{Nguyen2022a}. They differ in presentation, with $\bv$ admitting a sequent calculus but no known proof-net formulation, and conversely pomset is defined by proof-nets with no sequent calculus (the best approximation is via decorated sequents \cite{Slavnov2019}, which annotate derivations with a substantial amount of additional information).

The deductive system of $\bv$ employs \textit{deep inference}, where inference rules can be applied at any level of nesting within expressions, in contrast to the rules of $\mll$ which all apply at the top level of syntax trees. This is crucial to the logic, and placing any bound on the nesting level of applying rules strictly reduces the number of theorems \cite{Tiu2006}. This can either be expressed by describing the rules with respect to term contexts, or by visually nesting derivation trees as we will see in Example \ref{ex:bv}.

\noindent\begin{minipage}{\textwidth}
\begin{minipage}[b]{.5\textwidth}
\begin{prooftree}
\AxiomC{}
\LeftLabel{$I\downarrow$}
\UnaryInfC{$I$}
\end{prooftree}
\end{minipage}\hfill
\begin{minipage}[b]{.5\textwidth}
\begin{prooftree}
\AxiomC{$I$}
\LeftLabel{$\mathrm{ai}\downarrow$}
\UnaryInfC{$A^\bot \parr A$}
\end{prooftree}
\end{minipage}
\end{minipage}

\vspace{1em}

\noindent\begin{minipage}{\textwidth}
\begin{minipage}[b]{.5\textwidth}
\begin{prooftree}
\AxiomC{$(F \parr G) < (H \parr K)$}
\LeftLabel{$\mathrm{q}\downarrow$}
\UnaryInfC{$(F < H) \parr (G < K)$}
\end{prooftree}
\end{minipage}\hfill
\begin{minipage}[b]{.5\textwidth}
\begin{prooftree}
\AxiomC{$(F \parr G) \otimes H$}
\LeftLabel{$\mathrm{s}$}
\UnaryInfC{$F \parr (G \otimes H)$}
\end{prooftree}
\end{minipage}
\end{minipage}
\vspace{1em}

Instead of commutativity, associativity, and unitality ($I$ is a unit for all of $\left\{ \otimes, < \parr \right\}$) being properties which can be proved by the rules, they are permitted via a syntactic equivalence relation $\equiv$ which is a congruence so can similarly be applied at any level within terms.

\begin{example}\label{ex:bv}
The following example $\bv$ proof of $(A \otimes B) < C \multimap A < (B \parr C)$ demonstrates the nested application of rules within contexts. Dashed lines are applications of the syntactic equivalence.

\newsavebox{\bva}
\sbox{\bva}{
\AxiomC{$A^\bot < B^\bot$}
\dashedLine
\UnaryInfC{$\left( A^\bot \parr I \right) < \left( I \parr B^\bot \right)$}
\LeftLabel{$\mathrm{q}\downarrow$}
\UnaryInfC{$\left( A^\bot < I \right) \parr \left( I < B^\bot \right)$}
\DisplayProof
}
\newsavebox{\bvb}
\sbox{\bvb}{
\AxiomC{$I$}
\LeftLabel{$\mathrm{ai}\downarrow$}
\UnaryInfC{$A^\bot \parr A$}
\DisplayProof
}
\newsavebox{\bvc}
\sbox{\bvc}{
\AxiomC{$I$}
\LeftLabel{$\mathrm{ai}\downarrow$}
\UnaryInfC{$B^\bot \parr B$}
\dashedLine
\UnaryInfC{$B^\bot \parr B \parr I$}
\DisplayProof
}
\newsavebox{\bvd}
\sbox{\bvd}{
\AxiomC{$I$}
\LeftLabel{$\mathrm{ai}\downarrow$}
\UnaryInfC{$C^\bot \parr C$}
\DisplayProof
}
\newsavebox{\bvf}
\sbox{\bvf}{
\AxiomC{$\left( B^\bot \parr B \right) < \left( C^\bot \parr I \right)$}
\LeftLabel{$\mathrm{q}\downarrow$}
\UnaryInfC{$\left( B^\bot < C^\bot \right) \parr \left( B < I \right)$}
\DisplayProof
}
\newsavebox{\bve}
\sbox{\bve}{
\AxiomC{$\left( \usebox{\bvc} \right) < \left( \usebox{\bvd} \right)$}
\LeftLabel{$\mathrm{q}\downarrow$}
\UnaryInfC{$\left( \left( B^\bot \parr B \right) < C^\bot \right) \parr \left( I < C \right)$}
\dashedLine
\UnaryInfC{$\left( \usebox{\bvf} \right) \parr C$}
\dashedLine
\UnaryInfC{$\left( B^\bot < C^\bot \right) \parr B \parr C$}
\DisplayProof
}
\begin{scprooftree}{0.9}
\AxiomC{}
\LeftLabel{$I\downarrow$}
\UnaryInfC{$I$}
\dashedLine
\UnaryInfC{$\left( \usebox{\bvb} \right) < \left( \usebox{\bve} \right)$}
\LeftLabel{$\mathrm{q}\downarrow$}
\UnaryInfC{$\left( \left( \usebox{\bva} \right) < C^\bot \right) \parr \left( A < \left( B \parr C \right) \right)$}
\dashedLine
\UnaryInfC{$\left( \left( A^\bot \parr B^\bot \right) < C^\bot \right) \parr \left( A < \left( B \parr C \right) \right)$}
\end{scprooftree}
\end{example}

The combination of sequential and parallel composition is not enough to recover all partial overs or graphs over the atoms. A recent extension called $\mathrm{GV}$ adds an infinite number of new operators based on prime directed and undirected graphs, from which we can obtain every mixed graph from its modular decomposition. We refer the reader to Acclavio et al.~\cite{Acclavio2022} for further details.

We finally turn to pomset logic, which are built on $\mathrm{R\&B}$-graphs, i.e. graphs with red and blue edges. We still build proof-structures from links, with the correctness of a proof-net becoming the absence of an \textit{alternating elementary cycle}. The alternating colour requirement permits restricting the paths through each link in the same way as switchings.

\begin{definition}
A \textit{$\mathrm{R\&B}$-graph} $G = (V, R, B)$ is an edge-bi-coloured graph where:
\begin{itemize}
\item $V$ is the set of vertices;
\item $R \subseteq V \times V$ is the set of (directed and undirected) red edges, expressed as an irreflexive relation;
\item $B \subseteq V \times V$ is the set of (undirected) blue edges, expressed as an irreflexive, symmetric relation forming a perfect matching of the vertices $\forall x \in V . \exists ! y \in V . (x, y) \in B$.
\end{itemize}
\end{definition}

\begin{definition}
A \textit{$\mathrm{R\&B}$ proof-structure} is a $\mathrm{R\&B}$-graph composed of the following kinds of links:
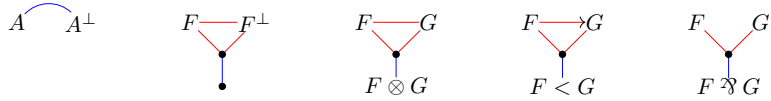
\begin{align*}
\begin{tikzpicture}
	\begin{pgfonlayer}{nodelayer}
		\node [style=none] (0) at (2, 0) {$A^\bot$};
		\node [style=none] (1) at (0, 0) {$A$};
	\end{pgfonlayer}
	\begin{pgfonlayer}{edgelayer}
		\draw [style=blue, bend left=45] (1) to (0);
	\end{pgfonlayer}
\end{tikzpicture}
 && \input{./pom_link_cut.tikz} && \input{./pom_link_tensor.tikz} && \input{./pom_link_seq.tikz} && \input{./pom_link_par.tikz}
\end{align*}
The conclusions of the proof-structure may also be connected by directed red edges according to a series-parallel partial order.
\end{definition}

\begin{definition}
A \textit{pomset proof-net} is a $\mathrm{R\&B}$ proof-structure containing no alternating elementary cycle (a cycle with no repeated vertices except for the first and last, where successive edges alternate colours).
\end{definition}

The additional edges beyond the conclusions elevate the multiset structure of $\mll$ sequents to partially-ordered multisets, hence the name of the logic. The series-parallel restriction means we can generally combine the conclusions using the $\parr$ and $<$ links to form an equivalent proof-structure with a single conclusion, reminiscent of how $\parr$ naturally combines conclusions in $\mll$.

From the construction of the links, we can see that any path from one link to another will change colour at the transition: on each vertex labelled with a formula, the unique incident blue edge belongs to the link where this vertex is a conclusion and any red edges (if they exist) belong to the link where it is a premise. Therefore the alternating colour requirement only constrains the possible paths within each individual link; most notably there is no alternating path between the premises of a $\parr$ link, and only in one direction between the premises of a $<$ link. On the $<$-free fragment, these precisely match the connectivities induced by switchings on $\mll$ proof-structures, so pomset is a conservative extension of $\mllmix$ since we can faithfully translate between cyclic switchings in an $\mll$ proof-structure and alternating cycles in a $\mathrm{R\&B}$ proof-structure.

Deciding whether a $\mathrm{R\&B}$ proof-structure is a pomset proof-net is $\mathrm{coNP}$-complete \cite{Nguyen2022a} - it is easy to verify that a given alternating cycle witnesses incorrectness, but searching for such cycles is hard. Determining whether a proof-net exists for a given conclusion is even harder at $\Sigma_2^p$-complete ($\Sigma_2^p = \mathrm{NP}^{\mathrm{NP}}$ is $\mathrm{NP}$ extended with an oracle for $\mathrm{NP}$). On the other hand, verifying $\bv$ proofs takes time linear in the size of the proof, and provability of a formula is $\mathrm{NP}$-complete \cite{Kahramanogullari2008}. Pomset is even hard compared to other proof-net formulations, where proof-nets can be verified for $\mll$ or $\mllmix$ in linear time \cite{Murawski2006,Nguyen2020}.

\section{Proofs from Section \ref{sec:caus_construction}}\label{sec:caus_construction_proofs}

\begin{proposition}[Restatement of Proposition \ref{prop:complement_binary_test}]
For any morphism $f \in \catc(A, B)$ of the base category and causal objects $\mathbf{A} = \left( A, c_{\mathbf{A}} \right), \mathbf{B} = \left( B, c_{\mathbf{B}} \right) \in \ob{\caus{\catc}}$, there exists some non-zero scalar $\kappa \in \catc(I, I)$, morphism $f' \in \catc(A, B)$ and causal morphism $t_f \in \caus{\catc}\left( \mathbf{A}, \mathbf{B \parr 2} \right)$ such that $t_f = (\kappa \cdot f \otimes \tbool) + (f' \otimes \fbool)$.
\end{proposition}

\begin{proof}
Apply Proposition \ref{prop:complement} to $f$ to get $f' \in \catc(A, B)$ and $\lambda \in \catc(I, I)$ where $f + f' = \lambda \cdot \discard_A \fatsemi \maxmix_B$ which is causal up to a scalar. We then pick $t_f = \left( \theta_{\mathbf{A}} \mu_{\mathbf{B}} \lambda^{-1} \cdot f \otimes \tbool \right) + \left( \theta_{\mathbf{A}} \mu_{\mathbf{B}} \lambda^{-1} \cdot f' \otimes \fbool \right)$.

To check $t_f \in \caus{\catc}\left( \mathbf{A}, \mathbf{B \parr 2} \right)$, we consider an arbitrary effect $\pi \otimes \ndiscard_{\mathbf{2}} \in c_{\mathbf{B \parr 2}}^*$.
\begin{equation}
\begin{split}
t_f \fatsemi \left( \pi \otimes \ndiscard_{\mathbf{2}} \right) &= \theta_{\mathbf{A}} \mu_{\mathbf{B}} \lambda^{-1} \cdot \left( f \otimes \tbool \fatsemi \pi \otimes \ndiscard_{\mathbf{2}} + f' \otimes \fbool \fatsemi \pi \otimes \ndiscard_{\mathbf{2}} \right) \\
&= \theta_{\mathbf{A}} \mu_{\mathbf{B}} \lambda^{-1} \cdot \left( f \fatsemi \pi + f' \fatsemi \pi \right) \\
&= \theta_{\mathbf{A}} \mu_{\mathbf{B}} \cdot \discard_A \fatsemi \maxmix_B \fatsemi \pi \\
&= \ndiscard_{\mathbf{A}} \fatsemi \nmaxmix_{\mathbf{B}} \fatsemi \pi \\
&= \ndiscard_{\mathbf{A}} \in c_{\mathbf{A}}^* \\
\end{split}
\end{equation}
\end{proof}

\begin{proposition}\label{prop:set_compatible_by_intersection}
Given objects $\mathbf{A}, \mathbf{A'} \in \ob{\caus{\catc}}$ with the same carrier object $A \in \ob{\catc}$, then the normalisation scalars agree $\mu_{\mathbf{A}} = \mu_{\mathbf{A'}}$, $\theta_{\mathbf{A}} = \theta_{\mathbf{A'}}$ iff $c_{\mathbf{A}} \cap c_{\mathbf{A'}} \neq \emptyset$.
\end{proposition}

\begin{proof}
By flatness, $\nmaxmix_{\mathbf{A}} \in c_{\mathbf{A}}$ and $\ndiscard_{\mathbf{A}} \in c_{\mathbf{A}}^*$.
\begin{equation}
\mu_{\mathbf{A}} \theta_{\mathbf{A}} d_A = \mu_{\mathbf{A}} \theta_{\mathbf{A}} \cdot \maxmix_A \fatsemi \discard_A = \nmaxmix_{\mathbf{A}} \fatsemi \ndiscard_{\mathbf{A}} = \id_I
\end{equation}
Similarly $\mu_{\mathbf{A'}} \theta_{\mathbf{A'}} d_A = \id_I$, so $\mu_{\mathbf{A}} = \mu_{\mathbf{A'}}$ iff $\theta_{\mathbf{A}} = \theta_{\mathbf{A'}}$.

$\Longrightarrow$: $c_{\mathbf{A}} \ni \nmaxmix_{\mathbf{A}} = \mu_{\mathbf{A}} \cdot \maxmix_A = \mu_{\mathbf{A'}} \cdot \maxmix_A = \nmaxmix_{\mathbf{A'}} \in c_{\mathbf{A'}}$.

$\Longleftarrow$: For any $\rho \in c_{\mathbf{A}}$, we have $\rho \fatsemi \ndiscard_{\mathbf{A}} = \id_I$. Then
\begin{equation}
\begin{split}
\rho \fatsemi \ndiscard_{\mathbf{A'}} &= \theta_{\mathbf{A'}} \cdot \rho \fatsemi \discard_A \\
&= \theta_{\mathbf{A}}^{-1} \theta_{\mathbf{A'}} \cdot \rho \fatsemi \ndiscard_{\mathbf{A}} \\
&= \theta_{\mathbf{A}}^{-1} \theta_{\mathbf{A'}}
\end{split}
\end{equation}
So for any $\rho \in c_{\mathbf{A}} \cap c_{\mathbf{A'}}$ (if it exists), this would also be $\id_I$, i.e. $\theta_{\mathbf{A}} = \theta_{\mathbf{A'}}$.
\end{proof}

\begin{lemma}[Restatement of Lemma \ref{lemma:union_intersection_rules}]
Given objects $\mathbf{A}, \mathbf{A'}, \mathbf{B} \in \ob{\caus{\catc}}$ with $\mathbf{A} \pitchfork \mathbf{A'}$, we have the following for each $\Box \in \left\{ \otimes, \parr, <, >, \times, \oplus \right\}$ and $\lozenge \in \left\{ \cup, \cap \right\}$:
\begin{align*}
\mathbf{(A \cup A')^*} &= \mathbf{A^* \cap A'^*} \\
\mathbf{A \otimes B} &= \mathbf{(A < B) \cap (A > B)} \\
\mathbf{A \parr B} &= \mathbf{(A < B) \cup (A > B)} \\
\mathbf{(A \lozenge A') \Box B} &= \mathbf{(A \Box B) \lozenge (A' \Box B)}
\end{align*}
\end{lemma}

\begin{proof}
\begin{itemize}
\item $\mathbf{\left( A \cup A' \right)^*} = \mathbf{A^* \cap A'^*}$:
\begin{equation}
\begin{split}
\pi \in c_{\mathbf{(A \cup A')^*}} &\Leftrightarrow \pi \in (c_{\mathbf{A}} \cup c_{\mathbf{A'}})^* \\
&\Leftrightarrow (\forall \rho \in c_{\mathbf{A}} . \rho \fatsemi \pi = \id_I) \wedge (\forall \rho' \in c_{\mathbf{A'}} . \rho' \fatsemi \pi = \id_I) \\
&\Leftrightarrow \pi \in c_{\mathbf{A}}^* \cap c_{\mathbf{A'}}^* \\
&\Leftrightarrow \pi \in c_{\mathbf{A^* \cap A'^*}}
\end{split}
\end{equation}
\item $\mathbf{A \otimes B} = \mathbf{\left( A < B \right) \cap \left( A > B \right)}$: See \cite[Theorem 35]{Simmons2022}.
\item $\mathbf{A \parr B} = \mathbf{\left( A < B \right) \cup \left( A > B \right)}$: Dual to above.
\item $\mathbf{\left( A \cap A' \right) < B} = \mathbf{\left( A < B \right) \cap \left( A' < B \right)}$:
\begin{equation}
\begin{split}
\rho \in c_{\mathbf{(A \cap A') < B}} \Leftrightarrow& \forall \pi \in c_{\mathbf{B}}^* . \rho \fatsemi \left( \id_A \otimes \pi \right) = \rho \fatsemi \left( \id_A \otimes \ndiscard_{\mathbf{B}} \right) \in c_{\mathbf{A \cap A'}} \\
\Leftrightarrow& \forall \pi \in c_{\mathbf{B}}^* . \rho \fatsemi \left( \id_A \otimes \pi \right) = \rho \fatsemi \left( \id_A \otimes \ndiscard_{\mathbf{B}} \right) \in c_{\mathbf{A}} \\ &\wedge \forall \pi \in c_{\mathbf{B}}^* . \rho \fatsemi \left( \id_A \otimes \pi \right) = \rho \fatsemi \left( \id_A \otimes \ndiscard_{\mathbf{B}} \right) \in c_{\mathbf{A'}} \\
\Leftrightarrow& \rho \in c_{\mathbf{A < B}} \wedge \rho \in c_{\mathbf{A' < B}} \\
\Leftrightarrow& \rho \in c_{\mathbf{(A < B) \cap (A' < B)}}
\end{split}
\end{equation}
\item $\mathbf{\left( A \cup A' \right) < B} = \mathbf{\left( A < B \right) \cup \left( A' < B \right)}$: Dual to above.
\item $\mathbf{\left( A \cap A' \right) > B} = \mathbf{\left( A > B \right) \cap \left( A' > B \right)}$: Note that equality of carrier objects and normalisation scalars give $\ndiscard_{\mathbf{A \cap A'}} = \ndiscard_{\mathbf{A}} = \ndiscard_{\mathbf{A'}}$.
\begin{equation}
\begin{split}
\rho \in c_{\mathbf{(A \cap A') > B}} \Leftrightarrow& \forall \pi \in c_{\mathbf(A \cap A')}^* . \\ & \rho \fatsemi \left( \pi \otimes \id_B \right) = \rho \fatsemi \left( \ndiscard_{\mathbf{A \cap A'}} \otimes \id_B \right) \in c_{\mathbf{B}} \\
\Leftrightarrow& \forall \pi \in \left( c_{\mathbf{A}}^* \cup c_{\mathbf{A}}^* \right)^{**} . \\ & \rho \fatsemi \left( \pi \otimes \id_B \right) = \rho \fatsemi \left( \ndiscard_{\mathbf{A \cap A'}} \otimes \id_B \right) \in c_{\mathbf{B}} \\
\Leftrightarrow& \forall \pi \in c_{\mathbf{A}}^* \cup c_{\mathbf{A}}^* . \\ & \rho \fatsemi \left( \pi \otimes \id_B \right) = \rho \fatsemi \left( \ndiscard_{\mathbf{A \cap A'}} \otimes \id_B \right) \in c_{\mathbf{B}} \\
\Leftrightarrow& \forall \pi \in c_{\mathbf{A}}^* . \rho \fatsemi \left( \pi \otimes \id_B \right) = \rho \fatsemi \left( \ndiscard_{\mathbf{A}} \otimes \id_B \right) \in c_{\mathbf{B}} \\ & \wedge \forall \pi \in c_{\mathbf{A'}}^* . \rho \fatsemi \left( \pi \otimes \id_B \right) = \rho \fatsemi \left( \ndiscard_{\mathbf{A'}} \otimes \id_B \right) \in c_{\mathbf{B}} \\
\Leftrightarrow& \rho \in c_{\mathbf{A > B}} \wedge \rho \in c_{\mathbf{A' > B}} \\
\Leftrightarrow& \rho \in c_{\mathbf{(A > B) \cap (A' > B)}}
\end{split}
\end{equation}
\item $\mathbf{\left( A \cup A' \right) > B} = \mathbf{\left( A > B \right) \cup \left( A' > B \right)}$: Dual to above.
\item $\mathbf{\left( A \cap A' \right) \otimes B} = \mathbf{\left( A \otimes B \right) \cap \left( A' \otimes B \right)}$:
\begin{equation}
\begin{split}
\mathbf{(A \cap A') \otimes B} &= \mathbf{((A \cap A') < B) \cap ((A \cap A') > B)} \\
&= \mathbf{(A < B) \cap (A' < B) \cap (A > B) \cap (A' > B)} \\
&= \mathbf{(A \otimes B) \cap (A' \otimes B)}
\end{split}
\end{equation}
\item $\mathbf{\left( A \cup A' \right) \otimes B} = \mathbf{\left( A \otimes B \right) \cup \left( A' \otimes B \right)}$: This variant relies on the fact that an affine combination over affine combinations is itself an affine combination of the elementary terms.
\begin{equation}
\begin{split}
c_{\mathbf{(A \cup A') \otimes B}} =& \left\{ \rho_A \otimes \rho_B \middle| \rho_A \in \left( c_{\mathbf{A}} \cup c_{\mathbf{A'}} \right)^{**}, \rho_B \in c_{\mathbf{B}} \right\}^{**} \\
=& \left\{ \rho_A \otimes \rho_B \middle| \rho_A \in \left( c_{\mathbf{A}} \cup c_{\mathbf{A'}} \right), \rho_B \in c_{\mathbf{B}} \right\}^{**} \\
=& \left( \left\{ \rho_A \otimes \rho_B \middle| \rho_A \in c_{\mathbf{A}}, \rho_B \in c_{\mathbf{B}} \right\} \right. \\ & \cup \left. \left\{ \rho_A \otimes \rho_B \middle| \rho_A \in c_{\mathbf{A'}}, \rho_B \in c_{\mathbf{B}} \right\} \right)^{**} \\
=& \left( \left\{ \rho_A \otimes \rho_B \middle| \rho_A \in c_{\mathbf{A}}, \rho_B \in c_{\mathbf{B}} \right\}^{**} \right. \\ & \cup \left. \left\{ \rho_A \otimes \rho_B \middle| \rho_A \in c_{\mathbf{A'}}, \rho_B \in c_{\mathbf{B}} \right\}^{**} \right)^{**} \\
=& c_{\mathbf{(A \otimes B) \cup (A' \otimes B)}}
\end{split}
\end{equation}
\item $\mathbf{\left( A \cap A' \right) \parr B} = \mathbf{\left( A \parr B \right) \cap \left( A' \parr B \right)}$: Dual to $\otimes, \cup$ case.
\item $\mathbf{\left( A \cup A' \right) \parr B} = \mathbf{\left( A \parr B \right) \cup \left( A' \parr B \right)}$: Dual to $\otimes, \cap$ case.
\item $\mathbf{\left( A \cap A' \right) \times B} = \mathbf{\left( A \times B \right) \cap \left( A' \times B \right)}$:
\begin{equation}
\begin{split}
c_{\mathbf{\left( A \cap A' \right) \times B}} =& \left\{ \langle \rho_A, \rho_B \rangle \middle| \rho_A \in c_{\mathbf{A \cap A'}}, \rho_B \in c_{\mathbf{B}} \right\} \\
=& \left\{ \langle \rho_A, \rho_B \rangle \middle| \rho_A \in c_{\mathbf{A}} \cap c_{\mathbf{A'}}, \rho_B \in c_{\mathbf{B}} \right\} \\
=& \left\{ \langle \rho_A, \rho_B \rangle \middle| \rho_A \in c_{\mathbf{A}}, \rho_B \in c_{\mathbf{B}} \right\} \\ & \cap \left\{ \langle \rho_A, \rho_B \rangle \middle| \rho_A \in c_{\mathbf{A'}}, \rho_B \in c_{\mathbf{B}} \right\} \\
=& c_{\mathbf{A \times B}} \cap c_{\mathbf{A' \times B}} \\
=& c_{\mathbf{\left( A \times B \right) \cap \left( A' \times B \right)}}
\end{split}
\end{equation}
\item $\mathbf{\left( A \cup A' \right) \times B} = \mathbf{\left( A \times B \right) \cup \left( A' \times B \right)}$: we first need to be able to pull affine combinations in and out of pairings. Fixing a single $\rho_B \in c_{\mathbf{B}}$, and supposing there is some $\rho_A \in \catc(I, A)$ for which $\sum_i \alpha_i \cdot \fsub{\rho_{A,i}} \sim \fsub{\rho_A}$,
\begin{equation}
\begin{split}
\left\langle \sum_i \alpha_i \cdot \fsub{\rho_{A,i}}, \fsub{\rho_B} \right\rangle &\sim \left\langle \sum_i \alpha_i \cdot \fsub{\rho_{A,i}}, \sum_i \alpha_i \cdot \fsub{\rho_B} \right\rangle \\
&\sim \sum_i \alpha_i \cdot \fsub{\langle \rho_{A,i}, \rho_B \rangle}
\end{split}
\end{equation}
Then with a collection of $\rho_{B,i} \in c_{\mathbf{B}} = c_{\mathbf{B}}^{**}$ and some $\rho_B \in \catc(I, B)$ for which $\sum_i \alpha_i \cdot \fsub{\rho_{B,i}} = \fsub{\rho_B}$,
\begin{equation}
\begin{split}
\sum_i \alpha_i \cdot \fsub{\langle \rho_{A,i}, \rho_{B,i} \rangle} &\sim \left\langle \sum_i \alpha_i \cdot \fsub{\rho_{A,i}}, \sum_i \alpha_i \cdot \fsub{\rho_{B,i}} \right\rangle \\
&= \left\langle \sum_i \alpha_i \cdot \fsub{\rho_{A,i}}, \fsub{\rho_B} \right\rangle
\end{split}
\end{equation}
Finally, we use these to change the point at which we take affine-closure.
\begin{equation}
\begin{split}
c_{\mathbf{\left( A \cup A' \right) \times B}} =& \left\{ \langle \rho_A, \rho_B \rangle \middle| \rho_A \in \left( c_{\mathbf{A}} \cup c_{\mathbf{A'}} \right)^{**}, \rho_B \in c_{\mathbf{B}} \right\} \\
=& \left\{ \langle \rho_A, \rho_B \rangle \middle| \rho_A \in c_{\mathbf{A}} \cup c_{\mathbf{A'}}, \rho_B \in c_{\mathbf{B}} \right\}^{**} \\
=& \left( \left\{ \langle \rho_A, \rho_B \rangle \middle| \rho_A \in c_{\mathbf{A}}, \rho_B \in c_{\mathbf{B}} \right\} \right. \\ & \left. \cup \left\{ \langle \rho_A, \rho_B \rangle \middle| \rho_A \in c_{\mathbf{A'}}, \rho_B \in c_{\mathbf{B}} \right\} \right)^{**} \\
=& c_{\mathbf{\left( A \times B \right) \cup \left( A' \times B \right)}}
\end{split}
\end{equation}
\item $\mathbf{\left( A \cap A' \right) \oplus B} = \mathbf{\left( A \oplus B \right) \cap \left( A' \oplus B \right)}$: Dual to $\times, \cup$ case.
\item $\mathbf{\left( A \cup A' \right) \oplus B} = \mathbf{\left( A \oplus B \right) \cup \left( A' \oplus B \right)}$: Dual to $\times, \cap$ case.
\end{itemize}
\end{proof}

\section{Proofs from Section \ref{sec:graph_types}}

We will start these proofs by presenting a standard way to reliably construct morphisms that can signal information via a noisy encoding of a single bit channel that we can then use throughout for building counterexamples.

\begin{lemma}\label{lemma:binary_channel_encoding_morphism}
Given objects $\mathbf{A}, \mathbf{B} \in \ob{\caus{\catc}}$ where neither are first-order dual ($\left| c_{\mathbf{A}} \right|, \left| c_{\mathbf{B}} \right| > 1$), and some chosen distinct states $\rho_A \neq \nmaxmix_{\mathbf{A}} \in c_{\mathbf{A}}$, $\rho_B \neq \nmaxmix_{\mathbf{B}} \in c_{\mathbf{B}}$, there exists a morphism in $c_{\mathbf{A^* < \mathbf{B}}} \setminus c_{\mathbf{A^* > B}}$, or equivalently some $f \in \caus{\catc}\left(\mathbf{A}, \mathbf{B}\right)$ which actively signals:
\begin{equation}
\begin{split}
\nmaxmix_{\mathbf{A}} \fatsemi f &= \nmaxmix_{\mathbf{B}} \\
\fsub{\rho_A \fatsemi f} &\sim \alpha \cdot \fsub{\rho_B} + \left( \id_I - \alpha \right) \cdot \fsub{\nmaxmix_{\mathbf{B}}} \not\sim \fsub{\nmaxmix_{\mathbf{B}}}
\end{split}
\end{equation}
\end{lemma}

\begin{proof}
Extend $\left\{ \rho_{A,1} = \rho_A, \rho_{A,2} = \nmaxmix_{\mathbf{A}} \right\}$ to a basis (via Proposition \ref{prop:dual_basis}) $\left\{ \rho_{A,i} \right\}_i \subseteq \catc \left( I, A \right)$, with dual basis $\left\{ e_{A,i} \right\}_i \subseteq \sub{\catc} \left( A, I \right)$. Picking the additional states from the default causal basis (\ref{apc:basis}) and renormalising, we can have $\forall i . \rho_{A,i} \fatsemi \ndiscard_{\mathbf{A}} = \id_I$ and hence $\sum_i e_{A,i} \sim \fsub{\ndiscard_{\mathbf{A}}}$. Similarly define $\left\{ \rho_{B,j} \right\}_j$ and $\left\{ e_{B,j} \right\}_j$ for $\mathbf{B}$.

We start by defining the following morphism:
\begin{equation}
\begin{split}
t :=& e_{A,1} \fatsemi \fsub{\rho_{A,1}} + \sum_{i \neq 1} e_{A,i} \fatsemi \fsub{\rho_{A,2}} \\
\sim& e_{A,1} \fatsemi \fsub{\rho_{A,1}} + \left( \fsub{\ndiscard_{\mathbf{A}}} - e_{A,1} \right) \fatsemi \fsub{\rho_{A,2}} \\
& \in \sub{\catc} \left( A, B \right)
\end{split}
\end{equation}
We can view this in the style of Proposition \ref{prop:complement_binary_test} as a binary test for $e_{A,1}$, with outcomes encoded into $\mathbf{B}$ instead of $\mathbf{2}$. $t$ may not exist in $\catc$ (i.e. if $\catc$ is $\MatRp$ or $\CPs$, $t$ may not be positive), but we can represent it as an affine combination of morphisms in $\catc$: pick some representatives $t^+, t^- \in \catc \left( A, B \right)$ s.t. $t \sim \fsub{t^+} - \fsub{t^-}$. Applying Proposition \ref{prop:complement} to $t^-$ gives some invertible $\lambda$ and $t'$ s.t. $\fsub{t'} \sim \fsub{\lambda \cdot \discard_A \fatsemi \maxmix_B} - \fsub{t^-}$ (wlog. suppose $\lambda \theta_{\mathbf{A}}^{-1} \mu_{\mathbf{B}}^{-1} + \id_I$ is invertible, otherwise we could freely pick a larger $\lambda$ and $t'$). We then define
\begin{equation}
f := \left( \lambda \theta_{\mathbf{A}}^{-1} \mu_{\mathbf{B}}^{-1} + \id_I \right)^{-1} \cdot \left( t^+ + t' \right) \in \catc \left( A, B \right)
\end{equation}
so that $t$ is an affine combination of $f$ and $\ndiscard_{\mathbf{A}} \fatsemi \rho_{B,2}$, and we claim that it is causal $f \in \caus{\catc}\left(\mathbf{A}, \mathbf{B}\right)$. Consider an arbitrary effect $\pi \in c_{\mathbf{B}}^*$.
\begin{equation}
\begin{split}
&\fsub{f \fatsemi \pi} \\
\sim& \fsub{\left( \lambda \theta_{\mathbf{A}}^{-1} \mu_{\mathbf{B}}^{-1} + \id_I \right)^{-1}} \cdot \left(\fsub{\lambda \cdot \discard_A \fatsemi \maxmix_B} + t \right) \fatsemi \fsub{\pi} \\
\sim& \fsub{\left( \lambda \theta_{\mathbf{A}}^{-1} \mu_{\mathbf{B}}^{-1} + \id_I \right)^{-1}} \cdot \left( \begin{array}{r@{}l}
& \fsub{\lambda \theta_{\mathbf{A}}^{-1} \mu_{\mathbf{B}}^{-1} \cdot \ndiscard_{\mathbf{A}}} \fatsemi \fsub{\rho_{B,2}} \\
+ & e_{A,1} \fatsemi \fsub{\rho_{B,1}} \\
+ & \fsub{\ndiscard_{\mathbf{A}}} \fatsemi \fsub{\rho_{B,2}} \\
- & e_{A,1} \fatsemi \fsub{\rho_{B,2}}
\end{array} \right) \fatsemi \fsub{\pi} \\
\sim& \fsub{\left( \lambda \theta_{\mathbf{A}}^{-1} \mu_{\mathbf{B}}^{-1} + \id_I \right)^{-1}} \cdot \left( \begin{array}{r@{}l}
& \fsub{\lambda \theta_{\mathbf{A}}^{-1} \mu_{\mathbf{B}}^{-1} \cdot \ndiscard_{A}} \\
+ & e_{A,1} \\
+ & \fsub{\ndiscard_{\mathbf{A}}} \\
- & e_{A,1}
\end{array} \right) \\
\sim& \fsub{\left( \lambda \theta_{\mathbf{A}}^{-1} \mu_{\mathbf{B}}^{-1} + \id_I \right)^{-1}} \cdot \fsub{\lambda \theta_{\mathbf{A}}^{-1} \mu_{\mathbf{B}}^{-1} + \id_I} \cdot \fsub{\ndiscard_{\mathbf{A}}} \\
\sim& \fsub{\ndiscard_{\mathbf{A}}}
\end{split}
\end{equation}
So $f \fatsemi \pi \in c_{\mathbf{A}}^*$, i.e. $f$ is causal, and there is no signalling from the output to the input (its partial transpose is in $c_{\mathbf{A^* < B}}$).

We can verify that it can signal from the input to output.
\begin{equation}
\begin{split}
& \fsub{\nmaxmix_{\mathbf{A}} \fatsemi f} \\
\sim& \fsub{\left( \lambda \theta_{\mathbf{A}}^{-1} \mu_{\mathbf{B}}^{-1} + \id_I \right)^{-1}} \cdot \fsub{\nmaxmix_{\mathbf{A}}} \fatsemi \left( \fsub{\lambda \cdot \discard_A \fatsemi \maxmix_B} + t \right) \\
\sim& \fsub{\left( \lambda \theta_{\mathbf{A}}^{-1} \mu_{\mathbf{B}}^{-1} + \id_I \right)^{-1}} \cdot \fsub{\rho_{A,2}} \fatsemi \left( \begin{array}{r@{}l}
& \fsub{\lambda \theta_{\mathbf{A}} \mu_{\mathbf{B}}^{-1} \cdot \ndiscard_{\mathbf{A}}} \fatsemi \fsub{\rho_{B,2}} \\
+ & e_{A,1} \fatsemi \fsub{\rho_{B,1}} \\
+ & \fsub{\ndiscard_{\mathbf{A}}} \fatsemi \fsub{\rho_{B,2}} \\
- & e_{A,1} \fatsemi \fsub{\rho_{B,2}}
\end{array} \right) \\
\sim& \fsub{\left( \lambda \theta_{\mathbf{A}}^{-1} \mu_{\mathbf{B}}^{-1} + \id_I \right)^{-1}} \cdot \fsub{\lambda \theta_{\mathbf{A}}^{-1} \mu_{\mathbf{B}}^{-1} + \id_I} \cdot \fsub{\rho_{B,2}} \\
\sim& \fsub{\nmaxmix_{\mathbf{B}}}
\end{split}
\end{equation}
\begin{equation}
\begin{split}
& \fsub{\rho_A \fatsemi f} \\
\sim& \fsub{\left( \lambda \theta_{\mathbf{A}}^{-1} \mu_{\mathbf{B}}^{-1} + \id_I \right)^{-1}} \cdot \fsub{\rho_A} \fatsemi \left( \fsub{\lambda \cdot \discard_A \fatsemi \maxmix_B} + t \right) \\
\sim& \fsub{\left( \lambda \theta_{\mathbf{A}}^{-1} \mu_{\mathbf{B}}^{-1} + \id_I \right)^{-1}} \cdot \fsub{\rho_{A,1}} \fatsemi \left( \begin{array}{r@{}l}
& \fsub{\lambda \theta_{\mathbf{A}}^{-1} \mu_{\mathbf{B}}^{-1} \cdot \ndiscard_{\mathbf{A}}} \fatsemi \fsub{\rho_{B,2}} \\
+ & e_{A,1} \fatsemi \fsub{\rho_{B,1}} \\
+ & \fsub{\ndiscard_{\mathbf{A}}} \fatsemi \fsub{\rho_{B,2}} \\
- & e_{A,1} \fatsemi \fsub{\rho_{B,2}}
\end{array} \right) \\
\sim& \fsub{\left( \lambda \theta_{\mathbf{A}}^{-1} \mu_{\mathbf{B}}^{-1} + \id_I \right)^{-1} \cdot \left( \rho_B + \lambda \theta_{\mathbf{A}}^{-1} \mu_{\mathbf{B}}^{-1} \cdot \nmaxmix_{\mathbf{B}} \right)}
\end{split}
\end{equation}

We can instead see this as $f$ mapping $e_{B,1}$ to $e_{A,1}$ (up to a non-zero scalar).
\begin{equation}
\begin{split}
& \fsub{f} \fatsemi e_{B,1} \\
\sim& \fsub{\left( \lambda \theta_{\mathbf{A}}^{-1} \mu_{\mathbf{B}}^{-1} + \id_I \right)^{-1}} \cdot \left( \begin{array}{r@{}l}
& \fsub{\lambda \theta_{\mathbf{A}}^{-1} \mu_{\mathbf{B}}^{-1} \cdot \ndiscard_{\mathbf{A}}} \fatsemi \fsub{\rho_{B,2}} \\
+ & e_{A,1} \fatsemi \fsub{\rho_{B,1}} \\
+ & \fsub{\ndiscard_{\mathbf{A}}} \fatsemi \fsub{\rho_{B,2}} \\
- & e_{A,1} \fatsemi \fsub{\rho_{B,2}}
\end{array} \right) \fatsemi e_{B,1} \\
\sim& \fsub{\left( \lambda \theta_{\mathbf{A}}^{-1} \mu_{\mathbf{B}}^{-1} + \id_I \right)^{-1}} \cdot e_{A,1}
\end{split}
\end{equation}
\end{proof}

\begin{lemma}[Restatement of Lemma \ref{lemma:acyclic_graph_morphisms_flat}]
The set of graph states over $G$ with a fixed edge interpretation $\Delta$ is flat iff in every cycle of $G$ there is some edge $(u \to v)$ for which $\Delta(u \to v) \cong \mathbf{I}$.
\end{lemma}

\begin{proof}
$\Longleftarrow$: Suppose that every cycle of $G$ contains an $\mathbf{I}$ edge. Let $G'$ be the graph obtained by removing each such edge from $G$, leaving a DAG. Since each removed edge was interpreted as $\mathbf{I}$, graph states over $G$ are always graph states over $G'$. We can then show that graph states over DAGs are flat by induction on the number of vertices of the graph.

If $G$ is acyclic, there must be some maximal vertex $v$ with no outgoing edges. Pick the effect $\ndiscard_{\Gamma(v)} : \Gamma(v) \to \mathbf{I}$ and apply it to a graph state $g$ at $g_v$.
\begin{equation}
g_v \fatsemi \left( \ndiscard_{\Gamma(v)} \otimes \id \right) : \mathbf{I} \to \bigparr_{(u \to v) \in E} \Delta(u \to v)^*
\end{equation}
Since $\Delta$ can only pick first-order objects, the only causal morphism of this type is $\bigotimes_{(u \to v) \in E} \nmaxmix_{\Delta(u \to v)^*}$. Each $\nmaxmix_{\Delta(u \to v)^*}$ can be absorbed into $g_u$ to show that applying $\ndiscard_{\Gamma(v)}$ at $v$ to our graph state over $G$ gives a graph state over $G \setminus \left\{ v \right\}$. By induction, the set of graph states over $G \setminus \left\{ v \right\}$ is flat, giving the normalised effect $\bigotimes_{u \in V \setminus \left\{ v \right\}} \ndiscard_{\Gamma(u)}$, and so $g \fatsemi \left( \bigotimes_{u \in V} \ndiscard_{\Gamma(u)} \right) = \id_I$. This holds for every graph state $g$, so $\bigotimes_{u \in V} \ndiscard_{\Gamma(u)} : \mathbf{Gr}_{G}^{\Gamma} \to \mathbf{I}$.

$\bigotimes_{v \in V} \nmaxmix_{\Gamma(v)}$ is trivially a graph state by supplementing it with a pair of local causal state and effect for each edge of the graph.

$\Longrightarrow$: Suppose there exists a (wlog. minimal) cycle $\sigma$ of $G$ for which $\Delta$ assigns a non-trivial object to each edge (i.e. there are at least two distinct causal states).

For each edge $(u \to v) \in \sigma$, pick two distinct states $\rho_{uv} \neq \nmaxmix_{uv} \in c_{\Delta(u \to v)}$. For each segment $(u \to v \to w) \in \sigma$, let $f_v : \caus{\catc}\left( \Delta(u \to v), \Delta(v \to w)\right)$ be the morphism generated by Lemma \ref{lemma:binary_channel_encoding_morphism} which implements a noisy bit channel encoded into the chosen states for $(u \to v)$ and $(v \to w)$.

We will use these $f_v$ morphisms to construct two graph states that are equal up to a non-unit scalar, and so by linearity there is no effect that can send them both to $\id_I$.

Firstly, we pick a graph state $g$ where every component $g_v : \mathbf{I} \to \mathrm{Comp}_G^{\Gamma;\Delta}(v)$ is separable:
\begin{equation}
\begin{split}
g_v =& \rho_v \otimes \left( \bigotimes_{(u \to v) \in E} \pi_{uv} \right) \otimes \left( \bigotimes_{(v \to w) \in E} \rho_{vw} \right) \\
& \rho_v \in c_{\Gamma(v)} \\
& \forall (u \to v) \in E . \pi_{uv} \in c_{\Delta(u \to v)}^* \\
& \forall (v \to w) \in E . \rho_{vw} \in c_{\Delta(v \to w)}
\end{split}
\end{equation}
so $g = \bigotimes_{v \in V} \rho_v$ (the state-effect pairs for each edge merge to unit scalars).

We build another graph state $g'$ by taking $g$ and for every vertex $v$ in the cycle $\sigma$ ($(u \to v \to w) \in \sigma$), we obtain $g'_v$ from $g_v$ by replacing $\pi_{uv} \otimes \rho_{vw}$ with $\eta \fatsemi \left( \id \otimes f_v \right)$ (where $\eta$ is the appropriate cup state). This still gives $g'_v : \mathbf{I} \to \mathrm{Comp}_G^{\Gamma;\Delta}(v)$ so we have a valid graph state. Removing the state-effect pairs for edges not in $\sigma$, we find that $g'$ is $g$ with an additional circular morphism following $\sigma$.
\begin{equation}
g' = \left( \eta \fatsemi \left( \id \otimes \prod_{v \in \sigma} f_v \right) \fatsemi \epsilon \right) \cdot g
\end{equation}
We expand one $f_{\hat{v}}$ in this scalar (with $\left( \hat{u} \to \hat{v} \to \hat{w} \right) \in \sigma$), and use the fact that each $f_v$ is causal and noisily preserves the dual effect to simplify. Here, we use similar notations as in Lemma \ref{lemma:binary_channel_encoding_morphism} for the chosen states, dual effects, and relevant scalars.

\begin{equation}
\begin{split}
& \fsub{\eta \fatsemi \left( \id \otimes \prod_{v \in \sigma} f_v \right) \fatsemi \epsilon} \\
\sim& \fsub{\left( \lambda_{\hat{v}} \theta_{uv}^{-1} \mu_{vw}^{-1} + \id_I \right)^{-1}} \cdot \fsub{\eta} \fatsemi \\ & \left( \id \otimes \left( \left( \begin{array}{r@{}l} & \fsub{\lambda_{\hat{v}} \theta_{\hat{u}\hat{v}}^{-1} \mu_{\hat{v}\hat{w}}^{-1} \cdot \ndiscard_{\hat{u}\hat{v}}} \fatsemi \fsub{\rho_2^{\hat{v}\hat{w}}} \\ +& e_1^{\hat{u}\hat{v}} \fatsemi \fsub{\rho_1^{\hat{v}\hat{w}}} \\ +& \fsub{\ndiscard_{\hat{u}\hat{v}}} \fatsemi \fsub{\rho_2^{\hat{v}\hat{w}}} \\ -& e_1^{\hat{u}\hat{v}} \fatsemi \fsub{\rho_2^{\hat{v}\hat{w}}} \end{array} \right) \fatsemi \fsub{\prod_{\hat{v} \neq v \in \sigma} f_v} \right) \right) \fatsemi \fsub{\epsilon} \\
\sim& \fsub{\rho_2^{\hat{v}\hat{w}}} \fatsemi \fsub{\prod_{\hat{v} \neq v \in \sigma} f_v} \fatsemi \fsub{\ndiscard_{\hat{u}\hat{v}}} \\ &+ \fsub{\left( \lambda_{\hat{v}} \theta_{uv}^{-1} \mu_{vw}^{-1} + \id_I \right)^{-1}} \cdot \left( \fsub{\rho_1^{\hat{v}\hat{w}}} - \fsub{\rho_2^{\hat{v}{\hat{w}}}} \right) \fatsemi \fsub{\prod_{\hat{v} \neq v \in \sigma} f_v} \fatsemi e_1^{\hat{u}\hat{v}} \\
\sim& \fsub{\rho_2^{\hat{v}\hat{w}}} \fsub{\ndiscard_{\hat{v}\hat{w}}} \\ & + \fsub{\prod_{v \in \sigma} \left( \lambda_v \theta_{uv}^{-1} \mu_{vw}^{-1} + \id_I \right)^{-1}} \cdot \left( \fsub{\rho_1^{\hat{v}\hat{w}}} - \fsub{\rho_2^{\hat{v}{\hat{w}}}} \right) \fatsemi e_1^{\hat{v}\hat{w}} \\
\sim& \fsub{\id_I + \prod_{v \in \sigma} \left( \lambda_v \theta_{uv}^{-1} \mu_{vw}^{-1} + \id_I \right)^{-1}}
\end{split}
\end{equation}
The big product must be non-zero since it is invertible, so we can conclude that $g$ and $g'$ are equal up to a non-unit scalar.
\end{proof}

\begin{theorem}[Restatement of Theorem \ref{thrm:all_graph_types_equivalent}]
For any DAG $G$ with ordered vertices and local interpretation $\Gamma$, $\mathbf{LoGr}_G^\Gamma = \mathbf{LoGr2}_G^\Gamma = \mathbf{RSiGr}_G^\Gamma = \mathbf{SiGr}_G^\Gamma = \mathbf{OrGr}_G^\Gamma$.
\end{theorem}

\begin{proof}
We will show a cycle of inclusions between the state sets of these types.

\textbf{Inclusion 1} $c_{\mathbf{LoGr2}_G^\Gamma} \subseteq c_{\mathbf{LoGr}_G^\Gamma}$: graph states using the constant edge interpretation are included in the set of all graph states.

\textbf{Inclusion 2} $c_{\mathbf{LoGr}_G^\Gamma} \subseteq c_{\mathbf{OrGr}_G^\Gamma}$: for any topological order $O$, we proceed in the manner of the $\Longleftarrow$ direction of Lemma \ref{lemma:acyclic_graph_morphisms_flat}: applying any causal effect at the maximal vertex $v$ gives a constant marginal over $V \setminus \left\{ v \right\}$ as the edges into it are first-order. We continue down the order of $O$ inductively, showing one-way signalling at each step.

\textbf{Inclusion 3} $c_{\mathbf{OrGr}_G^\Gamma} \subseteq c_{\mathbf{SiGr}_G^\Gamma}$: pick some state $h \in c_{\mathbf{OrGr}_G^\Gamma}$ and some down-closed set $U \subset V$. Since $U$ is down-closed, there must exist some topological sort $O$ which lists all vertices in $U$ and then all vertices in $V \setminus U$. Picking any choice of separable causal effects over $V \setminus U$, $h : \mathrm{perm}_V \left( \left( \bigseq_{v \in O \cap U} \Gamma(v) \right) < \left( \bigseq_{v \in O \setminus U} \Gamma(v) \right) \right)$ which is a subtype of $\mathrm{perm}_V \left( \left( \bigparr_{v \in U} \Gamma(v) \right) < \left( \bigparr_{v \in V \setminus U} \Gamma(v) \right) \right)$.

\textbf{Inclusion 4} $c_{\mathbf{SiGr}_G^\Gamma} \subseteq c_{\mathbf{RSiGr}_G^\Gamma}$: by induction on the number of vertices in $G$. The cases for zero and one vertex are trivial, so consider some $h \in c_{\mathbf{SiGr}_G^\Gamma}$ ($|V| > 1$) and an arbitrary down-closed $U \subset V$. It is sufficient to show that the constant marginal $h_U : \bigparr_{v \in U} \Gamma(v)$ is actually in $c_{\mathbf{SiGr}_{G[U]}^{\Gamma[U]}}$ (since $|U| < |V|$, this is a subset of $c_{\mathbf{RSiGr}_{G[U]}^{\Gamma[U]}}$ by the induction hypothesis).

Pick an arbitrary $U' \subset U$ which is down-closed wrt $G[U]$ (and hence down-closed wrt $G$). Using the signalling graph type property of $h$, it has a constant marginal $h_{U'} : \bigparr_{v \in U'} \Gamma(v)$ for any choice of local effects over $V \setminus U'$. Applying any choice of effects $\left\{ \pi_v : \Gamma(v) \to \mathbf{I} \right\}_{v \in U \setminus U'}$ on $h_U$ has a constant marginal, showing that $h_U : \left( \bigparr_{v \in U'} \Gamma(v) \right) <  \left( \bigparr_{v \in U \setminus U'} \Gamma(v) \right)$:
\begin{equation}
\begin{split}
& h_U \fatsemi \bigotimes_{v \in U} \left\{ \begin{array}{ll} \id_{\Gamma(v)} & v \in U' \\ \pi_v & v \in U \setminus U' \end{array} \right. \\
=& h \fatsemi \bigotimes_{v \in V} \left\{ \begin{array}{ll} \id_{\Gamma(v)} & v \in U' \\ \pi_v & v \in U \setminus U' \\ \ndiscard_{\Gamma(v)} & v \in V \setminus U \end{array} \right. \\
=& h_{U'}
\end{split}
\end{equation}
Taking the intersections of these types for each down-closed $U' \subset U$ gives $h_U : \mathbf{SiGr}_{G[U]}^{\Gamma[U]}$.

\textbf{Inclusion 5} $c_{\mathbf{RSiGr}_G^\Gamma} \subseteq c_{\mathbf{LoGr2}_G^\Gamma}$: This is by far the most involved direction. It is, again, by induction on the number of vertices where the cases for zero and one are trivial so we focus on the inductive case. Firstly, we will fix a bunch of notation. For each $v \in V$:
\begin{itemize}
\item Let $\left\{ \rho_{i_v}^v \right\}_{i_v} \subseteq \catc \left( I, \mathcal{U}\left( \Gamma(v)\right) \right)$ be the extension of some maximal linearly independent subset of $c_{\Gamma(v)}$ to a basis for $\mathcal{U}\left( \Gamma(v) \right)$. We normalise the added states s.t. $\forall i_v . \rho_{i_v}^v \fatsemi \ndiscard_{\Gamma(v)} = \id_I$.
\item Let $\left\{ e_{i_v}^v \right\}_{i_v} \subseteq \sub{\catc} \left( \mathcal{U}\left(\Gamma(v)\right), I \right)$ be the dual basis to $\left\{ \rho_{i_v}^v \right\}_{i_v}$. The normalisation of states guarantees that $\sum_i e_{i_v}^v \sim \fsub{\ndiscard_{\Gamma(v)}}$.
\item Let $\lambda_v$ be an invertible scalar and $\left\{ \sigma_{i_v}^v \right\}_{i_v} \subseteq \catc\left( I, \mathcal{U}\left(\Gamma(v)\right) \right)$ s.t. $\forall i_v . \rho_{i_v}^v + \sigma_{i_v}^v = \lambda_v \cdot \maxmix$ (e.g. apply Proposition \ref{prop:complement} for each $\rho_{i_v}^v$ to get $\lambda_{v,i_v}$ and take the maximum with respect to the partial order of \ref{apc:scalars}).
\item Let $\left\{t_i^v : \mathbf{2}^{\otimes |\mathrm{in}(v)|} \to \Gamma(v) \parr \mathbf{2}^{\parr |\mathrm{out}(v)|}\right\}_{i}$ be causal tests for $\left\{\rho_i^v\right\}_i$ (in the vain of Proposition \ref{prop:complement_binary_test}), conditioned on all predecessors in the graph and broadcast to all successors.
\begin{equation}
\begin{split}
& \left( \bigotimes_{u \in \mathrm{in}(v)} \iota_{x_u} \right) \fatsemi t_i^v := \\
& \begin{cases}
\mu_v \lambda_v^{-1} \cdot \left( \begin{array}{r@{}l} & \rho_i^v \otimes \bigotimes_{w \in \mathrm{out}(v)} \tbool \\ + & \sigma_i^v \otimes \bigotimes_{w \in \mathrm{out}(v)} \fbool \end{array} \right) & \forall u \in \mathrm{in}(v) . \iota_{x_u} = \tbool \\
\nmaxmix_{\Gamma(v)} \otimes \bigotimes_{w \in \mathrm{out}(v)} \fbool & \text{otherwise}
\end{cases}
\end{split}
\end{equation}
\item Let $\left\{ T_i^v\right\}_i$ be conditional preparations.
\begin{equation}
\left( \bigotimes_{u \in \mathrm{in}(v)} \iota_{x_u} \right) \fatsemi T_i^v := \begin{cases}
\rho_i^v & \forall u \in \mathrm{in}(v) . \iota_{x_u} = \tbool \\
\nmaxmix_{\Gamma(v)} & \text{otherwise}
\end{cases}
\end{equation}
For those $i$ where $\rho_i^v \in c_{\Gamma(v)}$, $T_i^v : \mathbf{2}^{\otimes |\mathrm{in}(v)|} \to \Gamma(v)$ is causal.
\end{itemize}

Let $\max G = \left\{ v \in V | \not\exists (v \to w) \in E \right\}$ be the set of maximal vertices of the graph (since $G$ is acyclic, this must be non-empty). Let $g_{\overrightarrow{i}} \in c_{\mathbf{LoGr2}_G^\Gamma}$ be the graph state over $G$ wrt the constant-$\mathbf{2}$ edge interpretation $\Delta_{\mathbf{2}}$ whose components are $\left\{ t_{i_v}^v | v \notin \max G\right\} \cup \left\{ T_{i_v}^v | v \in \max G \right\}$; note that this is only a valid graph state when $\rho_{i_v}^v \in c_{\Gamma(v)}$, otherwise $T_{i_v}^v$ is not causal for the appropriate component type.

Now, we take an arbitrary state $h : \mathbf{RSiGr}_G^\Gamma$.

For any $v \in \max V$, we claim the following:
\begin{equation}
\forall i_v . \rho_{i_v}^v \notin c_{\Gamma(v)} \Rightarrow \fsub{h} \fatsemi \left( \bigotimes_{u \in V} \begin{cases}
e_{i_v}^v & u = v \\
\id & u \neq v
\end{cases} \right) \sim 0
\end{equation}
$V \setminus \left\{v\right\}$ is down-closed, so the recursive signalling graph type tells us that $h : \mathrm{perm}_V \left( \mathbf{RSiGr}_{G[V \setminus \{v\}]}^{\Gamma[V \setminus \{v\}]} < \Gamma(v) \right)$. The definition of the sequence operator from Equation \ref{eq:seq_object_local_sum_def} breaks $\fsub{h}$ (up to permutation) into a sum $\sum_j f_j \otimes \fsub{f'_j}$ where $\forall j . f'_j \in c_{\Gamma(v)}$. Since the basis $\left\{\rho_{i_v}^v\right\}_{i_v}$ contains a maximal linearly-independent subset of $c_{\Gamma(v)}$, expanding each $f'_j$ into this basis only uses the causal states, i.e. if $\rho_{i_v}^v \notin c_{\Gamma(v)}$ then $\fsub{f'_j} \fatsemi e_{i_v}^v \sim 0$. This extends to $h$ by linearity.

Now, consider the following expression:
\begin{equation}\label{eq:canonical_affine_graph_sum}
\sum_{\overrightarrow{i}} \left( \fsub{h} \fatsemi \bigotimes_{v \in V} e_{i_v}^v \right) \cdot \fsub{g_{\overrightarrow{i}}}
\end{equation}
This is a linear combination of graph states wrt $\Delta_{\mathbf{2}}$, since any choice of $\overrightarrow{i}$ for which $g_{\overrightarrow{i}}$ is not valid ($\exists v \in \max V . \rho_{i_v}^v \notin c_{\Gamma(v)}$) would give a zero coefficient. The following shows it is affine (it is straightforward to show that $c_{\mathbf{RSiGr}_G^\Gamma} \subseteq c_{\bigparr_{v \in V} \Gamma(v)}$ so $h$ is normalised under local causal effects):
\begin{equation}
\begin{split}
\sum_{\overrightarrow{i}} \fsub{h} \fatsemi \bigotimes_{v \in V} e_{i_v}^v &\sim \fsub{h} \fatsemi \bigotimes_{v \in V} \sum_{i_v} e_{i_v}^v \\
& \sim \fsub{h} \fatsemi \bigotimes_{v \in V} \fsub{\ndiscard_{\Gamma(v)}} \\
& \sim \id_I
\end{split}
\end{equation}

The rest of the proof will show that expanding the definitions of each $t_{i_v}^v$ and $T_{i_v}^v$ in Equation \ref{eq:canonical_affine_graph_sum} equates $g_{\overrightarrow{i}}$ to another affine combination, where precisely one term in the sum is $h$ itself (with a non-zero coefficient) and the rest are other graph states wrt $\Delta_{\mathbf{2}}$. From this, we could solve the equation for $h$ to represent it as an affine combination of graph states wrt $\Delta_{\mathbf{2}}$, showing that it lives in $c_{\mathbf{LoGr2}_G^\Gamma}$.

For each index $i_v$ in the sum of \ref{eq:canonical_affine_graph_sum}, only $e_{i_v}^v$ and either $t_{i_v}^v$ or $T_{i_v}^v$ depend on it. Summing over that index locally gives the following decompositions:

\begin{equation}\label{eq:non_maximal_vertex_channel}
\begin{split}
& \sum_{i_v} e_{i_v}^v \fatsemi \fsub{\left( \bigotimes_{u \in \mathrm{in}(v)} \iota_{x_u} \right) \fatsemi t_{i_v}^v} \\
\sim & \left\{ \begin{array}{l@{\hskip-30pt}l}
\fsub{\mu_v \lambda_v^{-1}} \cdot \left( \begin{array}{rl} & \sum_{i_v} e_{i_v}^v \fatsemi \fsub{\rho_{i_v}^v \otimes \bigotimes_{w \in \mathrm{out}(v)} \tbool} \\ + & \sum_{i_v} e_{i_v}^v \fatsemi \fsub{\sigma_{i_v}^v \otimes \bigotimes_{w \in \mathrm{out}(v)} \fbool} \end{array} \right) & \\ & \forall u \in \mathrm{in}(v) . \iota_{x_u} = \tbool \\
\sum_{i_v} e_{i_v}^v \fatsemi \fsub{\nmaxmix_{\Gamma(v)} \otimes \bigotimes_{w \in \mathrm{out}(v)} \fbool} & \text{otherwise}
\end{array} \right. \\
\sim & \left\{ \begin{array}{l@{\hskip-35pt}l}
\left( \begin{array}{r@{}l} & \fsub{\mu_v \lambda_v^{-1}} \cdot \left( \sum_{i_v} e_{i_v}^v \fatsemi \fsub{\rho_{i_v}^v} \right) \otimes \fsub{\bigotimes_{w \in \mathrm{out}(v)} \tbool} \\ - & \fsub{\mu_v \lambda_v^{-1}} \cdot \left( \sum_{i_v} e_{i_v}^v \fatsemi \fsub{\rho_{i_v}^v} \right) \otimes \fsub{\bigotimes_{w \in \mathrm{out}(v)} \fbool} \\ + & \left( \sum_{i_v} e_{i_v}^v \fatsemi \fsub{\nmaxmix_{\Gamma(v)}} \right) \otimes \fsub{\bigotimes_{w \in \mathrm{out}(v)} \fbool} \end{array} \right) & \\ & \forall u \in \mathrm{in}(v) . \iota_{x_u} = \tbool \\
\left( \sum_{i_v} e_{i_v}^v \fatsemi \fsub{\nmaxmix_{\Gamma(v)}} \right) \otimes \fsub{\bigotimes_{w \in \mathrm{out}(v)} \fbool} & \text{otherwise}
\end{array} \right. \\
\sim & \begin{cases}
\left( \begin{array}{r@{}l} & \fsub{\mu_v \lambda_v^{-1}} \cdot \fsub{\id} \otimes \fsub{\bigotimes_{w \in \mathrm{out}(v)} \tbool} \\ - & \fsub{\mu_v \lambda_v^{-1}} \cdot \fsub{\id} \otimes \fsub{\bigotimes_{w \in \mathrm{out}(v)} \fbool} \\ + & \fsub{\ndiscard_{\Gamma(v)} \fatsemi \nmaxmix_{\Gamma(v)}} \otimes \fsub{\bigotimes_{w \in \mathrm{out}(v)} \fbool} \end{array} \right) & \forall u \in \mathrm{in}(v) . \iota_{x_u} = \tbool \\
\fsub{\ndiscard_{\Gamma(v)} \fatsemi \nmaxmix_{\Gamma(v)}} \otimes \fsub{\bigotimes_{w \in \mathrm{out}(v)} \fbool} & \text{otherwise}
\end{cases}
\end{split}
\end{equation}

\begin{equation}\label{eq:maximal_vertex_channel}
\begin{split}
& \sum_{i_v} e_{i_v}^v \fatsemi \fsub{\left( \bigotimes_{u \in \mathrm{in}(v)} \iota_{x_u} \right) \fatsemi T_{i_v}^i} \\
\sim & \begin{cases}
\sum_{i_v} e_{i_v}^v \fatsemi \fsub{\rho_{i_v}^v} & \forall u \in \mathrm{in}(v) . \iota_{x_u} = \tbool \\
\sum_{i_v} e_{i_v}^v \fatsemi \fsub{\nmaxmix_{\Gamma(v)}} & \text{otherwise}
\end{cases} \\
\sim & \begin{cases}
\id & \forall u \in \mathrm{in}(v) . \iota_{x_u} = \tbool \\
\fsub{\ndiscard_{\Gamma(v)} \fatsemi \nmaxmix_{\Gamma(v)}} & \text{otherwise}
\end{cases}
\end{split}
\end{equation}

In each case from these, we get an affine combination of a local channel applied at the vertex (either the identity or completely noisy channel) and some classical outcome passed to the next vertices. So in the complete expansion of \ref{eq:canonical_affine_graph_sum}, we get an affine combination of terms, each of which is $h$ with the completely noisy channel applied to some subset of vertices.
\begin{equation}
\sum_{\overrightarrow{i}} \left( \fsub{h} \fatsemi \bigotimes_{v \in V} e_{i_v}^v \right) \cdot \fsub{g_{\overrightarrow{i}}} \sim \sum_{U \subseteq V} \alpha_U \cdot \fsub{h \fatsemi \bigotimes_{v \in V} \left\{ \begin{array}{l@{\hskip-10pt}l}
\id & v \in U \\
\ndiscard_{\Gamma(v)} \fatsemi \nmaxmix_{\Gamma(v)} & \\ & v \notin U
\end{array} \right.}
\end{equation}
Not every such subset is generated, so some of these coefficients $\alpha_U$ may be zero.

The term with $U = V$ reduces completely to $h$ and has coefficient $\alpha_V = \prod_{v \notin \max V} \mu_v \lambda_v^{-1}$ which is invertible and non-zero, so we really can solve this equation for $h$.

All that remains is to show that the other terms in this expansion are in $c_{\mathbf{LoGr2}_G^\Gamma}$. Consider the term for some $U \subset V$ with a non-zero coefficient.

$U$ must be down-closed; for any $(u \to v) \in E$, if $v \in U$ then Equations \ref{eq:non_maximal_vertex_channel} and \ref{eq:maximal_vertex_channel} tell us $v$ must have received $\tbool$ from each of its predecessors, and the only term from the expansion at $u$ which yields $\tbool$ applies $\id$ at $u$, i.e. $u \in U$.

The $U$ case of the recursive signalling graph type of $h$ gives a sequence typing $h : \mathrm{perm}_v \left( \mathbf{RSiGr}_{G[U]}^{\Gamma[U]} < \bigparr_{v \in V \setminus U} \Gamma(v) \right)$, so after applying $\ndiscard_{\Gamma(v)}$ on each $v \in V \setminus U$, we are left with the constant marginal $h_U : \mathbf{RSiGr}_{G[U]}^{\Gamma[U]}$. Applying the induction hypothesis gives $h_U : \mathbf{LoGr2}_{G[U]}^{\Gamma[U]}$ and hence can be represented as an affine combination of graph states over $G[U]$ wrt $\Delta_{\mathbf{2}}$. Adding on the local states $\nmaxmix_{\Gamma(v)}$ at each $v \in V \setminus U$ gives an affine combination of graph states over $G$ with $\Delta_{\mathbf{2}}$.
\end{proof}

\begin{corollary}[Restatement of Corollary \ref{corollary:seq_as_2local}]
\begin{equation*}
\mathbf{A < B} = \left( A \otimes B, \left\{ \input{./semi_local_HO_named_2.tikz} \middle| \begin{array}{l} h_{A2} \in c_{\mathbf{A \parr 2}}, \\ h_{2B} \in c_{\mathbf{2^* \parr B}} \end{array} \right\}^{**} \right)
\end{equation*}
\end{corollary}

\begin{proof}
Immediate from Equation \ref{eq:seq_as_graph_type} and Theorem \ref{thrm:all_graph_types_equivalent}.
\end{proof}

\begin{corollary}[Restatement of Corollary \ref{corollary:series_parallel_graph_types}]
Given graph types over non-empty disjoint vertex sets $V$ and $V'$,
\begin{align*}
\mathbf{Gr}_{\left( V, E \right)}^{\Gamma} \otimes \mathbf{Gr}_{\left( V', E' \right)}^{\Gamma'} &= \mathbf{Gr}_{\left( V \cup V', E \cup E' \right)}^{\Gamma, \Gamma'} \\
\mathbf{Gr}_{\left( V, E \right)}^{\Gamma} < \mathbf{Gr}_{\left( V', E' \right)}^{\Gamma'} &= \mathbf{Gr}_{\left( V \cup V', E \cup E' \cup \left\{ v \to v' \middle| v \in V, v' \in V' \right\} \right)}
\end{align*}
\end{corollary}

\begin{proof}
For $\otimes$, consider the local graph type definition. Graph states over the combined graph $\left( V \cup V', E \cup E' \right)$ are precisely the parallel composition of a graph state over $\left( V, E \right)$ and one over $\left( V', E' \right)$. Each type just takes affine combinations of these.

For $<$, use the ordered signalling type definition. Each topological ordering of $\left( V \cup V', E \cup E' \cup \left\{ v \to v' \middle| v \in V, v' \in V' \right\} \right)$ is a topological ordering of $\left( V, E \right)$ followed by a topological ordering of $\left( V', E' \right)$. Distribution between $<$ and $\cap$ (Equation \ref{eq:iu_monoidal_distribution}) allows us to combine the quantification over the separate graphs into quantification over the combined graph.
\end{proof}

\begin{lemma}[Restatement of Lemma \ref{lemma:transitive_closure_graph_type_without_pruning}]
$\mathbf{Gr}_G^\Gamma = \mathbf{Gr}_{G^+}^\Gamma$
\end{lemma}

\begin{proof}
This is immediate from the signalling graph type definition, since $U \subseteq V$ is down-closed wrt $G$ iff it is down-closed wrt $G^+$.
\end{proof}

\begin{corollary}[Restatement of Corollary \ref{corollary:scalar_graph_types}]
$\mathbf{Gr}_{\left( V, E \right)}^{\left\{ v \mapsto \mathbf{I} \right\}_{v \in V}} \cong \mathbf{I}$
\end{corollary}

\begin{proof}
$\mathbf{Gr}_{\left( V, E \right)}^{\left\{ v \mapsto \mathbf{I} \right\}_{v \in V}} = \bigcap_{O \text{ topological sort of } \left( V, E \right)} \mathrm{perm}_V \left( \bigseq_{v \in O} \mathbf{I} \right) = \bigparr_{v \in O} \mathbf{I} \cong \mathbf{I}$
\end{proof}

\begin{lemma}[Restatement of Lemma \ref{lemma:unit_elimination_from_graph_types}]
$\mathbf{Gr}_{G}^{\Gamma, v \mapsto I} \cong \mathbf{Gr}_{G^+ \setminus \left\{ v \right\}}^{\Gamma}$
\end{lemma}

\begin{proof}
This is straightforward from the ordered graph type definition, since an ordering of $V \setminus \left\{ v \right\}$ is topological wrt $G$ iff it is topological wrt $G^+ \setminus \left\{ v \right\}$.
\end{proof}

\begin{lemma}\label{lemma:transitive_graph_edge_pruning}
Given two transitive DAGs $G = \left( V, E \right)$ and $G' = \left( V, E \cup (u \to v) \right)$ differing by a single edge, $\mathbf{Gr}_G^\Gamma = \mathbf{Gr}_{G'}^\Gamma$ iff $\Gamma(u)$ or $\Gamma(v)^*$ is first-order.
\end{lemma}

\begin{proof}
Again, we will focus on the signalling graph type definition. $c_{\mathbf{Gr}_G^\Gamma} \subseteq c_{\mathbf{Gr}_{G'}^\Gamma}$ is immediate since any $U \subset V$ that is down-closed wrt $G'$ is also down-closed wrt $G$. We will just show that the opposite direction of inclusion holds iff $\Gamma(u)$ or $\Gamma(v)^*$ is first-order.

$\Longrightarrow$: aiming for the contrapositive, suppose neither of $\Gamma(u)$ or $\Gamma(v)^*$ is first-order, i.e. $\left| c_{\Gamma(u)}^* \right|, \left| c_{\Gamma(v)} \right| > 1$. Applying Lemma \ref{lemma:binary_channel_encoding_morphism} gives us some $f \in c_{\Gamma(u) < \Gamma(v)} \setminus c_{\Gamma(u) > \Gamma(v)}$ which actively signals from $u$ to $v$. Let $h$ be a parallel composition of $f$ with some separable states $h_{v'} \in c_{\Gamma(v')}$ for each $v' \neq u, v$. $h \in c_{\mathbf{SiGr}_{G'}^\Gamma}$ because any down-closed set wrt $G'$ containing $v$ also contains $u$. However, $\mathrm{in}_{G}(v)$ is down-closed wrt $G$ and does not contain $u$, so the active signalling of $f$ means $h \notin c_{\mathbf{SiGr}_G^\Gamma}$.

$\Longleftarrow$: Consider an arbitrary $h \in c_{\mathbf{SiGr}_{G'}^\Gamma}$, so it is non-signalling into any subset down-closed wrt $G'$. Then consider some $U \subset V$ that is down-closed wrt $G$ but not wrt $G'$, i.e. $v \in U$ but $u \notin U$. However, transitivity of $G$ and $G'$ implies that both $U \cup \left\{ u \right\}$ and $U \setminus \left\{ v \right\}$ are down-closed wrt $G'$:

\begin{itemize}
\item $U \cup \left\{ u \right\}$: for any $(w \to u) \in G'$, $(w \to v) \in G'$ by transitivity and hence $w \in U$ by down-closure of $U$.
\item $U \setminus \left\{ v \right\}$: for any $(v \to w) \in G'$, $(u \to w) \in G'$ by transitivity and $(u \to w) \in G$ since the only edge that differs is $(u \to v)$. By down-closure of $U$ wrt $G$ and $u \notin U$, we have $w \notin U$.
\end{itemize}

If $\Gamma(u)$ is first-order, there is a unique effect $\ndiscard_{\Gamma(u)}$. We already have $h$ is non-signalling into $U \cup \left\{ u \right\}$, so applying $\ndiscard_{\Gamma(u)}$ must give a unique marginal over $U$.

Suppose instead $\Gamma(v)$ is first-order dual with unique state $\nmaxmix_{\Gamma(v)}$. Picking any choice of effects $\left\{ \pi_w \right\}_{w \in V \setminus U}$, we consider the marginal $h_U$ over $U$ when we apply these. Since $\Gamma(v)$ is first-order dual, its local causal effects (the states of $\Gamma(v)^*$) include a basis. Whichever basis element we apply to $h_U$, we must get the constant marginal $h_{U \setminus \left\{ v \right\}}$ since $h$ is non-signalling into $U \setminus \left\{ v \right\}$, so by equality on every element of a basis we have $h_U = h_{U \setminus \left\{ v \right\}} \otimes \nmaxmix_{\Gamma(v)}$. This is independent of our choices $\left\{ \pi_w \right\}_{w \in V \setminus U}$, so $h$ is non-signalling into $U$.
\end{proof}

\begin{lemma}[Restatement of Lemma \ref{lemma:transitive_closure_graph_type}]
Given a DAG $G = (V, E)$ and a local interpretation $\Gamma$, let $\overline{G} = (V, \overline{E})$ be the DAG with $(u \to v) \in \overline{E} \Leftrightarrow (u \to v) \in E^+ \wedge \left| c_{\Gamma(u)}^* \right| > 1 \wedge \left| c_{\Gamma(v)} \right| > 1$. Then $\mathbf{Gr}_G^\Gamma = \mathbf{Gr}_{\overline{G}}^\Gamma$. Furthermore, for any $G' = (V, E')$ over the same vertices, $c_{\mathbf{Gr}_G^\Gamma} \subseteq c_{\mathbf{Gr}_{G'}^\Gamma}$ iff $\overline{E} \subseteq \overline{E'}$.
\end{lemma}

\begin{proof}
$\mathbf{Gr}_{G}^\Gamma = \mathbf{Gr}_{\overline{G}}^\Gamma$ follows from Lemmas \ref{lemma:transitive_closure_graph_type_without_pruning} and \ref{lemma:transitive_graph_edge_pruning} by pruning edges one by one from the transitive closure. The remainder now requires us to show that $c_{\mathbf{Gr}_{\overline{G}}^\Gamma} \subseteq c_{\mathbf{Gr}_{\overline{G'}}^\Gamma}$ iff $\overline{E} \subseteq \overline{E'}$.

$\Longleftarrow$: If $\overline{E} \subseteq \overline{E'}$, then the inclusion of graph types follows straightforwardly from any of the graph type definitions.

$\Longrightarrow$: Suppose, aiming for the contrapositive, that there is some $(u \to v) \in \overline{E} \setminus \overline{E'}$. By definition of $\overline{E}$, $\left| c_{\Gamma(u)}^* \right| > 1$ and $\left| c_{\Gamma(v)} \right| > 1$, so we can use Lemma \ref{lemma:binary_channel_encoding_morphism} to generate a counterexample in $c_{\mathbf{SiGr}_{\overline{E}}^\Gamma} \setminus c_{\mathbf{SiGr}_{\overline{E'}}^\Gamma}$, following the same proof as the $\Longrightarrow$ direction of Lemma \ref{lemma:transitive_graph_edge_pruning}.
\end{proof}

\begin{lemma}[Restatement of Lemma \ref{lemma:transitive_closure_dual_graph_type}]
Let $G = (V, E)$ and $G' = (V, E')$ be two DAGs over the same set of vertices, and $\Gamma, \Gamma^* : V \to \ob{\caus{\catc}}$ be dual local interpretation functions, $\forall v \in V . \Gamma^*(v) = \Gamma(v)^*$. Then $c_{\mathbf{Gr}_G^\Gamma} \subseteq c_{\left( \mathbf{Gr}_{G'}^{\Gamma^*} \right)^*}$ iff $(V, \overline{E} \cup \overline{E'})$ is acyclic.
\end{lemma}

\begin{proof}
By Lemma \ref{lemma:transitive_closure_graph_type}, $c_{\mathbf{Gr}_G^\Gamma} \subseteq c_{\left( \mathbf{Gr}_{G'}^{\Gamma^*} \right)^*}$ iff $c_{\mathbf{Gr}_{\overline{G}}^\Gamma} \subseteq c_{\left( \mathbf{Gr}_{\overline{G'}}^{\Gamma^*} \right)^*}$. This inclusion holds exactly when, if we compose any pair of graph states over $\overline{G}$ with $\Gamma$ and over $\overline{G'}$ with $\Gamma^*$ by contracting the components at each vertex, then the resulting scalar is $\id_I$.

$\Longrightarrow$: Aiming for the contrapositive, assume $\left( V, \overline{E} \cup \overline{E'} \right)$ is not acyclic. Pick a minimal cycle $\sigma$ (so it visits each vertex at most once) and break it down into segments of edges in $\overline{E}$ and edges in $\overline{E'}$ (if the edge is in both, we may pick either arbitrarily). Let $G_\sigma$ be the graph over $V \uplus V$ consisting of $\overline{G}$ and $\overline{G'}$ in parallel with additional edges between matching vertices marking the points where $\sigma$ transitions from one graph to the other:
\begin{equation}
\left( u_a \to v_b \right) \in E_\sigma \Leftrightarrow \begin{cases}
\left( u \to v \right) \in \overline{E} & a = b = 1 \\
\left( u \to v \right) \in \overline{E'} & a = b = 2 \\
\left( \begin{array}{l} \exists \left( w \to u \right) \in \sigma \cap \overline{E}, \\ \left( v \to w' \right) \in \sigma \cap \overline{E'} . u = v \end{array} \right) & a = 1, b = 2 \\
\left( \begin{array}{l} \exists \left( w' \to u \right) \in \sigma \cap \overline{E'}, \\ \left( v \to w \right) \in \sigma \cap \overline{E} . u = v \end{array} \right) & a = 2, b = 1
\end{cases}
\end{equation}

graph states over $G_\sigma$ resemble a pair of graph states over $\overline{G}$ and $\overline{G'}$, contracted only over the vertices where $\sigma$ transitions from one graph to the other, with local effects on all other vertices. They aren't precisely the same thing, since edge interpretations must be first-order objects and the contractions may occur over some $\Gamma(v)$ which is neither first-order or first-order dual. However, by construction of $\overline{E}$ and $\overline{E'}$, we know that any transition from $\overline{E}$ to $\overline{E'}$ at vertex $v$ has $\left| c_{\Gamma(v)} \right| > 1$ since $v$ must have a predecessor in $\overline{E}$ and a successor in $\overline{E'}$. This is enough to be able to apply Lemma \ref{lemma:binary_channel_encoding_morphism} to obtain causal morphisms which actively signal to and from it, e.g. if $(u \to v) \in \sigma \cap \overline{E}$ and $(v \to w) \in \sigma \cap \overline{E'}$, it gives morphisms in $c_{\Delta(u \to v)^* < \Gamma(v)} \setminus c_{\Delta(u \to v)^* > \Gamma(v)}$ and $c_{\Gamma(v)^* < \Delta(v \to w)} \setminus c_{\Gamma(v)^* > \Delta(v \to w)}$. We can therefore apply the same cyclic feedback construction from the proof of Lemma \ref{lemma:acyclic_graph_morphisms_flat} to generate two pairs of graph states in $\mathbf{Gr}_{\overline{G}}^\Gamma$ and $\mathbf{Gr}_{\overline{G'}}^{\Gamma^*}$ whose inner products give distinct scalars - at least one of these inner products must be not $\id_I$.

$\Longleftarrow$: Now we suppose that $\left( V, \overline{E} \cup \overline{E'} \right)$ is acyclic. If we take any pair of graph states from $\mathbf{Gr}_{\overline{G}}^\Gamma$ and $\mathbf{Gr}_{\overline{G'}}^{\Gamma^*}$ and contract the components at each vertex into a single component, we obtain a valid graph state of $\mathbf{Gr}_{\left( V, \overline{E} \cup \overline{E'} \right)}^{\left\{ v \mapsto \mathbf{I} \right\}_{v \in V}}$ (this is a valid object by Lemma \ref{lemma:acyclic_graph_morphisms_flat} because of acyclicity). Corollary \ref{corollary:scalar_graph_types} tells us there is a unique scalar formed by this composition, which is $\id_I$.
\end{proof}

\begin{proposition}[Restatement of Proposition \ref{prop:dual_graph_type_as_union}]
\begin{equation*}
\left( \mathbf{Gr}_G^\Gamma \right)^* = \bigcup_{O \text{ topological sort of } G} \mathrm{perm}_V \left( \bigseq_{v \in O} \Gamma(v)^* \right)
\end{equation*}
\end{proposition}

\begin{proof}
This is immediate from the ordered graph type definition and Equations \ref{eq:union_intersection_de_morgan} and \ref{eq:seq_de_morgan}.
\end{proof}

\begin{proposition}[Restatement of Proposition \ref{prop:graph_types_preserve_union_intersection}]
\begin{align*}
\mathbf{Gr}_{G}^{\Gamma, v \mapsto \mathbf{A \cup B}} &= \mathbf{Gr}_{G}^{\Gamma, v \mapsto \mathbf{A}} \cup \mathbf{Gr}_{G}^{\Gamma, v \mapsto \mathbf{B}} \\
\mathbf{Gr}_{G}^{\Gamma, v \mapsto \mathbf{A \cap B}} &= \mathbf{Gr}_{G}^{\Gamma, v \mapsto \mathbf{A}} \cap \mathbf{Gr}_{G}^{\Gamma, v \mapsto \mathbf{B}}
\end{align*}
\end{proposition}

\begin{proof}
$\cup$: Using Equation \ref{eq:iu_monoidal_distribution}, the component type at $v$ from $\mathbf{Gr}_{G}^{\Gamma, v \mapsto \mathbf{A \cup B}}$ satisfies:
\begin{equation}
\mathrm{Comp}_G^{\Gamma, v \mapsto \mathbf{A \cup B};\Delta} = \mathbf{\left( A \cup B \right) \parr \left( \cdots \right)} = \mathbf{A \parr \left( \cdots \right) \cup B \parr \left( \cdots \right)}
\end{equation}
The $v$ component of any graph state of $\mathbf{Gr}_{G}^{\Gamma, v \mapsto \mathbf{A}}$ or $\mathbf{Gr}_{G}^{\Gamma, v \mapsto \mathbf{B}}$ is trivially valid for the component typing above, and conversely the $v$ component for any graph state of $\mathbf{Gr}_{G}^{\Gamma, v \mapsto \mathbf{A \cup B}}$ is an affine combination of valid components for $\mathbf{Gr}_{G}^{\Gamma, v \mapsto \mathbf{A}}$ and $\mathbf{Gr}_{G}^{\Gamma, v \mapsto \mathbf{B}}$. Hence any graph state of one type can always be broken down as an affine combination of graph states of the other.

$\cap$: Immediate from the ordered graph type definition and Equation \ref{eq:iu_monoidal_distribution}.
\end{proof}

\begin{proposition}[Restatement of Proposition \ref{prop:graph_substitution_in_graph_types}]
Let $G = (V, E)$ and $G' = (V', E')$ be two DAGs with disjoint vertices $V \cap V' = \emptyset$ and local interpretations $\Gamma$ and $\Gamma'$. Suppose that for some vertex $v$, $\Gamma(v) = \mathbf{Gr}_{G'}^{\Gamma'}$. Then $\mathbf{Gr}_{G}^\Gamma = \mathbf{Gr}_{G''}^{\Gamma \setminus \left\{ v \right\}, \Gamma'}$ where $G'' = \left(\left( V \setminus \left\{ v \right\} \right) \cup V', E''\right)$ with 
\begin{equation*}
(u \to w) \in E'' \Leftrightarrow \begin{cases}
u \to w \in E & u, w \in V \setminus \left\{ v \right\} \\
u \to v \in E & u \in V \setminus \left\{ v \right\}, w \in V' \\
v \to w \in E & u \in V', w \in V \setminus \left\{ v \right\} \\
u \to w \in E' & u, w \in V'
\end{cases}
\end{equation*}
\end{proposition}

\begin{proof}
The strategy to this proof is to exhibit a sequence of isomorphisms between graph types, where the isomorphisms themselves are purely structural (i.e. consisting of unitors, associators, permutations, etc.), so the equality of the underlying carrier objects implies the composite isomorphism is the identity and hence the types are equal.

We start with $\mathbf{Gr}_G^\Gamma = \mathbf{Gr}_{G^+}^\Gamma$ from Lemma \ref{lemma:transitive_closure_graph_type_without_pruning}. We can then use Lemma \ref{lemma:unit_elimination_from_graph_types} to insert additional vertices interpreted as $\mathbf{I}$ to collect predecessors and successors of $v$, meaning $v$ has precisely one predecessor $p$ and one successor $s$. Let's refer to this new graph type as $\mathbf{Gr}_{G_A}^{\Gamma_A}$.

We then do the same thing twice starting from $\mathbf{Gr}_{G''}^{\Gamma \setminus \left\{ v \right\}, \Gamma'}$ to give the $\mathbf{Gr}_{G_B}^{\Gamma_B}$ over $\left( V \setminus \left\{ v \right\} \right) \cup V' \cup \left\{ i, p, o, s \right\}$ with $p \to i$ and $o \to s$ where $p$ and $o$ are fan-in from predecessors of $v$ and $V'$, and $i$ and $s$ are fan-out into $V'$ and successors of $v$. $\mathbf{Gr}_{G_A}^{\Gamma_A}$ and $\mathbf{Gr}_{G_B}^{\Gamma_B}$ are still related by the same graph substitution of $\mathbf{I <} \mathbf{Gr}_{G'}^{\Gamma'} \mathbf{< I} \cong \Gamma(v)$.

Examining the component type at $v$ (assuming wlog. the constant edge interpretation $\Delta_{\mathbf{2}}$):
\begin{equation}
\begin{split}
\Gamma(v) \mathbf{\parr 2^* \parr 2} &\cong \mathbf{2^* <} \Gamma(v) \mathbf{< 2} \\
&= \mathbf{Gr}_{\left( \left\{ i \right\}, \emptyset \right)}^{i \mapsto \mathbf{2^*}} < \mathbf{Gr}_{G'}^{\Gamma'} < \mathbf{Gr}_{\left( \left\{ o \right\}, \emptyset \right)}^{o \mapsto \mathbf{2}} \\
&\cong \mathbf{Gr}_{\left( V' \cup \left\{ i, o \right\}, E' \cup \left\{ i \to v , v \to o \middle| v \in V \right\} \right)}^{\Gamma', i \mapsto \mathbf{2^*}, o \mapsto \mathbf{2}}
\end{split}
\end{equation}
This equates possible $v$ components of graph states of $\mathbf{Gr}_{G_A}^{\Gamma_A}$ with affine sums over the possible sub-graph states of $\mathbf{Gr}_{G_B}^{\Gamma_B}$ over $V' \cup \left\{ i, o \right\}$. The remaining components of the graph states of these two types are identical, so they generate the same graph type.
\end{proof}

\section{Proofs from Section \ref{sec:logic}}\label{sec:logic_proofs}

\begin{proposition}[Restatement of Proposition \ref{prop:pomset_extension}]
The logic of causal proof-nets is a conservative extension of pomset logic. Specifically, given a formula in the fragment $F, G ::= A | A^* | F \otimes G | F < G | F \parr G$, there exists a causal proof-net for $F$ iff there exists a pomset proof-net for $F$.
\end{proposition}

\begin{proof}
We will compare causal proof-nets to pomset proof-nets, i.e. bi-coloured graphs with no alternating elementary circuits \cite{Retore1997}. There is an obvious one-to-one correspondence between proof-structures for unit- and first-order-free formulae, operating link-wise over the proof-structure. The proof-net conditions both reduce down to finding cycles through the links with some connectivity constraints within the links themselves. We can show inductively that each link kind has the same connectivity between its premises and conclusions induced by the alternating colour condition for pomset or the up-down switchings for causal proof-nets, and therefore we can transform between any cycle in an up-down switching and an alternating elementary circuit.

For example, the up-down switchings of a $\parr$ link permit a path in either direction between one premise and the conclusion (the choice of switching just determines the direction), but never a path between the two premises. The corresponding bi-coloured link has a central vertex with undirected red edges to the premises and an undirected blue edge to the conclusion. The alternating-colour condition therefore prevents the path between the premises, but permits the undirected path between each premise and the conclusion.

Similarly, $<$ links only forbid connections between the premises in one direction, and $\otimes$ and axiom links permit any connectivity between their components.
\end{proof}

\begin{lemma}[Restatement of Lemma \ref{lemma:switching_decomposition_of_causal_type}]
Given a balanced formula $F$ and an interpretation $\Phi$ which FO-respects $F$, let $(V, E)$ be the syntax tree of $F$ (in negation normal form) and let $\Gamma_{\Phi,\mathbf{E}}^* : V \to \ob{\caus{\catc}}$ be the function:
\begin{equation*}
\Gamma_{\Phi,\mathbf{E}}^* (v) = \begin{cases}
\Phi(F)^* \mathbf{\parr E} & v \text{ is the only vertex } F ::= A | A^* | A^1 | \left( A^1 \right)^* | I \\
\Phi(a)^* & v \text{ is a leaf } a ::= A | A^* | A^1 | \left(A^1\right)^* | I \\
\mathbf{E} & v \text{ is the root } $F$ \\
I & \text{otherwise}
\end{cases}
\end{equation*}
Then $\Phi(F)^* \mathbf{\parr E} \cong \bigcup_{s \in \mathcal{S}_F} \mathbf{Gr}_{G_s}^{\mathbf{\Gamma}_{\Phi,\mathbf{E}}^*}$. In words with $\mathbf{E} = \mathbf{I}$, the state space of $\Phi(F)^*$ coincides with the space generated by graph states over the switchings of the syntax tree of $F$.
\end{lemma}

\begin{proof}
We proceed inductively over the syntax tree of $F$ (in negation normal form).

If $F ::= A | A^* | A^1 | \left( A^1 \right)^* | I$, the syntax tree is a single node with a single trivial switching graph.
\begin{equation}
\Phi(F)^* \mathbf{\parr E} = \Gamma_{\Phi, \mathbf{E}}^*(v) = \mathbf{Gr}_{G_s}^{\Gamma_{\Phi, \mathbf{E}}^*}
\end{equation}

If $F = G \otimes H$:
\begin{equation}
\begin{split}
\Phi(F)^* \mathbf{\parr E} =& \left( \Phi(G) \otimes \Phi(H) \right)^* \mathbf{\parr E} \\
=& \left( \Phi(G)^* \mathbf{\parr} \Phi(H)^* \right) \mathbf{\parr E} \\
=& \left( \left( \Phi(G)^* \mathbf{<} \Phi(H)^* \right) \mathbf{\cup} \left( \Phi(H)^* \mathbf{<} \Phi(G)^* \right) \right) \mathbf{\parr E} \\
=& \left( \Phi(G)^* \mathbf{<} \Phi(H)^* \mathbf{< E} \right) \mathbf{\cup} \left( \Phi(H)^* \mathbf{<} \Phi(G)^* \mathbf{< E} \right) \\ & \mathbf{\cup} \left( \mathbf{E <} \Phi(G)^* \mathbf{<} \Phi(H)^* \right) \mathbf{\cup} \left( \mathbf{E <} \Phi(H)^* \mathbf{<} \Phi(G)^* \right) \\
=& \mathbf{Gr}_{dr}^{\Gamma_\otimes} \cup \mathbf{Gr}_{dl}^{\Gamma_\otimes} \cup \mathbf{Gr}_{ur}^{\Gamma_\otimes} \cup \mathbf{Gr}_{ul}^{\Gamma_\otimes}
\end{split}
\end{equation}
where the graphs $dr, dl, ur, ul$ are given by the corresponding switching options for the $\otimes$ link, and $\Gamma_\otimes$ maps the premise vertices to $\Phi(G)^*$ and $\Phi(H)^*$ and the conclusion vertex to $\mathbf{E}$.

Looking at $\mathbf{Gr}_{dr}^{\Gamma_\otimes}$ as an example, the component typings for the graph states (assuming wlog. the constant edge interpretation $\Delta_{\mathbf{2}}$) are $\Phi(G)^* \mathbf{\parr 2}$, $\Phi(H)^* \mathbf{\parr \left( 2^* \parr 2 \right)}$, and $\mathbf{E \parr 2^*}$. The inductive hypothesis identifies the premise components with affine combinations of graph states over the switching graphs of $G$ and $H$. Identifying these within complete graph states gives $\mathbf{Gr}_{dr}^{\Gamma_\otimes} = \bigcup_{s \in \mathcal{S}_F, s(F) = dr} \mathbf{Gr}_{G_s}^{\Gamma_{\Phi,\mathbf{E}}^*}$ (the leaf and otherwise cases of $\Gamma_{\Phi,\mathbf{E}}^*$ handle the fact that we have moved the $\mathbf{2}$ parts of the premises from the local interpretation of the inductive graph type to edge interpretations in the larger switching graph). Repeating this for each of $dl, ur, ul$ generates all switching graphs over $F$.

The cases for $<$ and $\parr$ are similar with the following decompositions:
\begin{align}
(\Phi(G) < \Phi(H))^* \mathbf{\parr E} =& (\Phi(G)^* < \Phi(H)^* \mathbf{< E}) \\ & \cup (\mathbf{E <} \Phi(G)^* < \Phi(H)^*) \nonumber \\
=& \mathbf{Gr}_{ds}^{\Gamma_<} \cup \mathbf{Gr}_{us}^{\Gamma_<} \\
(\Phi(G) \parr \Phi(H))^* \mathbf{\parr E} =& ((\Phi(G)^* \otimes \Phi(H)^*) \mathbf{< E}) \\ & \cup (\mathbf{E <} (\Phi(G)^* \otimes \Phi(H)^*)) \nonumber \\
=& \mathbf{Gr}_{dp}^{\Gamma_\parr} \cup \mathbf{Gr}_{up}^{\Gamma_\parr}
\end{align}
\end{proof}

\begin{theorem}[Restatement of Theorem \ref{thrm:causal_logic_characterisation}]
Given a balanced formula $F$ and an interpretation $\Phi$ which FO-respects $F$, $\Vdash_{\catc}^\Phi F$ iff $P_F$ is a causal proof-net.
\end{theorem}

\begin{proof}
Recall the definition of $\Vdash_{\catc}^\Phi$ as stating whether or not $\epsilon_F^\Phi$ is causal $\Phi(F)^* \to \mathbf{I}$, i.e. for all causal states $\rho \in c_{\Phi(F)^*}$, $\rho \fatsemi \epsilon_F^\Phi = \id_I$. We decompose $\epsilon_F^\Phi$ into a permutation followed by the individual caps in parallel $\epsilon_F^\Phi = \sigma \fatsemi \bigotimes_{A \in F} \epsilon_{\Phi(A)}$ with each $\epsilon_{\Phi(A)} : \Phi(A^*)^* \otimes \Phi(A)^* \to \mathbf{I}$.

We first claim that $\Vdash_{\catc}^\Phi$ iff $c_{\bigotimes_{A \in F} \Phi(A^*) \parr \Phi(A)} \subseteq \left( c_{\Phi(F)^*} \right)^*$ up to the appropriate permutation of atomic wires. $\Longleftarrow$ is immediate from $\epsilon_{\Phi(A)} \in c_{\Phi(A^*) \parr \Phi(A)}$. For $\Longrightarrow$, we note that any element of $c_{\Phi(A^*) \parr \Phi(A)}$ can be written in the form $\epsilon_{\Phi(A)}^* \fatsemi \left( \id_{\Phi(A^*)} \otimes f \right)$ for some $f : \Phi(A) \to \Phi(A)$ by $*$-autonomy. Given any state $\rho \in c_{\Phi(F)^*}$, we have $\rho \fatsemi \left( \id \otimes f^* \otimes \id \right) \in c_{\Phi(F)^*}$ since every operator in $\Phi(F)^*$ is functorial. We can apply this for each atom $A$ in $F$ to reduce $c_{\bigotimes_{A \in F} \Phi(A^*) \parr \Phi(A)} \subseteq \left( c_{\Phi(F)^*} \right)^*$ to just checking $\epsilon_F^\Phi \in c_{\Phi(F)}$, i.e. $\Vdash_\catc^\Phi$.

Next, Proposition \ref{prop:first_order_equations} gives $\Phi(A^*) \parr \Phi(A) = \Phi(A^*) < \Phi(A)$ iff $\Phi(A)$ is first-order. Since $\Phi$ FO-respects $F$, this happens iff $A$ appears in $F$ as a first-order atom $A^1$. By Equation \ref{eq:seq_as_graph_type}, $\Phi(A^*) < \Phi(A)$ is a graph type matching the only switching option available to first-order axiom links. Similarly, Equation \ref{eq:parr_union_decomposition} gives the decomposition for generic atoms as $\Phi(A^*) \parr \Phi(A) = \left( \Phi(A^*) < \Phi(A) \right) \cup \left( \Phi(A^*) > \Phi(A) \right)$ which match the two switching options for regular axiom links.

Applying Lemma \ref{lemma:switching_decomposition_of_causal_type}, $\left( \Phi(F)^* \right)^*$ is the intersection of the dual graph types for each switching graph over the syntax tree of $F$. Our goal $c_{\bigotimes_{A \in F} \Phi(A^*) \parr \Phi(A)} \subseteq \left( c_{\Phi(F)^*} \right)^*$ is now inclusion of a union within an intersection, meaning we need to quantify over both the switchings of axiom links and the switchings of the syntax tree, i.e. over all switchings of the canonical cut-free proof-structure. For each such switching $s$, we ask the graph type given by the axiom switching (extended to all vertices of the syntax tree with $\mathbf{I}$s via Lemma \ref{lemma:unit_elimination_from_graph_types}) is contained in the dual graph type over the syntax tree $\left( \mathbf{Gr}_{G_s}^{\Gamma_{\Phi, \mathbf{I}}^*} \right)^*$. Finally, Lemma \ref{lemma:transitive_closure_dual_graph_type} equates this with checking acyclicity of the combined graph, i.e. the full switching graph over the entire proof-structure (the fact that this takes the standard forms $\overline{G}$ of graphs does not matter here, since the first-order axiom links will already be directed in the appropriate direction, so any cycle can still be used to disprove causality).
\end{proof}

\begin{corollary}[Restatement of Corollary \ref{corollary:complexity_of_causal_consistency}]
Causal consistency of a balanced formula is $\mathrm{coNP}$-complete.
\end{corollary}

\begin{proof}
It is $\mathrm{coNP}$-complete to verify that a proof-structure is a proof-net for pomset, i.e. whether a balanced formula is satisfiable in pomset logic. Since unit- and first-order free causal proof-nets coincide with pomset proof-nets, causal consistency must be at least $\mathrm{coNP}$-hard. It lies in $\mathrm{coNP}$ itself since giving an up-down switching and a cycle gives an efficiently verifiable refutation.
\end{proof}

\begin{proposition}[Restatement of Proposition \ref{prop:equivalent_fragments}]
For any causal proof-structure $P$, the following are equivalent:
\begin{itemize}
\item $P$ is a causal proof-net;
\item $\mathrm{pom}\left(P\right)$ is a causal proof-net;
\item $\mathrm{fo}\left(P\right)$ is a causal proof-net.
\end{itemize}
\end{proposition}

\begin{proof}
Much like Proposition \ref{prop:pomset_extension}, this amounts to checking that connectivity between the end-points is preserved through each of the selected rewrites when locally quantifying over the up-down switchings, and that the replacements don't contain internal cycles. In each case, this is straightforward.
\end{proof}

\section{Proofs from Section \ref{sec:discussion}}

\begin{lemma}[Restatement of Lemma \ref{lemma:unique_causal_extranatural}]
Given a balanced formula $F$, there is at most one extranatural transformation $\eta_F : \mathcal{I} \to \mathcal{F}_F$ which is extranatural across the pairs of indices given by matching atoms in $F$. When it exists, the term attributed to an interpretation $\Phi : \Var \to \ob{\caus{\catc}}$ is precisely $\left( \epsilon_F^\Phi \right)^*$, i.e. a selection of compact cups connecting matching atoms.
\end{lemma}

\begin{proof}
Consider an arbitrary extranatural transformation $\eta_F$ of this kind and pick an arbitrary component of it. For each object $\mathbf{A}$ in the component index, it embeds into some first-order object $\mathbf{A^1}$ by flatness (the causal effects of $\mathbf{A}$ includes the unique causal effect of $\mathbf{A^1}$). This embedding is witnessed by the identity from $\catc$, $\id_A : \mathbf{A} \to \mathbf{A^1}$. Extranaturality applies for all causal morphisms:
\begin{equation}
\input{./causal_extranatural0.tikz} = \input{./causal_extranatural1.tikz}
\end{equation}
Doing this for each variable index, we find that all components of $\eta_F$ are fixed by just the first-order ones.

Let $\left\{ \rho_i^A \right\}_i$ be the causal basis for $A$ from \ref{apc:basis}. Each $\rho_i^A$ is a causal state of $\mathbf{A^1}$ up to a scalar, so we can apply extranaturality again.
\begin{equation}
\input{./causal_extranatural_basis0.tikz} = \input{./causal_extranatural_basis1.tikz}
\end{equation}
Again, repeating this for each variable index reduces it to the basis state and the scalar term of $\eta_F$ which (up to unitors) must be $\id_I$, the only causal morphism $\mathbf{I} \to \mathbf{I}$. Doing this for each choice of basis states allows us to completely identify our original arbitrary component of $\eta_F$ by local tomography, seeing that it acts identically to just a compact cup for each index.
\end{proof}

\begin{proposition}[Restatement of Proposition \ref{prop:equivalent_defs_of_causal_consistency}]
For any balanced formula $F$, $\Vdash_\catc F$ iff for all interpretations $\Phi$ which FO-respect $F$ we have $\Vdash_\catc^\Phi F$, i.e. iff $P_F$ is a causal proof-net.
\end{proposition}

\begin{proof}
By Lemma \ref{lemma:unique_causal_extranatural}, $\Vdash_\catc F$ iff we can build an extranatural transformation where each component (specified by an interpretation function $\Phi$) is given by a cup connecting the pairs of matching variables from $F$, i.e. each such term from $\catc$ is causal $\mathbf{I} \to \Phi(F)$. This exactly describes the transpose of $\epsilon_F^\Phi$ and so occurs iff $\epsilon_F^\Phi : \Phi(F)^* \to \mathbf{I}$, i.e. iff $\Vdash_\catc^\Phi F$. Theorem \ref{thrm:causal_logic_characterisation} then connects this to causal proof-nets.
\end{proof}

\end{document}